\documentclass[a4paper,11pt,fleqn]{article}

\usepackage[ansinew]{inputenc}
\usepackage{hyperref}
\hypersetup{colorlinks, linkcolor=blue, citecolor=blue, urlcolor=blue}
\usepackage[mathscr]{eucal}
\usepackage{amsmath,amssymb,amsthm}

\allowdisplaybreaks
\newcommand{\p}{\partial}
\newcommand{\const}{\mathop{\rm const}\nolimits}
\newcommand{\EqOrd}{r}

\newcommand{\rsemioplus}{\mathbin{\mbox{$\lefteqn{\hspace{.67ex}\rule{.4pt}{1.2ex}}{\ni}$}}}

\newcommand{\spanindex}{{\mbox{\tiny$\langle\,\rangle$}}}

\newcommand{\ve}{\varepsilon}
\newcommand{\DD}{\mathrm D}
\newcommand{\ZZ}{\mathcal Z}
\newcommand{\PP}{\mathcal P}
\newcommand{\RR}{\mathcal R}
\newcommand{\DDD}{\mathcal D}

\newcommand{\todo}[1][\null]{\ensuremath{\clubsuit}}

\newcommand{\noprint}[1]{}
\newcommand{\checked}[1][\null]{\ensuremath{\boldsymbol{\surd}}}

\newtheorem{theorem}{Theorem}
\newtheorem{lemma}[theorem]{Lemma}
\newtheorem{corollary}[theorem]{Corollary}
\newtheorem{proposition}[theorem]{Proposition}
{\theoremstyle{definition}
\newtheorem{remark}[theorem]{Remark}
\newtheorem{definition}[theorem]{Definition}
}

\newcounter{tbn}
\newcounter{tbnn}[tbn]
\renewcommand{\thetbnn}{{\rm\arabic{tbn}\alph{tbnn}}}
\renewcommand{\thetbn}{{\rm\arabic{tbn}}}
\newcommand{\cn}[1]{\refstepcounter{tbn}\thetbn\label{#1}}
\newcommand{\ca}[1]{\refstepcounter{tbn}\refstepcounter{tbnn}\thetbnn\label{#1}}
\newcommand{\cb}[1]{\refstepcounter{tbnn}\thetbnn\label{#1}}

\begin{document}

\par\noindent {\LARGE\bf
Generalization of the algebraic method\\ of group classification with application\\ to nonlinear wave and elliptic equations
\par}

\vspace{4mm}\par\noindent
{\large Olena~O.~Vaneeva$^\dag$, Alexander Bihlo$^\ddag$ and Roman O. Popovych$^{\dag\S}$
\par}

{\it
\vspace{5mm}\par\noindent
$^\dag$Institute of Mathematics of NAS of Ukraine, 3 Tereshchenkivska Str., 01024 Kyiv, Ukraine

\vspace{2mm}\par\noindent
$^\ddag$Department of Mathematics and Statistics, Memorial University of Newfoundland,\\
$\phantom{^\ddag}$\,St.\ John's (NL) A1C 5S7, Canada

\vspace{2mm}\par\noindent
$^\S$Fakult\"at f\"ur Mathematik, Universit\"at Wien, Oskar-Morgenstern-Platz 1, A-1090 Wien, Austria

}

\vspace{2mm}

\noindent
E-mails: vaneeva@imath.kiev.ua, abihlo@mun.ca, rop@imath.kiev.ua

\vspace{6mm}\par\noindent\hspace*{8mm}\parbox{140mm}{\small
Enhancing and essentially generalizing previous results 
on a class of (1+1)-dimensional nonlinear wave and elliptic equations, 
we apply several new techniques to classify admissible point transformations 
within this class up to the equivalence generated by its equivalence group.
This gives an exhaustive description of its equivalence groupoid.
After extending the algebraic method of group classification 
to non-normalized classes of differential equations, 
we solve the complete group classification problem for the class under study
up to both usual and general point equivalences.
The solution includes the complete preliminary group classification of the class 
and the construction of singular Lie-symmetry extensions,
which are not related to subalgebras of the equivalence algebra.
The complete preliminary group classification is based on classifying 
appropriate subalgebras of the entire infinite-dimensional equivalence algebra 
whose projections are qualified as maximal extensions of the kernel invariance algebra. 
The results obtained can be used to construct exact solutions of nonlinear wave and elliptic equations.
\par}\vspace{4mm}

\noprint{
MSC: 35B06 (Primary) 35A30, 35L70, 35J60 (Secondary)
35-XX   Partial differential equations
    35A30   Geometric theory, characteristics, transformations [See also 58J70, 58J72]
    35B06   Symmetries, invariants, etc.
  35Jxx	 Elliptic equations and systems [See also 58J10, 58J20]
	35J60  	Nonlinear elliptic equations
	35J61  	Semilinear elliptic equations
  35Lxx	 Hyperbolic equations and systems [See also 58J45]
	35L70  	Nonlinear second-order hyperbolic equations
    35L71  	Semilinear second-order hyperbolic equations

Keywords: group classification of differential equations, 
          nonlinear wave equations, nonlinear elliptic equations, 
          Lie symmetry, equivalence group, equivalence groupoid, equivalence algebra,
          method of furcate splitting
          algebraic method of group classification
}

\section{Introduction}

Group classification is concerned with finding an exhaustive list of inequivalent equations 
from a class of differential equations containing one or more arbitrary elements. 
It was originally motivated from theoretical physics, 
where traditionally the equations admitting the maximal number of symmetries among equations 
from a given class yield the most promising model describing real-world phenomena.
Mathematically, group classification problems for classes of differential equations 
have been intensively investigated, 
starting with Sophus Lie's classifications of second-order ordinary differential equations~\cite{lie1891A} 
and of second-order linear partial differential equations with two independent variables~\cite{lie1881a}. 
Recently, a number of novel techniques for group classification have been introduced, 
which include various flavors of the algebraic method~\cite{bihl2012b,card2011a,kuru2018a,opan2017a,popo2010a}
and variations of the advanced modification of the direct method 
called the method of furcate splitting~\cite{bihl2019a,opan2020b,niki2001a}.
The algebraic method of group classification has proven so far to be the most powerful 
since it has been efficiently applied to classes of differential equations 
with arbitrary elements that are functions of several arguments. 

Among the classes considered in the literature on group classification, 
the most prominent are classes of (1+1)-dimensional evolution equations, 
see e.g.
\cite{akha1991a,CRC1994V1,basa2001a,basa2001b,bihlo2016a,blum1989A,doro1982a,card2011a,gaze1992a,gung2004a,
huan2009a,huan2012c,ivan2010a,maga1993a,opan2017a,opan2019b,opan2020b,ovsi1982A,popo2004b,vane2007a,vane2015d,vane2009a,vane2012a} 
and references therein. 
It is thus also no coincidence that in the field of invariant discretization, 
which is concerned with deriving numerical schemes for differential equations 
possessing the same symmetries as the original, undiscretized equation, 
mostly evolutionary equations have been considered in the past, see e.g.~\cite{bihlo2017a}. 
What distinguishes evolutionary equations from the symmetry-perspective is 
the special role of the time variable, which is similar to the role of a parameter. 
Thus,  the time component of any point or contact transformation 
between evolution equations only depends on the time variable~\cite{king1998a,maga1993a}.
This considerably simplifies the classification procedure.

Wave and elliptic equations play an important role in physics and in the mathematical sciences 
since wave equations model the transport of quantities at finite speeds 
whereas elliptic equations describe stationary processes. 
From the symmetry perspective, such equations are challenging 
since all the independent variables in them enter on equal footing. 
Lie symmetries of wave and elliptic equations with two independent variables have also been studied extensively, 
see e.g. \cite{CRC1994V1,bihl2012b,blum2010A,blum2006a,blum1987a,blum1989A,gand2004a,huan2007a,huan2012a,ibra1991a,lahn2005a,lahn2006a,lahn2007a} 
and references therein. 
Note that the first investigations of such equations within the framework of group analysis of differential equations, 
which are relevant for the subject of the present paper, 
were carried out by Sophus Lie 
in the course of his classification of second-order linear partial differential equations 
with two independent variables~\cite{lie1881a} 
and in the course of his study of contact transformations 
between nonlinear Klein--Gordon equations of the form $d^2z/dx\,dy=F(z)$~\cite{lie1881b}. 
Exact solutions constructed for nonlinear wave and elliptic equations using group-theoretical and related methods 
are collected, e.g., in \cite{CRC1994V1,poly2012A,poly2019b,poly2019f,zhur2020a}.

In the present paper, we exhaustively solve the group classification problem
for the class~$\mathcal W$ of nonlinear wave and elliptic equations of the form
\begin{equation}\label{eq:GenWaveEqs}
u_{tt}=f(x,u)u_{xx}+g(x,u).
\end{equation}
We need to explicitly impose two auxiliary inequalities on the arbitrary-element tuple $\theta=(f,g)$
in order to precisely describe the class~$\mathcal W$,
which is also referred to as the class~\eqref{eq:GenWaveEqs} in the paper.
The auxiliary inequality $f\ne0$ is natural 
since equations of the form~\eqref{eq:GenWaveEqs} with $f=0$ are not true partial differential equations.%
\footnote{
Since we work within the local framework, auxiliary inequalities on arbitrary elements are interpreted
as satisfied for all values of arguments of arbitrary elements on the relevant domain.
}
We denote by~$\mathcal W_{\rm gen}$ the superclass of equations of the form~\eqref{eq:GenWaveEqs} with $f\ne0$.
In order to guarantee nonlinearity of equations from the class~$\mathcal W$,
the definition of this class should also include the auxiliary inequality
\[(f_u,g_{uu})\ne(0,0).\]
The subclass~$\mathcal W_{\rm lin}$ of linear equations in~$\mathcal W_{\rm gen}$
is the complement of~$\mathcal W$ in~$\mathcal W_{\rm gen}$,
$\mathcal W=\mathcal W_{\rm gen}\setminus\mathcal W_{\rm lin}$.
The reason why we separate nonlinear and linear equations of the form~\eqref{eq:GenWaveEqs} is 
that they are not mixed by point transformations (see Remark~\ref{rem:OnInequivOfLinAndNonlinCasesOfGenWaveEqs} below)
and have quite different Lie-symmetry properties.
Although linear wave and elliptic equations with two independent variables were already extensively investigated
within the framework of classical symmetry analysis, see, e.g., \cite{blum1987a,blum1989A,lie1881a,ovsj1962A,ovsi1982A}, 
we discuss specific transformational and symmetry properties of equations from~$\mathcal W_{\rm lin}$ 
in Remark~\ref{rem:GenWaveEqsLinCase} below, relating them to equations from~$\mathcal W$.
The sign of~$f$ is not too essential in the course of group classification of the class~$\mathcal W$.
In fact, we classify both the subclass of hyperbolic equations for which $f>0$ and the subclass of elliptic equations with $f<0$.
Hyperbolic and elliptic equations are also not mixed by point transformations.
Note that the consideration is local and all values are real throughout the paper 
although the transition to the complex case needs only minor modifications.

\looseness=-1
Following~\cite{akha1991a,ibra1991a}, 
a so-called partial preliminary group classification problem~\cite{bihl2012b,card2011a} 
for the class~$\mathcal W$ has been considered in~\cite{song2009a}. 
Specifically, the authors selected a six-dimensional subalgebra~$\mathfrak g_6$  
of the infinite-dimensional equivalence algebra~$\mathfrak g^\sim$ of the class~$\mathcal W$ 
and tried to only classify one-dimensional subalgebras of the subalgebra~$\mathfrak g_6$ 
up to the equivalence generated by the corresponding six-dimensional subgroup~$G_6$ of 
the infinite-dimensional equivalence (pseudo)group~$G^\sim$ of the class~\eqref{eq:GenWaveEqs}. 
The $G_6$-equivalence is much weaker than the $G^\sim$-equivalence. 
This is why the classification in~\cite{song2009a} led to 
an excessively large list of 24 $G_6$-equivalent simplest classification cases 
of one-dimensional Lie-symmetry extensions
most of which are $G^\sim$-equivalent to each other and, up to $G^\sim$-equivalence, 
fit into the first four cases of Table~\ref{tab:GenWaveEqsExtensions} below. 
Moreover, a number of classification cases were missed 
even within the posed partial preliminary group classification problem. 

We enhance and substantially generalize the results of~\cite{song2009a}. 
The class~$\mathcal W$ is neither normalized nor semi-normalized 
in any sense (the usual, the generalized or the extended ones). 
It cannot be partitioned into normalized or semi-normalized subclasses 
that are not related by point transformations. 
There is no mapping of it by families of point transformations to a class 
with better transformational properties. 
This is why Lie symmetries of equations from the class~$\mathcal W$ 
cannot be exhaustively classified by the existing versions of 
the algebraic method of group classification, 
which are explicitly \cite{bihl2012b,bihlo2016a,card2011a,kuru2018a,opan2017a,popo2006b,popo2010a}
or implicitly \cite{basa2001a,basa2001b,gagn1993a,gaze1992a,gung2004a,huan2009a,huan2012c,lahn2005a,lahn2006a,lahn2007a,maga1993a}
based on certain normalization properties of classified classes.
(Note that most of the above papers are devoted to group classifications 
of various classes of single (1+1)-dimensional evolution equations.)
On the other hand, the class~$\mathcal W$ is not convenient 
for considering within the framework of the direct method of group classification 
\cite{akha1991a,CRC1994V1,blum1989A,doro1982a,ovsj1962A,ovsi1982A}, 
including its advanced versions like the method of furcate splitting suggested in~\cite{niki2001a}. 
The last method is especially efficient for classes of differential equations 
with arbitrary elements depending on single arguments 
\cite{huan2007a,huan2012a,ivan2010a,opan2020b,popo2004b,vane2015d,vane2012a},
although it has also been applied to classes whose arbitrary elements depend on two arguments 
\cite{bihl2019a,niki2001a}.
Various specific algebraic techniques were suggested for group classification of classes such that 
sets of certain objects related to Lie symmetries of equations from these classes can be endowed 
with Lie-algebra structures \cite{bihl2019a,niki2017a,popo2001b} 
but this is not applicable for the class~$\mathcal W$. 

This is why to efficiently solve the complete group classification problem 
for the class~$\mathcal W$,
we develop a new version of the algebraic method of group classification 
for non-normalized classes of (systems of) differential equations, which is based 
on classifying admissible transformations of the class under study
up to their equivalence generated by the equivalence group of this class. 
We revisit the general framework of the classification of admissible transformations 
via modifying its basic notion of equivalent admissible transformations 
and introducing the notion of generating sets for equivalence groupoids. 
Several new techniques for classifying admissible transformations 
of non-normalized classes are also suggested. 
More specifically, we show that the method of furcate splitting 
and the algebraic method for computing 
the complete point or contact symmetry groups of single systems of differential equations 
\cite{hydo1998a,hydo2000b,hydo2000A} 
and the complete equivalence groups of classes of such systems~\cite{bihl2015a}
(including discrete symmetry and equivalence transformations) 
can be extended to admissible transformations. 
We also use the unexpected opportunity of describing admissible transformations 
via establishing a functor between the equivalence groupoids of classes 
that are not related by families of point transformations. 
Revisiting the algebraic method of group classification, 
we introduce the notions of regular and singular Lie-symmetry extensions 
for a class of differential equations, $\mathcal L|_{\mathcal S}$. 
Regular Lie-symmetry extensions are associated with subalgebras 
of the equivalence algebra of the class~$\mathcal L|_{\mathcal S}$. 
They are the extensions that can be constructed by the algebraic method 
in the course of the complete preliminary group classification~\cite{bihl2012b,card2011a} 
of~$\mathcal L|_{\mathcal S}$. 
Singular Lie-symmetry extensions can involve only systems from~$\mathcal L|_{\mathcal S}$ 
being sources of admissible transformations of~$\mathcal L|_{\mathcal S}$ 
that are not generated by equivalence transformations of~$\mathcal L|_{\mathcal S}$.
As a result, the group classification problem for the class~$\mathcal W$ 
that originated the above studies 
turns into a proof-of-concept example for the new methods designed for its solution. 

The further organization of the present paper is as follows. 
In the subsequent Section~\ref{sec:GenWaveEqsTheory}, 
we briefly present a new view on the classification of admissible transformations 
and use it for revisiting the algebraic method of group classification. 
The determining equations for admissible transformations of the class~$\mathcal W$, 
its equivalence group~$G^\sim$ jointly with its equivalence algebra~$\mathfrak g^\sim$ 
and the determining equations for Lie symmetries of equations from this class 
are computed in Sections~\ref{sec:PreliminaryStudyOfAdmTrans}, \ref{sec:EquivGroup}
and~\ref{sec:DetEqsForLieSymsIbragrimovClass}, respectively. 
The results on the complete group classification of the class~$\mathcal W$ up to $G^\sim$-equivalence
and on the description of the equivalence groupoid~$\mathcal G^\sim$ of~$\mathcal W$
in terms of a generating set of its admissible transformations 
are collected in Section~\ref{sec:GenWaveEqsResultsOfClassifications} 
for convenience of further references.  
The obtained description of~$\mathcal G^\sim$ allows us to find 
a complete set of additional equivalence transformations~\cite{popo2004b} between 
listed $G^\sim$-equivalent Lie-symmetry extensions, 
which leads to the complete group classification of the class~$\mathcal W$ 
up to $\mathcal G^\sim$-equivalence as well. 
A generating set of the equivalence groupoid~$\mathcal G^\sim$ and 
an exhaustive list of $G^\sim$-inequivalent singular Lie-symmetry extensions are computed 
in Section~\ref{sec:GenWaveEqsEquivGroupoid} using the variety of techniques mentioned above. 
In Section~\ref{sec:ClassificationSubalgebrasGenWaveEqs}, 
we classify presumptively appropriate subalgebras of the entire infinite-dimensional equivalence algebra~$\mathfrak g^\sim$ 
whose projections may be qualified as maximal extensions of the kernel invariance algebra~$\mathfrak g^\cap$ 
of equations from the class~$\mathcal W$.
The complete preliminary group classification of the class~$\mathcal W$ 
is finished in Section~\ref{sec:GenWaveEqsRegularLieSymmetryExtensions}. 
It gives an exhaustive list of $G^\sim$-inequivalent regular Lie-symmetry extensions, 
which is in fact relevant for the $\mathcal G^\sim$-equivalence as well. 
In Section~\ref{sec:ConclusionGenWaveEqs} we analyze our approach to 
simultaneously classifying admissible transformations and Lie symmetries 
for equations from the class~$\mathcal W$.  
We also compare the obtained list of Lie-symmetry extensions with similar lists 
existing in the literature for related classes 
of (1+1)-dimensional nonlinear wave and elliptic equations.

\section[Classification of admissible transformations and group classification problem]
{Classification of admissible transformations\\ and group classification problem}\label{sec:GenWaveEqsTheory}

In the present paper, we use various concepts of the groupoid theory more intensively 
than in previous papers on admissible transformations in classes of differential equations 
and their group classification by the algebraic method, 
see e.g. \cite{basa2001a,bihl2012b,bihlo2016a,card2011a,gagn1993a,gaze1992a,kuru2018a,opan2017a,popo2006b,popo2010a}.
This is why below we relate existing notions and results from group analysis of differential equations 
to the groupoid-relevant terminology 
and then present procedures for classifying admissible transformations and Lie symmetries 
for non-normalized classes of differential equations.
The notation in this section differs from other parts of the present paper.

Consider a class $\mathcal L|_{\mathcal S}=\{\mathcal L_\theta\mid\theta\in\mathcal S\}$
of systems of differential equations~$\mathcal L_\theta$
for unknown functions $u=(u^1,\ldots,u^m)$ of independent variables $x=(x_1,\ldots,x_n)$
with the arbitrary-element tuple~$\theta=(\theta^1,\dots,\theta^k)$ running through a set~$\mathcal S$.
Here $\mathcal L_\theta$ denotes a system of differential equations of the form $L(x,u_{(\EqOrd)},\theta(x,u_{(\EqOrd)}))=0$
with a fixed tuple~$L$ of $\EqOrd$th order differential functions in~$u$ parameterized by~$\theta$.
We use the short-hand notation~$u_{(\EqOrd)}$ for the tuple of derivatives of~$u$ with respect to $x$ up to order~$\EqOrd$,
including~$u$ as the zeroth order derivatives.
The set~$\mathcal S$ is the solution set of an auxiliary system
of differential equations and inequalities in~$\theta$, 
where $\EqOrd$th order jet variables $(x,u_{(\EqOrd)})$ play the role of independent variables,
$S(x,u_{(\EqOrd)},\theta_{(q)})=0$ and, e.g., $\Sigma(x,u_{(\EqOrd)},\theta_{(q)})\ne0$
with the tuple $\theta_{(q)}$ constituted by the derivatives of~$\theta$ up to order $q$ with respect to $(x,u_{(\EqOrd)})$.
Up to the gauge equivalence of systems from~$\mathcal L|_{\mathcal S}$~\cite{popo2010a}, 
which is usually trivial, the correspondence $\theta\mapsto\mathcal L_\theta$
between~$\mathcal S$ and~$\mathcal L|_{\mathcal S}$ is bijective. 

The \emph{equivalence groupoid} $\mathcal G^\sim$ of the class~$\mathcal L|_{\mathcal S}$ 
is the small category with~$\mathcal L|_{\mathcal S}$ or, equivalently, with $\mathcal S$ 
as the set of objects and with the set of point transformations of~$(x,u)$, 
i.e., of (local) diffeomorphisms in the space with the coordinates~$(x,u)$, 
between pairs of systems from~$\mathcal L|_{\mathcal S}$ as the set of arrows. 
Specifically, 
\[
\mathcal G^\sim=\big\{\mathcal T=(\theta,\Phi,\tilde\theta)\mid
\theta,\tilde\theta\in\mathcal S,\,\Phi\in{\rm Diff}^{\rm loc}_{(x,u)}\colon\Phi_*\mathcal L_\theta=\mathcal L_{\tilde\theta}\big\}.
\]
Elements of~$\mathcal G^\sim$ are called \emph{admissible (point) transformations} within the class~$\mathcal L|_{\mathcal S}$. 
The pushforward of $\theta$ by $\Phi$ is defined by $\Phi_*\theta=\tilde\theta$ if $\Phi_*\mathcal L_\theta=\mathcal L_{\tilde\theta}$.

The definitions of all notions related to groupoids are obvious. 
Thus, the \emph{source} and \emph{target} maps ${\rm s},{\rm t}\colon\mathcal G^\sim\to\mathcal S$ 
are defined by ${\rm s}(\mathcal T)=\theta$ and ${\rm t}(\mathcal T)=\tilde\theta$ 
for any $\mathcal T=(\theta,\Phi,\tilde\theta)\in\mathcal G^\sim$, 
which gives rise to the groupoid notation $\mathcal G^\sim\rightrightarrows\mathcal S$, 
where the symbol~``$\rightrightarrows$'' denotes the pair of the source and target maps. 
Admissible transformations~$\mathcal T$ and $\mathcal T'=(\theta',\Phi',\tilde\theta')$ are \emph{composable} 
if $\tilde\theta=\theta'$, and then their \emph{composition} is $\mathcal T\star\mathcal T'=(\theta,\Phi'\circ\Phi,\tilde\theta')$, 
which defines a natural partial multiplication on~$\mathcal G^\sim$.
For any $\theta\in\mathcal S$, the \emph{unit at~$\theta$} is given by \smash{${\rm id}_\theta:=(\theta,{\rm id}_{(x,u)},\theta)$}, 
where \smash{${\rm id}_{(x,u)}$} is the identity transformation of~$(x,u)$. 
This defines the \emph{object inclusion map} $\mathcal S\ni\theta\mapsto{\rm id}_\theta\in\mathcal G^\sim$, 
i.e., the object set~$\mathcal S$ can be regarded 
to coincide with the base groupoid $\mathcal S\rightrightarrows\mathcal S:=\{{\rm id}_\theta\mid\theta\in\mathcal S\}$. 
The \emph{inverse} of~$\mathcal T$ is $\mathcal T^{-1}:=(\tilde\theta,\Phi^{-1},\theta)$, where 
$\Phi^{-1}$ is the inverse of~$\Phi$.
All required properties like associativity of the partial multiplication, 
its consistency with the source and target maps,  
natural properties of units and inverses are obviously satisfied. 

The \emph{${\rm s}$-fibre over $\theta\in\mathcal S$}, ${\rm s}^{-1}(\theta)\subseteq\mathcal G^\sim$, 
is the set of possible admissible transformations within~$\mathcal L|_{\mathcal S}$ with source at~$\theta$.
Similarly, the \emph{${\rm t}$-fibre over $\theta\in\mathcal S$}, ${\rm t}^{-1}(\theta)\subseteq\mathcal G^\sim$, 
is the set of possible admissible transformations within~$\mathcal L|_{\mathcal S}$ with target at~$\theta$.
The subset $\mathcal G(\theta,\tilde\theta):={\rm s}^{-1}(\theta)\cap{\rm t}^{-1}(\tilde\theta)$ of~$\mathcal G^\sim$ 
with $\theta,\tilde\theta\in\mathcal S$ 
corresponds to the set of point transformations mapping the system~$\mathcal L_\theta$ to the system~$\mathcal L_{\tilde\theta}$.
The \emph{vertex group} $\mathcal G_\theta:=\mathcal G(\theta,\theta)={\rm s}^{-1}(\theta)\cap{\rm t}^{-1}(\theta)$ 
is associated with the point symmetry (pseudo)group~$G_\theta$ of the system~$\mathcal L_\theta$, 
\[G_\theta=\big\{\Phi\in{\rm Diff}^{\rm loc}_{(x,u)}\mid(\theta,\Phi,\theta)\in\mathcal G_\theta\big\}.\] 
The \emph{orbit} $\mathcal O_\theta:={\rm t}\big({\rm s}^{-1}(\theta)\big)$ of~$\theta$ 
is the subset of values of the arbitrary-element tuple 
such that the corresponding systems in the class~$\mathcal L|_{\mathcal S}$ 
are similar to~$\mathcal L_\theta$ with respect to point transformations.

Denote by $\varpi$ and $\varpi^\EqOrd$ the projections from the space with the coordinates $(x,u_{(\EqOrd)},\theta)$ 
to the spaces with the coordinates $(x,u)$ and $(x,u_{(\EqOrd)})$, respectively.

The \emph{(usual) equivalence group}~$G^\sim$ of the class~$\mathcal L|_{\mathcal S}$ 
is the (pseudo)group of point transformations, $\mathscr T$, in the space with the coordinates $(x,u_{(\EqOrd)},\theta)$ 
that are projectable to the spaces with the coordinates $(x,u)$ and $(x,u_{(\EqOrd)})$ 
with $\varpi^\EqOrd_*\mathscr T$ being the standard prolongation of~$\varpi_*\mathscr T$ to $\EqOrd$th order jets $(x,u_{(\EqOrd)})$ 
and that map the class~$\mathcal L|_{\mathcal S}$ onto itself. 
The group~$G^\sim$ can be considered to act in the space with the coordinates $(x,u_{(\EqOrd')},\theta)$, where $\EqOrd'<\EqOrd$, 
if the arbitrary-element tuple depends only on $(x,u_{(\EqOrd')})$.
The notion of usual equivalence group can be generalized in several ways by weakening the specific restrictions on equivalence transformations, 
which are their projectability and their locality with respect to arbitrary elements. 
This gives the notions of generalized equivalence group and extended equivalence group, respectively, 
or the notions of extended generalized equivalence group if both restrictions are weakened simultaneously 
\cite{ivan2010a,mele1996a,opan2017a,popo2006b,popo2010a,vane2007a,vane2009a}.

The \emph{action groupoid} $\mathcal G^{G^\sim}\!\!$ of the equivalence group~$G^\sim$ 
of the class~$\mathcal L|_{\mathcal S}$, 
\[
\mathcal G^{G^\sim}\!\!:=\big\{(\theta,\varpi_*\mathscr T,\mathscr T_*\theta)\mid\theta\in\mathcal S,\,\mathscr T\in G^\sim\big\},
\] 
is a subgroupoid of the equivalence groupoid~$\mathcal G^\sim$ of this class, 
$\mathcal G^{G^\sim}\!\!\subseteq\mathcal G^\sim$, with the same object set~$\mathcal S$. 
We say that an admissible transformation~$\mathcal T$ in the class~$\mathcal L|_{\mathcal S}$ 
is generated by an equivalence transformation of this class if $\mathcal T\in\mathcal G^{G^\sim}\!\!$.

The \emph{fundamental groupoid}~$\mathcal G^{\rm f}$ of the class~$\mathcal L|_{\mathcal S}$ 
is the disjoint union of the vertex groups $\mathcal G_\theta$, $\theta\in\mathcal S$, 
$\mathcal G^{\rm f}=\sqcup_{\theta\in\mathcal S}\mathcal G_\theta$.
Since it has the same object set~$\mathcal S$ and the same vertex groups as~$\mathcal G^\sim$
and $\mathcal T^{-1}\mathcal G_{\tilde\theta}\mathcal T=\mathcal G_\theta$ 
for any $\mathcal T\in\mathcal G(\theta,\tilde\theta)$,
it is a normal subgroupoid of the equivalence groupoid~$\mathcal G^\sim$, 
which is also called the \emph{fundamental subgroupoid} of~$\mathcal G^\sim$. 
In other words, the groupoid~$\mathcal G^{\rm f}$ is constituted by 
the admissible transformations generated by point symmetry transformations 
of systems from~$\mathcal L|_{\mathcal S}$, 
\[
\mathcal G^{\rm f}:=\big\{(\theta,\Phi,\theta)\mid\theta\in\mathcal S,\,\Phi\in G_\theta\big\}. 
\]
The \emph{kernel point symmetry group} $G^\cap:=\cap_{\theta\in\mathcal S}G_\theta$ 
of systems from the class~$\mathcal L|_{\mathcal S}$, 
which consists of the common point symmetries of these systems,
can be associated with the normal subgroup~$\tilde G^\cap$ of~$G^\sim$ 
whose elements are obtained from elements of~$G^\cap$ 
by the standard prolongation to $\EqOrd'$th order jets $(x,u_{(\EqOrd')})$ 
and the trivial prolongation to the arbitrary-element tuple~$\theta$, 
$G^\cap=\varpi_*\tilde G^\cap$.
Thus, \smash{$\tilde G^\cap$} is the unfaithful subgroup of~$G^\sim$ 
under the action on~$\mathcal L|_{\mathcal S}$.   

The ${\rm s}$-, the ${\rm t}$- and the conjugation actions of~$G^\sim$ on~$\mathcal G^\sim$ 
respectively defined by 
\[
\mathcal T=(\theta,\Phi,\tilde\theta) \ \stackrel{\mathscr T}\mapsto \
(\mathscr T_*\theta,\Phi\circ(\varpi_*\mathscr T)^{-1},\tilde\theta),\ 
(\theta,(\varpi_*\mathscr T)\circ\Phi,\mathscr T_*\tilde\theta),\ 
(\mathscr T_*\theta,(\varpi_*\mathscr T)\circ\Phi\circ(\varpi_*\mathscr T)^{-1},\mathscr T_*\tilde\theta) 
\] 
for any $\mathscr T\in G^\sim$ and for any $\mathcal T=(\theta,\Phi,\tilde\theta)\in\mathcal G^\sim$, 
induce several equivalence relations on~$\mathcal G^\sim$
(${\rm s}$-$G^\sim$-equivalence, ${\rm t}$-$G^\sim$-equivalence, $G^\sim$-conjugation 
and $G^\sim$-equivalence).

\begin{definition}
Admissible transformations $\mathcal T^1=(\theta^1,\Phi^1,\tilde\theta^1)$ and $\mathcal T^2=(\theta^2,\Phi^2,\tilde\theta^2)$ 
in the class $\mathcal L|_{\mathcal S}$
are called \emph{conjugate} with respect to the equivalence group~$G^\sim$ of this class
if there exists $\mathscr T\in G^\sim$ such that
$\theta^2=\mathscr T_*\theta^1$, $\tilde\theta^2=\mathscr T_*\tilde\theta^1$
and $\Phi^2=(\varpi_*\mathscr T)\circ\Phi^1\circ(\varpi_*\mathscr T)^{-1}$.
Admissible transformations $\mathcal T^1$ and $\mathcal T^2$ 
are called \emph{$G^\sim$-equivalent} 
if there exist $\mathscr T,\tilde{\mathscr T}\in G^\sim$ such that
$\theta^2=\mathscr T_*\theta^1$, $\tilde\theta^2=\tilde{\mathscr T}_*\tilde\theta^1$
and $\Phi^2=(\varpi_*\tilde{\mathscr T})\circ\Phi^1\circ(\varpi_*\mathscr T)^{-1}$. 
If additionally \smash{$\tilde{\mathscr T}={\rm id}_{(x,u_{(\EqOrd)},\theta)}$} 
(resp.\ \smash{$\mathscr T={\rm id}_{(x,u_{(\EqOrd)},\theta)}$}), 
then the admissible transformations $\mathcal T^1$ and $\mathcal T^2$ 
are called \emph{${\rm s}$-$G^\sim$-equivalent} (resp.\ \emph{${\rm t}$-$G^\sim$-equivalent}).
\end{definition}

A different terminology was used in~\cite{popo2010a}, 
where the stronger equivalence relation of $G^\sim$-con\-ju\-ga\-tion of admissible transformations 
was called $G^\sim$-equivalence,
whereas in the present paper we use a weaker notion of $G^\sim$-equivalence of admissible transformations.
An admissible transformation in~$\mathcal L|_{\mathcal S}$ belongs to~$\mathcal G^{G^\sim}\!\!$ if and only if
this admissible transformation is $G^\sim$-equivalent in the above sense
to the identity admissible transformation with the same source system.

Since the fundamental groupoid is a normal subgroupoid of~$\mathcal G^\sim$,
the Frobenius product 
\[
\mathcal G^{\rm f}\star\mathcal G^{G^\sim}\!\!=\big\{\mathcal T\star\mathcal T'\mid
\mathcal T\in\mathcal G^{\rm f},\,\mathcal T'\in\mathcal G^{G^\sim}\!\!,\,{\rm t}(\mathcal T)={\rm s}(\mathcal T')\big\}
\] 
is a subgroupoid of~$\mathcal G^\sim$, 
which coincides with the image of~$\mathcal G^{\rm f}$ 
under the ${\rm s}$-action (resp.\ the ${\rm t}$-action) of~$G^\sim$ on~$\mathcal G^\sim$.

There are several kinds of classes of differential equations 
that are convenient for group classification by the algebraic method 
in different ways \cite{bihl2012b,kuru2018a,popo2006b,popo2010a}. 

\begin{definition}
The class~$\mathcal L|_{\mathcal S}$ is called
\emph{normalized} if $\mathcal G^{G^\sim}\!\!=\mathcal G^\sim$. 
It is called \emph{semi-normalized} if $\mathcal G^{\rm f}\star\mathcal G^{G^\sim}\!\!=\mathcal G^\sim$. 
Depending on the kind of the equivalence group~$G^\sim$ 
(the usual, the generalized, the extended or the extended generalized equivalence group of~$\mathcal L|_{\mathcal S}$), 
we distinguish the \mbox{(semi-)}normalization in the usual, the generalized, the extended or the extended generalized sense. 
\end{definition}

\begin{definition}
Let $\mathcal G^H$ be the action groupoid of a subgroup $H$ of~$G^\sim$. 
Suppose that a family
$N_{\mathcal S}:=\{N_\theta<G_\theta\mid\theta\in\mathcal S\}$ 
of subgroups of the point symmetry groups~$G_\theta$
with the associated subgroups $\mathcal N_\theta:=\{(\theta,\Phi,\theta)\mid\theta\in\mathcal S,\,\Phi\in N_\theta\}$ 
of the vertex groups~$\mathcal G_\theta$ 
satisfies the property $\mathcal T\mathcal N_\theta=\mathcal N_{\mathcal T\theta}\mathcal T$
for any $\theta\in\mathcal S$ and for any $\mathcal T\in\mathcal G^H$ with ${\rm s}(\mathcal T)=\theta$.
Then the Frobenius product 
\[
\mathcal N^{\rm f}\star\mathcal G^H=\big\{\mathcal T\star\mathcal T'\mid
\mathcal T\in\mathcal N^{\rm f},\,\mathcal T'\in\mathcal G^H,\,{\rm t}(\mathcal T)={\rm s}(\mathcal T')\big\},
\]
with $\mathcal N^{\rm f}:=\sqcup_{\theta\in\mathcal S}\mathcal N_\theta$  
is a subgroupoid of~$\mathcal G^\sim$, 
which coincides with the image of~$\mathcal N^{\rm f}$ 
under the ${\rm s}$-action (resp.\ the ${\rm t}$-action) of~$H$ on~$\mathcal G^\sim$.
If $\mathcal N^{\rm f}\star\mathcal G^H=\mathcal G^\sim$, 
we call the class~$\mathcal L|_{\mathcal S}$ 
\emph{semi-normalized with respect to the subgroup $H$ of~$G^\sim$ 
and the family $N_{\mathcal S}$ of subgroups of the point symmetry groups}.
If additionally $\mathcal G^H\cap\mathcal N^{\rm f}=\mathcal S\rightrightarrows\mathcal S$, 
then the class~$\mathcal L|_{\mathcal S}$ is called
\emph{disjointedly semi-normalized with respect to the subgroup $H$ of~$G^\sim$ 
and the family $N_{\mathcal S}$ of subgroups of the point symmetry groups}.
\end{definition}

If $H=G^\sim$ and $N_\theta=\{{\rm id}_{(x,u)}\}$ for any $\theta\in\mathcal S$, 
then a class (disjointedly) semi-normalized with respect to the group $H$
and the family $N_{\mathcal S}$ is literally normalized.
If $H=G^\sim$ and $N_\theta=G_\theta$ for any $\theta\in\mathcal S$, 
then a class semi-normalized with respect to the group $H$
and the family $N_{\mathcal S}$ is literally semi-normalized.
It is obvious that a normalized class is semi-normalized.

The most powerful method for describing admissible transformations 
within a class of differential equations is still the \emph{direct method}, 
which is based on the definition of admissible transformations.  
Applying this method to the class~$\mathcal L|_{\mathcal S}$,
we consider an arbitrary pair $(\theta,\tilde\theta)\in\mathcal S\times\mathcal S$ 
and a point transformation in the space with coordinates $(x,u)$ 
of the most general form $\Phi$: $\tilde x=X(x,u)$, $\tilde u=U(x,u)$ 
with nonzero Jacobian $|\p(X,U)/\p(x,u)|$
and assume that $\Phi_*\mathcal L_\theta=\mathcal L_{\smash{\tilde\theta}}$.
Expressing the required derivatives of~$\tilde u$ with respect to~$\tilde x$ 
in terms of derivatives of~$u$ with respect to~$x$ using the chain rule, 
we substitute the derived expressions into the system~$\mathcal L_{\smash{\tilde\theta}}$, 
obtaining the system~$(\Phi^{-1})_*\mathcal L_{\smash{\tilde\theta}}$, 
which should be identically satisfied by solutions of the system~$\mathcal L_\theta$. 
To take into account the last condition, 
we fix a ranking of derivatives of~$u$ that is consistent with the structure of~$\mathcal L_\theta$, 
substitute the expressions for the leading derivatives of~$u$ 
in view of the system~$\mathcal L_\theta$ and its differential consequences 
into~$(\Phi^{-1})_*\mathcal L_{\smash{\tilde\theta}}$ 
and split the resulting system with respect to the involved parametric derivatives of~$u$. 
As a result, we obtain a system that implies both the expression of~$\tilde\theta$ via~$(\theta,X,U)$ 
and the system~${\rm DE}$ of determining equations for components of~$\Phi$. 
The system~${\rm DE}$ involves only the arbitrary-element tuple~$\theta$ (resp.\ $(\Phi^{-1})_*\tilde\theta$).

Assuming~$\theta$ varying within~$\mathcal S$ 
and splitting with respect to derivatives of~$\theta$ in view of the auxiliary system defining the set~$\mathcal S$,%
\footnote{%
This means that we set a ranking among the derivatives of~$\theta$ 
that is consistent with structure of the auxiliary system, 
solve this system jointly with its differential consequences for the leading derivatives of~$\theta$, 
substitute the derived expressions into~${\rm DE}$ and 
split the obtained system with respect to the parametric derivatives of~$\theta$.
}
we get the system~${\rm DE}^\sim$ of determining equations for the $(x,u)$-components 
of usual equivalence transformations. 
After finding the $(x,u)$-components via the integration of~${\rm DE}^\sim$, 
the $\theta$-component of usual equivalence transformations is obtained from the above expression for~$\tilde\theta$. 
As a result, we construct the usual equivalence group~$G^\sim$ of the class~$\mathcal L|_{\mathcal S}$.

If the solution sets of~${\rm DE}$ and~${\rm DE}^\sim$ coincide, 
then $\mathcal G^\sim=\mathcal G^{G^\sim}\!\!$, i.e., the class~$\mathcal L|_{\mathcal S}$ is normalized, 
which completes the description of the equivalence groupoid~$\mathcal G^\sim$. 
The first example of such description in the literature was given for 
the normalized class of generalized Burgers equations of the form $u_t+uu_x+f(t,x)u_{xx}=0$ 
in~\cite{king1991c} although the normalization property was implicitly used there. 

Otherwise, the class~$\mathcal L|_{\mathcal S}$ is not normalized, 
and integrating~${\rm DE}$, which can be carried out up to $G^\sim$-equivalence of admissible transformations,   
is a complicated problem. 
A number of various techniques can be used to simplify the solution of this problem. 
Below we present some of them.

\medskip\par\noindent 
{\it Partition of classes.} 
Let the set~$\mathcal S$ be represented as a disjoint union of its subsets, 
$\mathcal S=\sqcup_{\gamma\in\Gamma}\mathcal S_\gamma$ with some index set~$\Gamma$, 
where each of the subsets~$\mathcal S_\gamma$ is singled out from~$\mathcal S$ 
by additional constraints, which are differential equations or differential inequalities. 
The partition of~$\mathcal S$ is equivalent to the partition 
of the class~$\mathcal L|_{\mathcal S}$ into the subclasses $\mathcal L|_{\mathcal S_\gamma}$ 
with $\gamma$ running through~$\Gamma$, 
$\mathcal L|_{\mathcal S}=\sqcup_{\gamma\in\Gamma}\mathcal L|_{\mathcal S_\gamma}$. 
Denote by~$G^\sim_\gamma$ and by~$\mathcal G^\sim_\gamma$ 
the equivalence group and the equivalence groupoid of the subclass~$\mathcal L|_{\mathcal S_\gamma}$, 
respectively.

If systems from different subclasses of the partition are not related by point transformations, 
then the partition of the class~$\mathcal L|_{\mathcal S}$ induces 
the partition $\mathcal G^\sim=\sqcup_{\gamma\in\Gamma}\mathcal G^\sim_\gamma$ of its equivalence groupoid.
In general, the structure of the groupoid of a subclass may even be more complicated 
than that of the entire class. 
This is why a preliminary analysis of the system~${\rm DE}$ is needed  
for an appropriate partition of the class~$\mathcal L|_{\mathcal S}$, 
where for any $\gamma\in\Gamma$ the structure of~$\mathcal G^\sim_\gamma$ is simpler than the structure of~$\mathcal G^\sim$. 
Then we can find the subgroupoids~$\mathcal G^\sim_\gamma$ separately and unite them. 
The best kind of partitions is given by partitions into normalized subclasses,  
for which \smash{$\mathcal G^\sim_\gamma=\mathcal G^{G^\sim_\gamma}$} and thus $\mathcal G^\sim=\sqcup_{\gamma\in\Gamma}\mathcal G^{G^\sim_\gamma}$ 
\cite{popo2006b,popo2010a}. 

There are several generalizations of the partition technique. 

Disjoint subclasses may be related by point transformations, 
and thus the partition of the class~$\mathcal L|_{\mathcal S}$ 
into the subclasses~$\mathcal L|_{\mathcal S_\gamma}$
does not induce the partition of the equivalence groupoid~$\mathcal G^\sim$ 
into the equivalence groupoids~$\mathcal G^\sim_\gamma$. 
Consider a simple situation, where we have a partition of~$\mathcal L|_{\mathcal S}$ 
into normalized subclasses~$\mathcal L|_{\mathcal S_\gamma}$, $\gamma\in\Gamma$, 
and for some fixed $\gamma_0\in\Gamma$
and for each $\gamma\in\Gamma$ there exists a point transformation~$\Phi_\gamma$ 
that maps~$\mathcal L|_{\mathcal S_\gamma}$ onto~$\mathcal L|_{\mathcal S_{\gamma_0}}$. 
We can assume that $\Phi_{\gamma_0}={\rm id}_{(x,u)}$. 
In fact, for any $\gamma\in\Gamma$ the normalization of $\mathcal L|_{\mathcal S_\gamma}$ 
follows from the normalization of $\mathcal L|_{\mathcal S_{\gamma_0}}$ and the existence of~$\Phi_\gamma$.    
Then 
\begin{gather}\label{eq:GroupoidsWithMixedTrans}
\mathcal G^\sim=\Big\{
\big((\Phi_{\gamma}^{-1})_*\theta,\Phi_{\gamma'}^{-1}\circ(\varpi\mathscr T)\circ\Phi_{\gamma},(\Phi_{\gamma'})_*(\mathscr T_*\theta)\big)
\,\big|\,\theta\in\mathcal S_{\gamma_0},\,\mathscr T\in G^\sim_{\gamma_0}, \gamma,\gamma'\in\Gamma
\Big\}.
\end{gather}
This structure is admitted by the groupoid of the class~\eqref{eq:GenWaveEqsSpecial},
where the parameter~$\sigma$ plays the role of~$\gamma$, see Remark~\ref{rem:StructuteOfGroupoidOfW12} below.

The condition that the class~$\mathcal L|_{\mathcal S}$ is a disjoint union of appropriate subclasses 
can be weakened by allowing a proper intersection of subclasses in the union. 
Thus, in \cite{vane2007a} a class of variable-coefficient reaction--diffusion equations with power nonlinearities
was represented as a non-disjoint union of normalized subclasses, 
and its groupoid was proved to be constituted by the admissible transformations for the action groupoids of the subclasses 
and the compositions of such composable admissible transformations from the action groupoids of different subclasses 
with nonempty intersections. 

\medskip\par\noindent 
{\it Construction of generalized/extended/extended generalized equivalence group.} 
If the class~$\mathcal L|_{\mathcal S}$ is not normalized in the usual sense, 
one can try to describe its equivalence groupoid~$\mathcal G^\sim$  
via finding a generalized counterpart of the usual equivalence group~$G^\sim$, 
with respect to which the class~$\mathcal L|_{\mathcal S}$ is normalized 
in the corresponding sense~\cite{opan2017a,popo2010a}.

\medskip\par\noindent 
{\it Mappings between classes.}
Suppose that there are a class~$\mathcal L'|_{\mathcal S'}$ of (systems of) differential equations 
with the same independent and dependent variables $x$ and~$u$ as systems from the class~$\mathcal L|_{\mathcal S}$
and a family of point transformations $\mathcal F=\{\Psi^\theta\mid\theta\in\mathcal S\}$
such that $\Psi^\theta_*\mathcal L_\theta\in\mathcal L'|_{\mathcal S'}$ for any $\theta\in\mathcal S$, 
and for any $\theta'\in\mathcal S'$ there exists $\theta\in\mathcal S$ with $\Psi^\theta_*\mathcal L_\theta=\mathcal L'_{\theta'}$. 
Then we say that the family~$\mathcal F$ generates the mapping~$\mathcal F_*$ 
from the class~$\mathcal L|_{\mathcal S}$ onto the class~$\mathcal L'|_{\mathcal S'}$, where
$\mathcal F_*\mathcal L_\theta:=\Psi^\theta_*\mathcal L_\theta$, 
or, equivalently, the mapping~$\mathcal F_*\colon\mathcal S\to\mathcal S'$ 
with $\mathcal F_*\theta=\theta'$ if $\Psi^\theta_*\mathcal L_\theta=\mathcal L'_{\theta'}$; 
see \cite{popo2010a,vane2009a} for the first explicit discussions of mappings between the classes.
Via~$\mathcal F_*$, the family~$\mathcal F$ also induces the mapping 
from the equivalence groupoid~$\mathcal G^\sim$ of the class~$\mathcal L|_{\mathcal S}$
to the equivalence groupoid~$\mathcal G^\sim{}'$ of the class~$\mathcal L'|_{\mathcal S'}$ 
that is defined by 
\[
\mathcal G^\sim\ni\mathcal T=(\theta,\Phi,\tilde\theta) \ \stackrel{\mathcal F_*\mathstrut}\mapsto \
\big((\Psi^\theta)_*\theta,\Psi^{\tilde\theta}\circ\Phi\circ(\Psi^\theta)^{-1},(\Psi^{\tilde\theta})_*\tilde\theta\big)\in\mathcal G^\sim{}'.
\] 
We will denote this mapping by the same symbol as the corresponding mapping between classes. 
The mapping $\mathcal F_*\colon\mathcal G^\sim\to\mathcal G^\sim{}'$ is in fact a groupoid homomorphism 
since 
\[
\mathcal F_*(\mathcal T_1\star\mathcal T_2)=(\mathcal F_*\mathcal T_1)\star(\mathcal F_*\mathcal T_2),\quad
\mathcal F_*({\rm id}_\theta)={\rm id}_{\mathcal F_*\theta},\quad
\mathcal F_*(\mathcal T^{-1})=(\mathcal F_*\mathcal T)^{-1}
\] 
for any $\mathcal T_1,\mathcal T_2\in\mathcal G^\sim$, 
any $\theta\in\mathcal S$ and any $\mathcal T\in\mathcal G^\sim$. 
Moreover, this homomorphism is surjective. 
Indeed, take any $\mathcal T'=(\theta',\Phi',\tilde\theta')\in\mathcal G^\sim{}'$. 
By the choice of the family~$\mathcal F$, there exist $\theta,\tilde\theta\in\mathcal S$ 
such that $\mathcal F_*(\theta)=\theta'$ and $\mathcal F_*(\tilde\theta)=\tilde\theta'$. 
The triple $\mathcal T=(\theta,\Phi,\tilde\theta)$ with $\Phi=(\Psi^{\tilde\theta})^{-1}\circ\Phi'\circ\Psi^\theta$
belongs to $\mathcal G^\sim$, and $\mathcal F_*\mathcal T=\mathcal T'$.

Under an appropriate choice of~$\mathcal F$, 
the structure of~$\mathcal G^\sim{}'$ is simpler than the structure of~$\mathcal G^\sim$. 
Then after the study of~$\mathcal G^\sim{}'$, we can pull back obtained results 
with respect to~$\mathcal F$ and thus get results on~$\mathcal G^\sim$. 
For example, an appropriate partition of a class into its subclasses can become evident 
only after a mapping of this class to another class~\cite{opan2019b}. 
The complete group classification of~$\mathcal L|_{\mathcal S}$ up to $\mathcal G^\sim$-equivalence 
can easily be derived from the analogous classification of~$\mathcal L|_{\mathcal S'}$ 
using the pullback by~$\mathcal F$. 
In this way, the complete group classifications 
of the classes of (1+1)-dimensional Kolmogorov and Fokker--Planck equations 
modulo the general point equivalence were obtained 
from the classical group classification of the class of linear heat equations with potentials, 
see Corollaries~7 and~17 in~\cite{popo2008a}.
Using mappings between classes for deriving complete group classifications 
up to $G^\sim$-equivalence needs a more delicate consideration~\cite{vane2009a}.

An important particular case of mappings between classes 
is given by mappings of classes to their subclasses 
that are generated by subgroups of the corresponding equivalence groups. 
Let $\mathcal S'$ be a subset of~$\mathcal S$ that is singled out from~$\mathcal S$ 
by additional auxiliary differential equations or inequalities 
with respect to the arbitrary-element tuple~$\theta$. 
Thus, $\mathcal L|_{\mathcal S'}$ is a subclass of the class~$\mathcal L|_{\mathcal S}$, 
and its equivalence groupoid~$\mathcal G^\sim{}'$ is a subgroupoid 
of the equivalence groupoid~$\mathcal G^\sim$ of the class~$\mathcal L|_{\mathcal S}$.
Suppose that for some subgroup~$H$ of~$G^\sim$
each orbit of the action groupoid~$\mathcal G^H$ 
intersects $\mathcal S'$ by a single~$\theta'$.
Denote by $\Psi^\theta$ the point transformation $\varpi_*\mathscr T$ 
with $\mathscr T\in H$ such that $\mathscr T_*\theta\in\mathcal S'$. 
Then the family $\mathcal F=\{\Psi^\theta\mid\theta\in\mathcal S\}$ 
satisfies the required conditions to generate the corresponding mapping
$\mathcal F_*\colon\mathcal L|_{\mathcal S}\to\mathcal L|_{\mathcal S'}$
and the corresponding surjective homomorphism
$\mathcal F_*\colon\mathcal G^\sim\to\mathcal G^\sim{}'$.
In practice, such mappings are realized via gauging arbitrary elements 
by equivalence transformations. 

\medskip\par\noindent 
{\it Conditional equivalence groups.}
The equivalence group of a subclass of the class~$\mathcal L|_{\mathcal S}$ is called 
a conditional equivalence group of this class. 
The conditional equivalence group~$G^\sim{}'$ of~$\mathcal L|_{\mathcal S}$ 
associated with the subclass~$\mathcal L|_{\mathcal S'}$ is called maximal if 
for any subclass of~$\mathcal L|_{\mathcal S}$ properly containing~$\mathcal L|_{\mathcal S'}$, 
its equivalence group does not contain~$G^\sim{}'$. 
The equivalence group~$G^\sim$ of the entire class~$\mathcal L|_{\mathcal S}$ 
acts on subclasses of~$\mathcal L|_{\mathcal S}$ simultaneously with their equivalence groups, 
and the set of maximal conditional equivalence groups of~$\mathcal L|_{\mathcal S}$ is closed under 
this action. 
Hence maximal conditional equivalence groups of~$\mathcal L|_{\mathcal S}$ can be classified 
modulo $G^\sim$-equivalence. 
This classification can be a step in the description of~$\mathcal G^\sim$ \cite{popo2010a}. 
For some classes, this lone step gives the complete description 
of the corresponding equivalence groupoids~\cite{opan2019a,opan2017a}.
The classification of maximal conditional equivalence groups of the class~$\mathcal L|_{\mathcal S}$ 
can be combined with a partition of~$\mathcal L|_{\mathcal S}$ into subclasses 
that is consistent with the structure of the set of such groups. 
Generalized versions of conditional equivalence groups also can be considered~\cite{opan2019a,vane2017a}.

\medskip\par\noindent 
{\it Generating set of admissible transformations.}
A set $\mathcal B=\{\mathcal T_\gamma\in\mathcal G^\sim\mid\gamma\in\Gamma\}$, 
where $\Gamma$ is an index set, 
is called a generating set of admissible transformations for the class~$\mathcal L|_{\mathcal S}$ 
up to $G^\sim$-equivalence if any admissible transformation of this class 
can be represented as the composition of a finite number of elements of the set~$\mathcal B\cup\hat{\mathcal B}\cup\mathcal G^{G^\sim}\!\!$, 
where $\hat{\mathcal B}$ is the set of inverses of admissible transformations from~$\mathcal B$, 
$\hat{\mathcal B}:=\{\mathcal T^{-1}\mid\mathcal T\in\mathcal B\}$.
To make the set~$\mathcal B$ as small as possible, it is natural 
to choose~$\mathcal B$ as a subset of $\mathcal G^\sim\setminus\mathcal G^{G^\sim}\!\!$. 
Moreover, if a canonical representative in a coset of $G^\sim$-equivalent admissible transformations can be assigned, 
only this representative should be selected from the coset for including in~$\mathcal B$.  

We call admissible transformations~$\mathcal T_1$ and~$\mathcal T_2$ for the class~$\mathcal L|_{\mathcal S}$
\emph{composable up to $G^\sim$-equiv\-a\-lence} if 
an admissible transformation that is $G^\sim$-equivalent to~$\mathcal T_1$ is composable with~$\mathcal T_2$ 
or, equivalently, if 
$\mathcal T_1$ is composable with an admissible transformation that is $G^\sim$-equivalent to~$\mathcal T_2$. 
It happens if and only if there exists $\mathscr T\in G^\sim$ such that
$\mathscr T_*\big({\rm t}(\mathcal T_1)\big)={\rm s}(\mathcal T_2)$.
We call a subset~$\mathcal B$ of~$\mathcal G^\sim$ \emph{self-consistent with respect to $G^\sim$-equivalence}
if the composability of elements of~$\mathcal B\cup\hat{\mathcal B}$ up to $G^\sim$-equivalence implies 
their usual composability. (The converse implication always holds.) 
A~necessary condition for the self-consistency of~$\mathcal B$ is the equality
\[
\big({\rm s}(\mathcal B)\times{\rm t}(\mathcal B)\big)\cap
({\rm s}\times{\rm t})\big(\mathcal G^{G^\sim}\!\setminus\mathcal G^{\rm f}\big)=\varnothing
\]
meaning that there is no element of the action groupoid~$\mathcal G^{G^\sim}\!$ 
with different source and target in ${\rm s}(\mathcal B)\times{\rm t}(\mathcal B)$.
If $\mathcal B$ is a self-consistent generating set of~$\mathcal G^\sim$ with respect to $G^\sim$-equivalence, 
then any element of~$\mathcal G^\sim$ is $G^\sim$-equivalent 
to the composition of a finite number of elements of~$\mathcal B\cup\hat{\mathcal B}$. 
More specifically, then for any $\mathcal T\in\mathcal G^\sim$ 
there exist $n\in\mathbb N\cup\{0\}$, $\mathcal T_0,\mathcal T_{n+1}\in\mathcal G^{G^\sim}\!$ 
and $\mathcal T_1,\dots,\mathcal T_n\in\mathcal B\cup\hat{\mathcal B}$ such that 
\[\mathcal T=\mathcal T_0\star\mathcal T_1\star\dots\star\mathcal T_n\star\mathcal T_{n+1}.\] 
If additionally $\mathcal B$ is minimal up to $G^\sim$-equivalence, 
then this representation for~$\mathcal T$ is unique up to the transformations 
\[
\tilde{\mathcal T}_0=\mathcal T_0\star\breve{\mathcal T},\quad 
\tilde{\mathcal T}_{n+1}=\mathcal T_n^{-1}\star\dots\star\mathcal T_1^{-1}\star\breve{\mathcal T}^{-1}\star\mathcal T_1\star\dots\star\mathcal T_n\star\mathcal T_{n+1}
\] 
with an arbitrary $\breve{\mathcal T}\in\mathcal G_{{\rm t}(\mathcal T_0)}$ if $n>0$  
and under setting $\mathcal T_{n+1}={\rm id}_{{\rm t}(\mathcal T_0)}$ if $n=0$.

\medskip\par\noindent 
{\it Furcate splitting.}  
This technique was suggested in~\cite{niki2001a}
as a refinement of the direct method of group classification. 
Its essence is a special way of handling the system of determining equations 
for Lie symmetries of systems from the class under study 
depending on the possible number of independent constraints on values of~$\theta$
that are induced by this system.
This is why it can be extended to descriptions of other objects 
that are related to systems from classes of differential equations 
and are computed via solving certain systems of determining equations, 
including conservation laws~\cite{bihlo2019b}, conditional equivalence groups~\cite{opan2017a} 
and generating sets of admissible transformations 
(see footnote~\ref{fnt:MethodOfFurcateSplittingForGenerating SetsOfAdmTrans} below).
The method of furcate splitting can further be enhanced 
by involving algebraic techniques~\cite{bihl2019a,bihlo2019b}.

\medskip\par\noindent 
{\it (Bijective) functors between groupoids.} 
Suppose that we construct an isomorphism between~$\mathcal G^\sim$ and 
the equivalence groupoid~$\tilde{\mathcal G}^\sim$ of a class~$\tilde{\mathcal L}|_{\tilde{\mathcal S}}$, 
and the description of~$\tilde{\mathcal G}^\sim$ has been known or 
it is easier or more convenient to describe the groupoid~$\tilde{\mathcal G}^\sim$ than the original groupoid~$\mathcal G^\sim$. 
For the latter option, for example, some computation techniques that are relevant for~$\tilde{\mathcal G}^\sim$ 
might be inapplicable to~$\mathcal G^\sim$.  
Here it is not necessary for the classes~$\mathcal L|_{\mathcal S}$ and~$\tilde{\mathcal L}|_{\tilde{\mathcal S}}$
to be related by a family of point transformations. 
Then the description of~\smash{$\tilde{\mathcal G}^\sim$} implies the description of~$\mathcal G^\sim$.

The technique involving bijective functors 
is effectively applied to the study of equivalence groupoids of classes of differential equations
for the first time in the present paper. 
It is still unclear whether the construction of non-bijective functors 
from~$\mathcal G^\sim$ to~\smash{$\tilde{\mathcal G}^\sim$} or conversely is useful as well. 

\medskip\par
The infinitesimal counterparts of the (pseudo)groups~$G^\theta$, $G^\cap$ and~$G^\sim$ 
are the Lie algebras~$\mathfrak g^\theta$, $\mathfrak g^\cap$ and~$\mathfrak g^\sim$ 
that are constituted by the generators of local one-parameter subgroups of the corresponding groups 
and which are called  
the \emph{maximal Lie invariance algebras} of the systems~$\mathcal L_\theta$,
the \emph{kernel invariance algebra} of systems from the class~$\mathcal L|_{\mathcal S}$ and 
the \emph{equivalence algebra} of the class~$\mathcal L|_{\mathcal S}$, respectively. 
Note that $\mathfrak g^\cap=\cap_{\theta\in\mathcal S}\mathfrak g_\theta$.

The \emph{(complete) group classification problem} for the class~$\mathcal L|_{\mathcal S}$ 
up to $G^\sim$-equivalence (resp.\ up to $\mathcal G^\sim$-equivalence)
is to find $\mathfrak g^\cap$ and an exhaustive list 
of $G^\sim$-inequivalent (resp.\ $\mathcal G^\sim$-inequivalent) values of~$\theta$ 
jointly with the corresponding algebras~$\mathfrak g_\theta$ 
for which $\mathfrak g_\theta\ne\mathfrak g^\cap$.
An admissible transformation from~$\mathcal G^\sim\setminus\mathcal G^{G^\sim}\!$ 
between systems from the final group classification list modulo $G^\sim$-equivalence 
is called an \emph{additional equivalence transformation}. 
Supplementing the group classification up to $G^\sim$-equivalence 
with the complete set of additional equivalence transformations results in 
the group classification up to $\mathcal G^\sim$-equivalence.

Any version of the algebraic method of group classification in fact reduces to 
the classification, modulo $G^\sim$-equivalence, 
of certain subalgebras contained by the span 
$\mathfrak g_\spanindex:=\langle\mathfrak g_\theta, \theta\in\mathcal S\rangle$.
The efficiency of using the algebraic method depends on additional conditions 
satisfied by the class~$\mathcal L|_{\mathcal S}$, 
in particular, how consistent the span $\mathfrak g_\spanindex$ is with $G^\sim$-equivalence
\cite[Section~12]{bihl2012b}. 

Normalized classes are the most convenient for group classification by the algebraic method. 
If the class~$\mathcal L|_{\mathcal S}$ is normalized, 
then $\mathfrak g_\spanindex\subseteq\varpi_*\mathfrak g^\sim$, 
and the solution of the complete group classification problem for this class reduces 
to the classification of appropriate subalgebras of~$\mathfrak g^\sim$ 
whose pushforwards by~$\varpi$ can be qualified as 
the maximal Lie invariance algebras of systems from~$\mathcal L|_{\mathcal S}$. 
Since then $G^\sim$-equivalence coincides with $\mathcal G^\sim$-equivalence, 
it is obvious that there are no additional equivalence transformations 
between Lie-symmetry extensions classified modulo $G^\sim$-equivalence.
Moreover, it is inessential which of the two equivalences is used in the course 
of the classification. 
The above is also true if the class~$\mathcal L|_{\mathcal S}$ 
is semi-normalized with respect to a subgroup $H$ of~$G^\sim$ 
and a family $N_{\mathcal S}$ of subgroups of the point symmetry groups, 
and additionally the subgroup~$H$ and the family $N_{\mathcal S}$ are known. 
In this case, we call the class~$\mathcal L|_{\mathcal S}$ \emph{definitely semi-normalized},
and looking for $G^\sim$-inequivalent subalgebras of $\mathfrak g^\sim$ 
is substituted in the algebraic method
by looking for $H$-inequivalent subalgebras 
of the infinitesimal counterpart~$\mathfrak h$ of~$H$~\cite{kuru2018a}. 
The pure semi-normalization of~$\mathcal L|_{\mathcal S}$ 
at least guarantees that the group classification~$\mathcal L|_{\mathcal S}$ up to $G^\sim$-equivalence 
coincides with that up to $\mathcal G^\sim$-equivalence.

If the class~$\mathcal L|_{\mathcal S}$ is not normalized, 
then some Lie-symmetry extensions within this class
are not related to subalgebras of its equivalence algebra~$\mathfrak g^\sim$.

\begin{definition}\label{def:RegularAndSingularLieSymExtensions}
We call the maximal Lie invariance algebra~$\mathfrak g_\theta$ 
of a system~$\mathcal L_\theta$ from the class~$\mathcal L|_{\mathcal S}$ 
\emph{regular} in this class if there exists a subalgebra~$\mathfrak s$ of~$\mathfrak g^\sim$ 
such that $\mathfrak g_\theta=\varpi_*\mathfrak s$, and \emph{singular} in~$\mathcal L|_{\mathcal S}$ otherwise.  
\end{definition}

If $\mathfrak g_\theta=\mathfrak g^\cap$, 
then the maximal Lie invariance algebra~$\mathfrak g_\theta$ is regular in~$\mathcal L|_{\mathcal S}$ 
since $\mathfrak g^\cap\subseteq\varpi_*\mathfrak g^\sim$.%
\footnote{%
More specifically, the kernel invariance algebra~$\mathfrak g^\cap$ is naturally embedded into~$\mathfrak g^\sim$ 
via the standard prolongation of its elements to $u_{(\EqOrd)}$ in view of the contact structure 
and the trivial prolongation to the arbitrary elements~$\theta$~\cite{card2011a}.
}
If $\mathfrak g_\theta\ne\mathfrak g^\cap$ 
and $\mathfrak g_\theta$ is a regular (resp.\ singular) maximal Lie invariance algebra in the class~$\mathcal L|_{\mathcal S}$, 
then we say that the pair constituted by the value of the arbitrary-element tuple 
and the algebra~$\mathfrak g_\theta$ presents a regular (resp.\ singular) 
Lie-symmetry extension of~$\mathfrak g^\cap$ in this class.  

It is obvious that the sets of regular and singular Lie-symmetry extensions are separately invariant 
with respect to the action of~$G^\sim$ but in general this is not the case for the action of~$\mathcal G^\sim$.
In other words, regular Lie-symmetry extensions are $G^\sim$-inequivalent to singular ones
but may be $\mathcal G^\sim$-equivalent to them, 
see Remark~\ref{rem:GenWaveClassRegularCaseEquivToSingularCases} 
as an example on this claim. 
This also means that the Lie-symmetry extension for a system~$\mathcal L_\theta$ with $\theta$ satisfying 
${\rm s}^{-1}(\theta)\ne{\rm s}^{-1}(\theta)\cap\mathcal G^{G^\sim}$ or, equivalently, 
${\rm t}^{-1}(\theta)\ne{\rm t}^{-1}(\theta)\cap\mathcal G^{G^\sim}$ 
(i.e., for a system being the source or the target of an admissible transformation 
that is not generated by an equivalence transformation in~$\mathcal L|_{\mathcal S}$)
may also be regular. 
This is definitely the case if the quotient of the set~${\rm s}^{-1}(\theta)$ 
with respect to ${\rm t}$-$G^\sim$-equivalence of admissible transformations is discrete, 
as it appears for  the regular Cases~\ref{case14d} and~\ref{case19d} 
of Table~\ref{tab:GenWaveEqsExtensions} below. 

Using Definition~\ref{def:RegularAndSingularLieSymExtensions}, 
we suggest the following procedure of group classification 
for a non-normalized class~$\mathcal L|_{\mathcal S}$ of differential equations 
within the framework of the algebraic method. 
\begin{enumerate}
\item
Describe the equivalence groupoid~$\mathcal G^\sim$ of the class~$\mathcal L|_{\mathcal S}$ 
up to $G^\sim$-equivalence, e.g., via constructing a generating set~$\mathcal B$ of admissible transformations. 
The further consideration simplifies if the set~$\mathcal B$ is minimal and self-consistent with respect to $G^\sim$-equivalence.
\item
Classify, modulo $\mathcal G^\sim$-equivalence, Lie symmetries 
of systems~$\mathcal L_\theta$ with $\theta$ satisfying the condition 
${\rm s}^{-1}(\theta)\ne{\rm s}^{-1}(\theta)\cap\mathcal G^{G^\sim}$ or, equivalently, 
${\rm t}^{-1}(\theta)\ne{\rm t}^{-1}(\theta)\cap\mathcal G^{G^\sim}$. 
This leads to the complete list of $\mathcal G^\sim$-inequivalent Lie-symmetry extensions 
within the class~$\mathcal L|_{\mathcal S}$ 
that are singular or regular but related to other Lie-symmetry extensions 
with elements from~$\mathcal G^\sim\setminus\mathcal G^{G^\sim}\!$. 
Here both the direct and the algebraic methods of group classification 
might be applicable.   
\item
Carry out the (complete) preliminary group classification of the class~$\mathcal L|_{\mathcal S}$. 
The optimized version of such classification includes 
the classification of candidates for appropriate subalgebras of the equivalence algebra~$\mathfrak g^\sim$ 
up to $G^\sim$-equivalence
and, whenever it is possible, the construction of systems from the class~$\mathcal L|_{\mathcal S}$ 
that admit the projections of the above candidates by~$\varpi_*$
as their Lie invariance algebras. 
For each obtained system, we select the candidate that is maximal by inclusion; 
such a candidate always exists. 
\item
Merge the lists obtained in steps~2 and~3 and exclude repetitions up to $\mathcal G^\sim$-equivalence,
which leads to the complete list of Lie-symmetry extensions 
within the class~$\mathcal L|_{\mathcal S}$ up to $\mathcal G^\sim$-equivalence.  
\item
Extend the part of the list from step~4 that is related to the cases of step~2  
by compositions of admissible transformations from the set~$\mathcal B$ modulo $G^\sim$-equivalence.
This gives the complete list of Lie-symmetry extensions 
within the class~$\mathcal L|_{\mathcal S}$ up to $G^\sim$-equivalence. 
All possible additional equivalence transformations between cases in this list 
are generated by elements of~$\mathcal B$ modulo $G^\sim$-equivalence.
\end{enumerate}

The order of steps or even single operations may vary 
depending on the class of differential equations to be studied.

\section{Preliminary study of admissible transformations}\label{sec:PreliminaryStudyOfAdmTrans}

To find the complete point equivalence group~$G^\sim$ of the class~\eqref{eq:GenWaveEqs}
(including both continuous and discrete equivalence transformations)
and the equivalence groupoid~$\mathcal G^\sim$ of this class, it is necessary to apply the direct method
of computing point transformations that relate systems of differential equations.
We will start our consideration with a preliminary study of admissible transformations of the superclass~$\mathcal W_{\rm gen}$,
which constitute the groupoid~$\mathcal G^\sim_{\rm gen}$ of~$\mathcal W_{\rm gen}$.
This will also give relevant information on the group~$G^\sim$ and the groupoid~$\mathcal G^\sim$.

Denote by $\mathcal L_\theta$ the equation in the class~$\mathcal W_{\rm gen}$
that corresponds to a fixed value of the arbitrary-element tuple $\theta=(f,g)$.
An admissible transformation of the class~$\mathcal W_{\rm gen}$ is a triple $(\theta,\Phi,\tilde\theta)$,
where $\theta=(f,g)$ and $\tilde\theta=(\tilde f,\tilde g)$
are respectively the source and target arbitrary-element tuples for~$\mathcal T$,~and
\begin{equation}\label{eq:GeneralEquivalenceTransformation}
\Phi\colon\quad \tilde t = T(t,x,u),\quad \tilde x = X(t,x,u),\quad \tilde u = U(t,x,u)
\end{equation}
with nonvanishing Jacobian $J:=\p(T,X,U)/\p(t,x,u)$
is a point transformation in the space with the coordinates $(t,x,u)$
that maps the equation~$\mathcal L_\theta$ to the equation~$\mathcal L_{\tilde\theta}$.
Therefore, we should directly seek for all point transformations
mapping a fixed equation~$\mathcal L_\theta$ of the form~\eqref{eq:GenWaveEqs}
to an equation~$\mathcal L_{\tilde\theta}$ of the same form,
\[
\tilde u_{\tilde t\tilde t} = \tilde f(\tilde x,\tilde u)\tilde u_{\tilde x\tilde x} + \tilde g(\tilde x,\tilde u).
\]

To carry out the transformation~$\Phi$ in practice,
it is necessary to find its prolongation to derivatives of $u$ up to order two.
For this we act by the total derivative operators $\DD_t$ and $\DD_x$, respectively,
on the expression $\tilde u(\tilde t,\tilde x)=U(t,x,u)$, assuming $\tilde t=T(t,x,u)$ and $\tilde x=X(t,x,u)$.
This gives
\begin{gather*}
\tilde u_{\tilde t}\DD_tT + \tilde u_{\tilde x}\DD_tX - \DD_tU = 0, \\
\tilde u_{\tilde t}\DD_xT + \tilde u_{\tilde x}\DD_xX - \DD_xU = 0, \\
\tilde u_{\tilde t\tilde t}(\DD_tT)^2 + 2\tilde u_{\tilde t\tilde x}(\DD_tX)(\DD_tT) + \tilde u_{\tilde x\tilde x}(\DD_tX)^2+\tilde u_{\tilde t}\DD_t^2T+\tilde u_{\tilde x}\DD_t^2X - \DD_t^2U = 0, \\
\tilde u_{\tilde t\tilde t}(\DD_xT)^2 + 2\tilde u_{\tilde t\tilde x}(\DD_xX)(\DD_xT) + \tilde u_{\tilde x\tilde x}(\DD_xX)^2+\tilde u_{\tilde t}\DD_x^2T+\tilde u_{\tilde x}\DD_x^2X - \DD_x^2U = 0,
\end{gather*}
cf.\ \cite{bihl2012b}.
Solving the last two equations for $u_{tt}$ and $u_{xx}$, respectively,
and substituting the derived expressions into~\eqref{eq:GenWaveEqs}, we obtain
\begin{align}\label{eq:TransformationOfGenWaveEqs}
\begin{split}
&\tilde u_{\tilde t\tilde t}(\DD_tT)^2+ 2\tilde u_{\tilde t\tilde x}(\DD_tT)(\DD_tX)+ \tilde u_{\tilde x\tilde x}(\DD_tX)^2
+\tilde u_{\tilde t}V^tT+\tilde u_{\tilde x}V^tX-V^tU\\
&{}= f\bigl(\tilde u_{\tilde t\tilde t}(\DD_xT)^2+2\tilde u_{\tilde t\tilde x}(\DD_xT)(\DD_xX) +\tilde u_{\tilde x\tilde x}(\DD_xX)^2
+\tilde u_{\tilde t}V^xT+\tilde u_{\tilde x}V^xX-V^xU\bigr)\\
&{}-g(\tilde u_{\tilde t} T_u+\tilde u_{\tilde x}X_u-U_u),
\end{split}
\end{align}
where we use the notation $V^t:=\p_{tt}+2u_t\p_{tu}+u_t^{\;2}\p_{uu}$ and $V^x:=\p_{xx}+2u_x\p_{xu}+u_x^{\ 2}\p_{uu}$.
The substitution $\tilde u_{\tilde t\tilde t}=\tilde f\tilde u_{\tilde x\tilde x} + \tilde g$
in view of~$\mathcal L_{\tilde\theta}$ into~\eqref{eq:TransformationOfGenWaveEqs}
wherever $\tilde u_{\tilde t\tilde t}$ occurs leads to
an identity with respect to~$\tilde u_{\tilde t\tilde x}$ and~$\tilde u_{\tilde x\tilde x}$.
In particular, we can collect the coefficients of $\tilde u_{\tilde t\tilde x}$ in \eqref{eq:TransformationOfGenWaveEqs},
which results~in
\begin{align}\label{eq:DetEqsForAdmTrans1}
(T_t+T_uu_t)(X_t+X_uu_t)=f(T_x+T_uu_x)(X_x+X_uu_x).
\end{align}
The equation~\eqref{eq:DetEqsForAdmTrans1} involves only original quantities (without tilde)
and is a polynomial in~$u_t$ and~$u_x$.
Therefore, we can split it with respect to~$u_t$ and~$u_x$ by collecting the coefficients of different powers of these derivatives.
As a result, we get
\begin{gather}
\label{eq:DetEqsForAdmTrans1ut2ux2}
T_uX_u=0,
\\ \label{eq:DetEqsForAdmTrans1ut1}
T_uX_t+T_tX_u=0,
\\ \label{eq:DetEqsForAdmTrans1ux1}
T_uX_x+T_xX_u=0,
\\ \label{eq:DetEqsForAdmTrans1ut0ux0}
T_tX_t = fT_xX_x.
\end{gather}
We multiply the equation~\eqref{eq:DetEqsForAdmTrans1ut1} by~$T_u$ (resp.\ $X_u$)
and, in view of the equation~\eqref{eq:DetEqsForAdmTrans1ut2ux2}, derive  that $T_uX_t=0$ (resp.\ $T_tX_u=0$).
We apply the same trick also to the equation~\eqref{eq:DetEqsForAdmTrans1ux1}
to have the equations $T_uX_x=0$ and (resp.\ $X_uT_x=0$).
The system $T_uX_t=0$, $T_uX_x=0$, $T_uX_u=0$ (resp.\ $X_uT_t=0$, $X_uT_x=0$, $X_uT_u=0$) implies that $T_u=0$ (resp.\ $X_u=0$)
since otherwise the Jacobian~$\mathrm J$ of the point transformation~\eqref{eq:GeneralEquivalenceTransformation} vanishes.
The condition
\[
T_u=X_u=0
\]
means that any admissible point transformation of the class~\eqref{eq:GenWaveEqs} is fiber-preserving.
In view of this condition,
the equations~\eqref{eq:DetEqsForAdmTrans1ut2ux2}--
\eqref{eq:DetEqsForAdmTrans1ux1} are identically satisfied,
and the splitting of~\eqref{eq:TransformationOfGenWaveEqs} with respect to $\tilde u_{\tilde x\tilde x}$
leads to the equations
\begin{gather}\label{eq:DetEqsForAdmTrans2}
\tilde f T_t^{\;2}+X_t^{\;2} = f(\tilde fT_x^{\;2}+X_x^{\;2}),
\\[1ex] \nonumber
\begin{split}
 &\tilde g T_t^2 + \tilde u_{\tilde t} T_{tt} + \tilde u_{\tilde x}X_{tt} -(U_{tt}+2U_{tu}u_t + U_{uu}u_t^2)\\&\qquad{}
= f(\tilde g T_x^2+\tilde u_{\tilde t}T_{xx} + \tilde u_{\tilde x}X_{xx}-(U_{xx}+2U_{xu}u_x+U_{uu}u_x^2))+ gU_u.
\end{split}
\end{gather}
Splitting the last equation with respect to~$\tilde u_{\tilde t}$ and~$\tilde u_{\tilde x}$ gives the equations
$U_{uu}=0$ and
\begin{gather}
\label{eq:DetEqsForAdmTrans4ut1}
T_{tt}-2\frac{U_{ut}}{U_u}T_t=f\left(T_{xx}-2\frac{U_{ux}}{U_u}T_x\right),
\\ \label{eq:DetEqsForAdmTrans4ux1}
X_{tt}-2\frac{U_{ut}}{U_u}X_t=f\left(X_{xx}-2\frac{U_{ux}}{U_u}X_x\right),
\\ \label{eq:DetEqsForAdmTrans4ut0ux0}
\tilde gT_t^2-U_{tt}+2\frac{U_{ut}}{U_u}U_t=f\left(\tilde gT_x^2-U_{xx}+2\frac{U_{ux}}{U_u}U_x\right)+gU_u.
\end{gather}
The equation~$U_{uu}=0$ 
implies the representation $U=U^1(t,x)u+U^0(t,x)$.
The additional condition to keep in mind is the nondegeneracy of~$\Phi$,
which in view of the conditions $T_u=X_u=0$ reduces to the inequality $U_u(T_tX_x-T_xX_t)\ne0$,
and hence $T_tX_x-T_xX_t\ne0$ and $U^1:=U_u\ne0$.

Note that the equations $T_u=X_u=U_{uu}=0$ for admissible point transformations within the class~$\mathcal W_{\rm gen}$ 
can also be derived using item~(c) of Theorem~4.4b in~\cite{king1998a}.

Rewriting the equation~\eqref{eq:DetEqsForAdmTrans2} as $\tilde f(T_t^{\;2}-fT_x^{\;2})+X_t^{\;2}-fX_x^{\;2}=0$
shows that both the expressions $T_t^{\;2}-fT_x^{\;2}$ and $X_t^{\;2}-fX_x^{\;2}$ are nonzero,
\[
T_t^{\;2}-fT_x^{\;2}\ne0, \quad X_t^{\;2}-fX_x^{\;2}\ne0.
\]
Indeed, otherwise $f>0$, $T_t=\ve_1f^{1/2}T_x\ne0$, where $\ve_1=\pm1$,
and hence the equation~\eqref{eq:DetEqsForAdmTrans1ut0ux0} would imply that $X_t=\ve_1f^{1/2}X_x$
but this contradicts the nondegeneracy of~$\Phi$.
Thus, $\tilde f_{\tilde u}=0$ if $f_u=0$.
Conversely, supposing $\tilde f_{\tilde u}=0$ and using the same argumentation for the inverse of~$\mathcal T$, 
we derive that $f_u=0$.
Therefore, $\tilde f_{\tilde u}=0$ if and only if $f_u=0$.

In view of the nonvanishing expression $T_t^{\;2}-fT_x^{\;2}$,
the equation~\eqref{eq:DetEqsForAdmTrans4ut0ux0} similarly implies that
$\tilde f_{\tilde u}=\tilde g_{\tilde u\tilde u}=0$ if and only if $f_u=g_{uu}=0$.

\begin{remark}\label{rem:OnInequivOfLinAndNonlinCasesOfGenWaveEqs}
Preserving the constraint $f_u=g_{uu}=0$ by admissible transformations 
of the class~$\mathcal W_{\rm gen}=\mathcal W\sqcup\mathcal W_{\rm lin}$
means that the equivalence groupoid~$\mathcal G^\sim_{\rm gen}$ of the class $\mathcal W_{\rm gen}$
is the disjoint union of the equivalence groupoids~$\mathcal G^\sim$ and~$\mathcal G^\sim_{\rm lin}$
of the subclasses~$\mathcal W$ and~$\mathcal W_{\rm lin}$,
$\mathcal G^\sim_{\rm gen}=\mathcal G^\sim\sqcup\mathcal G^\sim_{\rm lin}$.
In other words, equations from the class~$\mathcal W$
are not related to equations from the class~$\mathcal W_{\rm lin}$ by point transformations.
This justifies the exclusion of the class~$\mathcal W_{\rm lin}$ from the consideration,
which has been mentioned in the introduction.
\end{remark}

\section{Equivalence group and equivalence algebra}\label{sec:EquivGroup}

At this point, we continue the consideration
by computing the equivalence group of the class~\eqref{eq:GenWaveEqs}
as it is needed both for the exhaustive description of admissible transformations
and for the analysis of the determining equations for components of Lie-symmetry vector fields.
In the course of computing equivalence transformations,
the arbitrary elements~$f$ and~$g$ of the class should be varied.
We can therefore split
the equations~\eqref{eq:DetEqsForAdmTrans1ut0ux0}--\eqref{eq:DetEqsForAdmTrans4ut0ux0}
with respect to these arbitrary elements.
The equation~\eqref{eq:DetEqsForAdmTrans1ut0ux0} and the nondegeneracy constraint $T_tX_x-T_xX_t\ne0$ imply
that either $T_t=X_x=0$ and $T_xX_t\ne0$ or $T_x=X_t=0$ and $T_tX_x\ne0$.

For $T_t=X_x=0$, the equation~\eqref{eq:DetEqsForAdmTrans2} is simplified to $X_t^2=f\tilde fT_x^2$, or $T_x^2f=X_t^2/\tilde f$. 
Differentiating the last equation with respect to~$t$ 
and splitting the arising equation with respect to $\tilde f$ and its derivatives implies $X_t=0$, 
which contradicts the nondegeneracy condition.

Therefore, we necessarily have $X_t=T_x=0$ and thus $T=T(t)$, $X=X(x)$, where $T_tX_x\ne0$.
Then the equation~\eqref{eq:DetEqsForAdmTrans2} reduces to $\tilde fT_t^2=fX_x^2$
and the differentiation of this equation with respect to~$t$ yields
\begin{equation}\label{eq:DetEqsForAdmTrans2a}
 2T_tT_{tt}\tilde f + T_t^2U_t\tilde f_{\tilde u} = 0.
\end{equation}
Since the equation~\eqref{eq:DetEqsForAdmTrans2a} holds for all~$\tilde f$, 
we can split it with respect to $\tilde f$ and $\tilde f_{\tilde u}$ and derive $T_{tt}=0$ and $U_t=0$.
The equation~\eqref{eq:DetEqsForAdmTrans4ut1} is identically satisfied in view of the above equations.
The equation~\eqref{eq:DetEqsForAdmTrans4ux1} reduces to the equation $(U_u^2/X_x)_x=0$
and the equation~\eqref{eq:DetEqsForAdmTrans4ut0ux0} yields the transformation relation for~$g$.

Integrating the derived equations in view of the nondegeneracy condition $\mathrm J\ne0$, 
we prove the following theorem.

\begin{theorem}\label{thm:EquivalenceGroupIbragrimovClass}
The equivalence (pseudo)group~$G^\sim$ of the class~\eqref{eq:GenWaveEqs} consists of the transformations
\begin{align}\label{eq:EquivalenceGroupGenWaveEqs}
\begin{split}
 &\tilde t = c_1t+c_0, \quad \tilde x=\varphi(x), \quad \tilde u = c_2|\varphi_x|^{1/2}u+\psi(x), \\
 &\tilde f = \frac{\varphi_x^2}{c_1^2}f, \quad
  \tilde g = \frac{c_2}{c_1^2}|\varphi_x|^{1/2}g
   -\frac{1}{c_1^2}\left(c_2\frac{2\varphi_{xxx}\varphi_x-3\varphi_{xx}^{\,2}}{4|\varphi_x|^{3/2}}u
   +\psi_{xx}-\frac{\varphi_{xx}}{\varphi_x}\psi_x\right)f,
\end{split}
\end{align}
where $c_0$, $c_1$ and $c_2$ are arbitrary constants satisfying the condition $c_1c_2\ne0$,
and $\varphi$ and~$\psi$ run through the set of smooth functions of~$x$ with $\varphi_x\ne0$.
\end{theorem}

\begin{corollary}
The class~\eqref{eq:GenWaveEqs} admits exactly three discrete equivalence transformations
that are independent up to combining with each other and with continuous equivalence transformations of this class.
These are involutions alternating signs of variables,
$(t,x,u,f,g)\mapsto(-t,x,u,f,g)$, $(t,x,u,f,g)\mapsto(t,-x,u,f,g)$ and $(t,x,u,f,g)\mapsto(t,x,-u,f,-g)$.
\end{corollary}

In contrast to~\cite{song2009a}, we have proved that 
there are no other independent discrete equivalence transformations for the class~\eqref{eq:GenWaveEqs}.

Theorem~\ref{thm:EquivalenceGroupIbragrimovClass} implies that any transformation~$\mathscr T$ from~$G^\sim$ can be represented as the composition
\[
  \mathscr T = \mathscr P^t(c_0)\circ\mathscr D^t(c_1)\circ\mathscr Z(\psi)\circ\mathscr D(\varphi)\circ\mathscr D^u(c_2)
\]
of elementary equivalence transformations, each of which belongs to a family of equivalence transformations
parameterized by a single constant or functional parameter,
\[
\begin{array}{lllllll}
\mathscr P^t(c_0)  \colon\ & \tilde t=t+c_0,\ & \tilde x=x,      \ \quad& \tilde u=u,                  \  & \tilde f=f,               \ & \tilde g=g,\\
\mathscr D^t(c_1)  \colon\ & \tilde t=c_1t, \ & \tilde x=x,      \ \quad& \tilde u=u,                  \  & \tilde f=c_1^{-2}f,       \ & \tilde g=c_1^{-2}g,\\
\mathscr D^u(c_2)  \colon\ & \tilde t=t,    \ & \tilde x=x,      \ \quad& \tilde u=c_2u,               \  & \tilde f=f,               \ & \tilde g=c_2g,\\
\mathscr D(\varphi)\colon\ & \tilde t=t,    \ & \tilde x=\varphi,\ \quad& \tilde u=|\varphi_x|^{1/2}u, \  & \tilde f=\varphi_x^{\;2}f,\ & \tilde g=|\varphi_x|^{1/2}g+\alpha^\varphi(x) uf,\\
\mathscr Z(\psi)   \colon\ & \tilde t=t,    \ & \tilde x=x,      \ \quad& \tilde u=u+\psi,             \  & \tilde f=f,               \ & \tilde g=g-\psi_{xx}f,
\end{array}
\]
where the parameters are described in Theorem~\ref{thm:EquivalenceGroupIbragrimovClass}, and
\[
  \alpha^\varphi(x):= \frac{2\varphi_{xxx}\varphi_x-3\varphi_{xx}^{\;2}}{4|\varphi_x|^{3/2}}.
\]
These transformations are shifts and scalings in $t$, arbitrary transformations in $x$,
scalings of $u$ and shifts of $u$ with arbitrary functions of $x$, respectively.

The equivalence algebra~$\mathfrak g^\sim$ of the class~\eqref{eq:GenWaveEqs}
can be easily derived as the set of vector fields that generate local one-parametric subgroups of the equivalence group~$G^\sim$.
It is spanned by the vector fields
\begin{gather}\label{eq:EquivalenceAlgebraGenWaveEqs}
\begin{split}&
\PP^t=\p_t,\quad
\DDD^t=t\p_t-2f\p_f-2g\p_g,\quad
\DDD^u=u\p_u+g\p_g,\\&
\DDD(\zeta)=\zeta\p_x+\tfrac12\zeta_xu\p_u+2\zeta_xf\p_f+\tfrac12(\zeta_xg - \zeta_{xxx}u f)\p_g,\quad
\ZZ(\chi)=\chi\p_u-\chi_{xx}f\p_g,
\end{split}
\end{gather}
where~$\zeta=\zeta(x)$ and~$\chi=\chi(x)$ run through the set of smooth functions of~$x$.
The nonvanishing commutation relations between these vector fields are exhausted by
\begin{gather*}
[\PP^t,\DDD^t]=\PP^t,\quad [\ZZ(\chi),\DDD^u]=\ZZ(\chi), \\
[\DDD(\zeta^1),\DDD(\zeta^2)]=\DDD(\zeta^1\zeta^2_x-\zeta^1_x\zeta^2), \quad [\DDD(\zeta),\ZZ(\chi)]=\ZZ(\zeta\chi_x-\tfrac12\zeta_x\chi).
\end{gather*}


\section{Determining equations for Lie symmetries}\label{sec:DetEqsForLieSymsIbragrimovClass}

Suppose that a vector field $Q =\tau(t,x,u)\p_t+\xi(t,x,u)\p_x+\eta(t,x,u)\p_u$ 
belongs to the maximal Lie invariance algebra~$\mathfrak g^{\max}$ of an equation $\mathcal L$: $L=0$ from the class~\eqref{eq:GenWaveEqs}, 
i.e., it is the generator of a one-parameter Lie-symmetry group of~$\mathcal L$.
The criterion for infinitesimal invariance of $\mathcal L$ with respect to $Q$~\cite{blum1989A,olve1986A,ovsi1982A} implies that $Q_{(2)}L|_{L=0}=0$,
where the notation $|_{L=0}$ means that the condition $Q_{(2)}L$ is required to hold only on solutions of $\mathcal L$ and
$Q_{(2)}$ is the second prolongation of $Q$,
$
 Q_{(2)} = Q + \eta^t\p_{u_t} + \eta^x\p_{u_x} + \eta^{tt}\p_{u_{tt}} + \eta^{tx}\p_{u_{tx}}+\eta^{xx}\p_{u_{xx}}.
$
The coefficients $\eta^t$, $\eta^x$, $\eta^{tt}$, $\eta^{xx}$ in $Q_{(2)}$ can be determined from the general prolongation formula for vector fields. In particular,
\begin{gather*}
\eta^{xx} = \DD_x^2(\eta-\tau u_t-\xi u_x) + \tau u_{txx}+\xi u_{xxx}, \\
\eta^{tt} = \DD_t^2(\eta-\tau u_t-\xi u_x) + \tau u_{ttt}+\xi u_{ttx},
\end{gather*}
where
$\DD_t$ and $\DD_x$ denote the total derivative operators with respect to $t$ and $x$, respectively,
which in the present case of two independent and one dependent variables are given by
\begin{gather*}
 \DD_t = \p_t + u_t\p_{u}+u_{tt}\p_{u_{t}} + u_{tx}\p_{u_{x}}+\cdots, \\
 \DD_x = \p_x + u_x\p_{u}+u_{tx}\p_{u_{t}} + u_{xx}\p_{u_{x}}+\cdots.
\end{gather*}
Applying the infinitesimal invariance condition to the class~\eqref{eq:GenWaveEqs} then yields
\begin{equation}\label{eq:InfinitesimalInvarianceCriterionGenWaveEqs}
 \eta^{tt}-(\xi f_x+\eta f_u)u_{xx}-f\eta^{xx}-\xi g_x - \eta g_u =0 \qquad \textup{for} \qquad u_{tt}=fu_{xx}+g.
\end{equation}
It follows from the equations $T_u=X_u=U_{uu}=0$ for admissible transformation within the class~\eqref{eq:GenWaveEqs} that
the  vector field $Q$ is projectable to the space of independent variables and affine in~$u$,
i.e., $\tau_u=\xi_u=\eta_{uu}=0$ or $\tau=\tau(t,x)$, $\xi=\xi(t,x)$ and $\eta=\eta^1(t,x)u+\eta^0(t,x)$.
Taking into account these restrictions and expanding the infinitesimal invariance condition~\eqref{eq:InfinitesimalInvarianceCriterionGenWaveEqs},
we obtain the equation
\begin{gather}\label{eq:InfinitesimalInvarianceCriterionGenWaveEqsExpanded}
\begin{split}
 &   \DD_t^2\eta - \tau_{tt}u_t - \xi_{tt}u_x - 2\tau_tu_{tt} - 2\xi_tu_{tx}\\
 &=f(\DD_x^2\eta - \tau_{xx}u_t - \xi_{xx}u_x - 2\tau_xu_{tx} - 2\xi_xu_{xx}) + (\xi f_x + \eta f_u)u_{xx} + \xi g_x + \eta g_u,
\end{split}
\end{gather}
where we have to substitute $u_{tt}=fu_{xx}+g$.
Collecting the coefficients of the derivatives~$u_{tx}$, $u_{xx}$, $u_t$ and $u_x$ in the above equation results in the system of determining equations
\begin{gather}\label{eq:DetEqForLieSymsOfGenWaveEqs1}
\xi_t = \tau_xf,
\\\label{eq:DetEqForLieSymsOfGenWaveEqs2}
\tau_{tt}-\tau_{xx}f = 2\eta_{tu},
\\\label{eq:DetEqForLieSymsOfGenWaveEqs3}
\xi_{tt}-\xi_{xx}f = -2\eta_{xu}f,
\\\label{eq:DetEqForLieSymsOfGenWaveEqs4}
\xi f_x+\eta f_u=2(\xi_x-\tau_t)f,
\\\label{eq:DetEqForLieSymsOfGenWaveEqs5}
\xi g_x + \eta g_u = (\eta_u-2\tau_t)g - \eta_{xx}f + \eta_{tt}.
\end{gather}
In order to derive the kernel algebra~$\mathfrak g^\cap$ of the class~\eqref{eq:GenWaveEqs},
 we further split the determining equations with respect to the arbitrary elements and their derivatives.
This immediately gives that
\[\label{eq:KernelGenWaveEqs}
  \mathfrak g^\cap =\langle\p_t\rangle.
\]
Consequently, the Lie symmetries admitted by each equation from the class~\eqref{eq:GenWaveEqs}
are exhausted by transformations of the form $(t,x,u)\mapsto(t+c_0,x,u)$, where $c_0$ is an arbitrary constant.


\section{Results of classifications}
\label{sec:GenWaveEqsResultsOfClassifications}

For convenience, we collect, in a single table, the Lie-symmetry classification cases derived below
and formulate the final result of group classification of the class~\eqref{eq:GenWaveEqs} as a theorem.

\begin{theorem}\label{thm:GenWaveEqsGroupClassification}
All  $G^\sim$-inequivalent (resp.\ point-inequivalent) cases of Lie-symmetry extensions 
of the kernel algebra~$\mathfrak g^\cap=\langle\p_t\rangle$ in the class~\eqref{eq:GenWaveEqs}
are exhausted by cases presented in Table~\ref{tab:GenWaveEqsExtensions}.
\end{theorem}

\begin{table}[ptb]\footnotesize\setcounter{tbn}{0}
\renewcommand{\arraystretch}{1.4}\belowcaptionskip=.7ex\abovecaptionskip=.0ex
{\centering
\caption{$G^\sim$-inequivalent Lie-symmetry extensions of~$\mathfrak g^\cap=\langle\p_t\rangle$ for the class~\eqref{eq:GenWaveEqs}\label{tab:GenWaveEqsExtensions}}
\begin{tabular}{|c|l|l|l|}
\hline
N\hfil &\hfil $f$ &\hfil $g$ &\hfil Basis of extension \\
\hline
\cn{case1}    & $\hat f(\omega)|u|^p$ & $\hat g(\omega)|u|^pu$   & $-pt\p_t + 2\delta\p_x+2u\p_u$ \\
\cn{case2}    & $\hat f(u){\rm e}^x$  & $\hat g(u){\rm e}^x$     & $t\p_t-2\p_x$ \\
\cn{case3}    & $\hat f(x){\rm e}^u$  & $\hat g(x){\rm e}^u$     & $t\p_t - 2\p_u$ \\
\cn{case4}    & $\hat f(u) $          & $\hat g(u)$              & $\p_x$ \\
\hline
\ca{case5a}  & $\ve$                 & $\hat g(u)$              & $\p_x$, $x\p_t+\ve t\p_x$ \\
\cb{case5b}  & $1$                   & $\hat g(u){\rm e}^{-2x}$ & $\RR({\rm e}^{x+t})$, $\RR({\rm e}^{x-t})$ \\
\cb{case5c}  & $-1$                  & $\hat g(u){\rm e}^{-2x}$ & $\RR({\rm e}^x\cos t)$, $\RR({\rm e}^x\sin t)$ \\
\ca{case6a}  & $\ve$                 & $\hat g(u)x^{-2}$        & $t\p_t+x\p_x$, $(t^2+\ve x^2)\p_t+2tx\p_x$ \\
\cb{case6b}  & $1$                   & $\hat g(u)\cos^{-2}x$    & $\RR(\cos t\,\cos x)$, $\RR(\sin t\,\cos x)$ \\
\cb{case6c}  & $1$                   & $-\hat g(u)\cosh^{-2}x$  & $\RR({\rm e}^t\cosh x)$, $\RR({\rm e}^{-t}\cosh x)$ \\
\cb{case6d}  & $1$                   & $\hat g(u)\sinh^{-2}x$   & $\RR({\rm e}^t\sinh x)$, $\RR({\rm e}^{-t}\sinh x)$ \\
\cb{case6e}  & $-1$                  & $\hat g(u)\cos^{-2}x$    & $\RR({\rm e}^t\cos x)$, $\RR({\rm e}^{-t}\cos x)$ \\
\cb{case6f}  & $-1$                  & $\hat g(u)\sinh^{-2}x$   & $\RR(\cos t\,\sinh x)$, $\RR(\sin t\,\sinh x)$ \\
\cn{case7}   & $-1$                  & $\hat g(u)\cosh^{-2}x$   & $\RR(\cos t\,\cosh x)$, $\RR(\sin t\,\cosh x)$ \\
\ca{case8a}  & $\ve u^{-4}$          & $\mu(x)u^{-3}$           & $2t\p_t+u\p_u$, $t^2\p_t+tu\p_u$ \\
\cb{case8b}  & $\ve u^{-4}$          & $\mu(x)u^{-3}+u$         & ${\rm e}^{2t}(\p_t+u\p_u)$, ${\rm e}^{-2t}(\p_t-u\p_u)$ \\
\cb{case8c}  & $\ve u^{-4}$          & $\mu(x)u^{-3}-u$         & $\cos(2t)\p_t-\sin(2t)u\p_u$, $\sin(2t)\p_t+\cos(2t)u\p_u$ \\
\cn{case9}   & $\ve {\rm e}^x|u|^p$  & $\nu {\rm e}^x|u|^pu$    & $p\p_x-u\p_u$, $t\p_t-2\p_x$ \\
\cn{case10}  & $\ve x^2{\rm e}^u$    & $\nu {\rm e}^u$          & $x\p_x$, $t\p_t-2\p_u$ \\
\cn{case11}  & $\hat f(u)$           & $0$                      & $\p_x$, $t\p_t+x\p_x$ \\
\cn{case12}  & $\ve {\rm e}^u$       & $\ve' {\rm e}^{qu}$      & $\p_x$, $qt\p_t+(q-1)x\p_x-2\p_u$ \\
\cn{case13}  & $\ve |u|^p$           & $\ve'|u|^q$              & $\p_x$, $(1-q)t\p_t+(1+p-q)x\p_x+2u\p_u$ \\
\hline
\ca{case14a} & $\ve u^{-4}$          & $\ve' u^{-3}$            & $2t\p_t+u\p_u$, $t^2\p_t+tu\p_u$, $\p_x$ \\
\cb{case14b} & $\ve u^{-4}$          & $\ve' u^{-3}+u$          & ${\rm e}^{2t}(\p_t+u\p_u)$, ${\rm e}^{-2t}(\p_t-u\p_u)$, $\p_x$ \\
\cb{case14c} & $\ve u^{-4}$          & $\ve' u^{-3}-u$          & $\cos(2t)\p_t-\sin(2t)u\p_u$, $\sin(2t)\p_t+\cos(2t)u\p_u$, $\p_x$ \\
\cb{case14d} & $\ve u^4$             & $\ve' u$                 & $\p_x$, $2x\p_x+u\p_u$, $x^2\p_x+xu\p_u$ \\
\ca{case15a} & $\ve u^{-4}$          & $\nu x^{-2}u^{-3}$       & $2t\p_t+u\p_u$, $t^2\p_t+tu\p_u$, $2x\p_x-u\p_u$ \\
\cb{case15b} & $\ve u^{-4}$          & $\nu x^{-2}u^{-3}+u$     & ${\rm e}^{2t}(\p_t+u\p_u)$, ${\rm e}^{-2t}(\p_t-u\p_u)$, $2x\p_x-u\p_u$ \\
\cb{case15c} & $\ve u^{-4}$          & $\nu x^{-2}u^{-3}-u$     & $\cos(2t)\p_t-\sin(2t)u\p_u$, $\sin(2t)\p_t+\cos(2t)u\p_u$, $2x\p_x-u\p_u$ \\
\cn{case16}  & $\ve |u|^p$           & $0$                      & $\p_x$, $t\p_t+x\p_x$, $px\p_x+2u\p_u$ \\
\cn{case17}  & $\ve {\rm e}^u$       & $0$                      & $\p_x$, $t\p_t+x\p_x$, $x\p_x+2\p_u$ \\
\ca{case18a} & $\ve$                 & $\ve'|u|^q$              & $\p_x$, $t\p_x+\ve x\p_t$, $(q-1)t\p_t+(q-1)x\p_x-2u\p_u$ \\
\cb{case18b} & $1$                   & $\ve'|u|^q{\rm e}^{-2x}$ & $\RR({\rm e}^{x+t})$, $\RR({\rm e}^{x-t})$, $(q-1)\p_x+2u\p_u$\\
\cb{case18c} & $-1$                  & $\ve'|u|^q{\rm e}^{-2x}$ & $\RR({\rm e}^x\cos t)$, $\RR({\rm e}^x\sin t)$, $(q-1)\p_x+2u\p_u$\\
\hline
\ca{case19a} & $\ve u^{-4}$          & $0$                      & $2t\p_t+u\p_u$, $t^2\p_t+tu\p_u$, $\p_x$, $2x\p_x-u\p_u$ \\
\cb{case19b} & $\ve u^{-4}$          & $u$                      & ${\rm e}^{2t}(\p_t+u\p_u)$, ${\rm e}^{-2t}(\p_t-u\p_u)$, $\p_x$, $2x\p_x-u\p_u$ \\
\cb{case19c} & $\ve u^{-4}$          & $-u$                     & $\cos(2t)\p_t-\sin(2t)u\p_u$, $\sin(2t)\p_t+\cos(2t)u\p_u$, $\p_x$, $2x\p_x-u\p_u$ \\
\cb{case19d} & $\ve u^4$             & $0$                      & $\p_x$, $t\p_t+x\p_x$, $2x\p_x+u\p_u$, $x^2\p_x+xu\p_u$ \\ 
\hline
\cn{case20}  & $\ve$                 & $\ve' {\rm e}^u$         & $\tau\p_t+\xi\p_x-2\tau_t\p_u$ \\
\hline
\end{tabular}
\\[2ex]}
Here $\ve,\ve'=\pm1\bmod G^\sim$, $\delta\in\{0,1\}\bmod G^\sim$, 
$p$, $q$ and~$\nu$ are arbitrary constants with $p\ne0$ and $\nu\ne0$.
Additionally, $p\ne\pm4$ in Case~\ref{case16}, 
and $q\ne0,1$ in Cases~\ref{case18a}--\ref{case18c}. 
In~Case~1, $\omega:=x-\delta\ln|u|$.
$\RR(\Phi):=\Phi_x\p_t+\Phi_t\p_x$. 
The tuple $(\tau,\xi)$ of smooth functions depending on $(t,x)$ 
runs through the solution set of the system $\tau_t=\xi_x$, $\xi_t=\ve\tau_x$. 
Each of the functions $\hat f$, $\hat g$ and $\mu$ is an arbitrary smooth function of a single argument 
that satisfied the conditions needed for the maximality of the corresponding invariance algebra, 
see the paragraph after Theorem~\ref{thm:GenWaveEqsGroupClassification}.
\end{table}

In each case of Table~\ref{tab:GenWaveEqsExtensions} we present
only vector fields which extend the basis $(\p_t)$ of~$\mathfrak g^\cap$
into a basis of the corresponding Lie invariance algebra.
The spans of~$\mathfrak g^\cap$ and the vector fields given in cases of Table~\ref{tab:GenWaveEqsExtensions}
that are parameterized by functions~$\hat f$, $\hat g$ or~$\mu$
are the maximal Lie invariance algebras of the corresponding equations
for the general values of these parameter functions,
but for certain their values additional extensions are possible,
which are equivalent to other cases of Table~\ref{tab:GenWaveEqsExtensions}. 
Thus, $\hat f\ne\const$ in Case~\ref{case4} since otherwise up to $G^\sim$-equivalence we obtain Case~\ref{case5a}.
There are also constraints for constant parameters 
that are imposed by the condition of inequivalence of the corresponding extensions 
or the condition of their maximality.
                                                                                  
Depending on the dimension of Lie-symmetry extension (one, two, three, four or infinity), 
we split the cases of Table~\ref{tab:GenWaveEqsExtensions} into groups separated by horizontal lines.
Note that all Lie-symmetry extensions of maximal and submaximal dimensions (infinity and four) 
for equations from the class~\eqref{eq:GenWaveEqs} 
are not associated with subalgebras of the equivalence algebra~$\mathfrak g^\sim$, 
i.e., they are singular.

The usage of two-level numeration for classification cases listed in Table~\ref{tab:GenWaveEqsExtensions}
is justified by the presence of additional equivalences among them. 
Namely, numbers with the same Arabic numerals and different Roman letters
correspond to cases that are $G^\sim$-inequivalent
but equivalent with respect to additional equivalence transformations. 
To construct all additional equivalence transformations among $G^\sim$-inequivalent classification cases 
and thus to solve the group classification problem for the class~\eqref{eq:GenWaveEqs} 
up to $\mathcal G^\sim$-equivalence, 
we need to classify admissible transformations of this class up to $G^\sim$-equivalence.
This classification is presented in the following theorem, 
which is proved in Section~\ref{sec:GenWaveEqsEquivGroupoid}.

\begin{theorem}\label{thm:ClassWAdmTrans}
A generating (up to $G^\sim$-equivalence) set~$\mathcal B$ of admissible transformations for the class~$\mathcal W$, 
which is minimal and self-consistent with respect to $G^\sim$-equivalence, 
is the union of the following families of admissible transformations, where $\ve,\ve',\ve''=\pm1$:

\begin{enumerate}
\renewcommand{\theenumi}{{\rm T\arabic{enumi}}}
\renewcommand{\labelenumii}{{\rm\alph{enumii}.}}

\item\label{T1} 
$f_x=g_x=0$, \ $f_u\ne0$ or $f=1$, \ $\tilde f=1/f$, \ $\tilde g=-g/f$, \ $\Phi{:}$ \ $\tilde t=x$, \ $\tilde x=t$, \ $\tilde u=u$;

\item\label{T2} 
$f=\ve u^{-4}$, \ $g=\mu(x)u^{-3}+\sigma u$, \ $\sigma\in\{-1,0,1\}$, \ 
$\tilde f=\ve\tilde u^{-4}$, \ $\tilde g=\mu(\tilde x)\tilde u^{-3}$, \ 
$\mu$ runs through the set of smooth functions of~$x$ with $\mu_x\ne0$, 
  \begin{enumerate}
  \item\label{T2a}
  $\Phi{:}$ \ $\tilde t=t^{-1}$, \ $\tilde x=x$, \ $\tilde u=t^{-1}u$\quad if\quad $\sigma=0$;
  \item\label{T2b} 
  $\Phi{:}$ \ $\tilde t=\frac12{\rm e}^{2t}$, \ $\tilde x=x$, \ $\tilde u={\rm e}^t u$\quad if\quad $\sigma=1$;
  \item\label{T2c} 
  $\Phi{:}$ \ $\tilde t=\tan t$, \ $\tilde x=x$, \ $\tilde u=u\cos t$\quad if\quad $\sigma=-1$;
  \end{enumerate}

\item\label{T3} 
$f=1$, \ $g={\rm e}^{-2x}g^2(u)$, \ $\tilde f=1$, \ $\tilde g=g^2(\tilde u)$,\ \
$\Phi{:}$ \ $\tilde t={\rm e}^{-x}\sinh t$, \ $\tilde x={\rm e}^{-x}\cosh t$, \ $\tilde u=u$;

\item\label{T4} 
$f=1$, \ $g=g^1(x)g^2(u)$, \ 
$\tilde f=1$, \ $\tilde g=\tilde x^{-2}g^2(\tilde u)$,\vspace{-1ex}
  \begin{enumerate}
  \item\label{T4a} 
  $g^1(x)=x^{-2}$, \ $\Phi{:}$ \ $\tilde t=\dfrac t{x^2-t^2}$, \ $\tilde x=\dfrac x{x^2-t^2}$, \ $\tilde u=u$;
  \item\label{T4b} 
  $g^1(x)=\cos^{-2}x$, 
  $\Phi{:}$ \ $\tilde t=\dfrac{\cos t}{\sin t+\sin x}$, \ $\tilde x=\dfrac{\cos x}{\sin t+\sin x}$, \ $\tilde u=u$;\vspace{.5ex}
  \item\label{T4c} 
  $g^1(x)=-\cosh^{-2}x$, 
  $\Phi{:}$ \ $\tilde t={\rm e}^t\sinh x$, \ $\tilde x={\rm e}^t\cosh x$, \ $\tilde u=u$;\vspace{1ex}
  \item\label{T4d} 
  $g^1(x)=\sinh^{-2}x$, 
  $\Phi{:}$ \ $\tilde t={\rm e}^t\cosh x$, \ $\tilde x={\rm e}^t\sinh x$, \ $\tilde u=u$;
  \end{enumerate}

\item\label{T5} 
$f=-1$, \ $g={\rm e}^{-2x}g^2(u)$, \ $\tilde f=-1$, \ $\tilde g=g^2(\tilde u)$,\ \
$\Phi{:}$ \ $\tilde t={\rm e}^{-x}\sin t$, \ $\tilde x={\rm e}^{-x}\cos t$, \ $\tilde u=u$;

\item\label{T6} 
$f=-1$, \ $g=g^1(x)g^2(u)$, \ 
$\tilde f=-1$, \ $\tilde g=\tilde x^{-2}g^2(\tilde u)$,\vspace{-1ex}
  \begin{enumerate}
  \item\label{T6a} 
  $g^1(x)=x^{-2}$, \ $\Phi{:}$ \ $\tilde t=\dfrac t{x^2+t^2}$, \ $\tilde x=\dfrac x{x^2+t^2}$, \ $\tilde u=u$;\vspace{.5ex}
  \item\label{T6b} 
  $g^1(x)=\cos^{-2}x$, 
  $\Phi{:}$ \ $\tilde t={\rm e}^t\sin x$, \ $\tilde x={\rm e}^t\cos x$, \ $\tilde u=u$;\vspace{.5ex}
  \item\label{T6c} 
  $g^1(x)=\sinh^{-2}x$, \ $\Phi{:}$ \ $\tilde t=\dfrac{\sin t}{\cos t+\cosh x}$, \ $\tilde x=\dfrac{\sinh x}{\cos t+\cosh x}$, \ $\tilde u=u$;
  \end{enumerate}

\item\label{T7} 
$f=-1$, \ $g=g^2(u)\cosh^{-2}x$, \ $\tilde f=-1$, \ $\tilde g=g^2(\tilde u)\cosh^{-2}\tilde x$,\\[1.5ex]
$\Phi{:}$ \ 
$\tilde t=\arctan\dfrac{\sin\gamma\,\sinh x+\cos\gamma\sin t}{\cos t}$, \ 
$\tilde x=\mathop{\rm arctanh}\dfrac{\cos\gamma\,\sinh x-\sin\gamma\sin t}{\cosh x}$, \ $\tilde u=u$,\\[.5ex] $\gamma\in(0,2\pi)$;

\item\label{T8} 
  \begin{enumerate}
  \item\label{T8a} 
  $f=\tilde f=1$, \ $g_x=0$, \ $\tilde g=g$,\\ 
  $\Phi{:}$ \ $\tilde t=t\cosh\gamma+x\sinh\gamma$, \ $\tilde x=t\sinh\gamma+x\cosh\gamma$, $\tilde u=u$, \ $\gamma\in\mathbb R_{\ne0}$;
  \item\label{T8b} 
  $f=\tilde f=-1$, \ $g_x=0$, \ $\tilde g=g$, \ 
  $\Phi{:}$ \ $\tilde t=t\cos\gamma-x\sin\gamma$, \ $\tilde x=t\sin\gamma+x\cos\gamma$, \ $\tilde u=u$,\\ $\gamma\in(0,2\pi)$;
  \end{enumerate}

\item\label{T9} 
$f=\ve$, \ $g=\ve'{\rm e}^u$, \ $\tilde f=\ve$, \ $\tilde g=\ve'{\rm e}^{\tilde u}$, \ 
$\Phi{:}$ \ $\tilde t=T(t,x)$, \ $\tilde x=X(t,x)$, \ $\tilde u=u+\ln|T_t^{\,2}-\ve T_x^{\,2}|$, 
where $\ve,\ve'=\pm1$, $\ve'=1$ if $\ve=1$,
$(T,X)$ runs through a complete set of representatives of solution cosets of the system $T_t=X_x$, $X_t=\ve T_x$ with $(T_{tt},T_x)\ne(0,0)$
with respect to the action of the group constituted by the transformations of the form 
$\hat t=c_1t+c_0$, $\hat x=c_1x+c_2$, $\hat T=\tilde c_1T+\tilde c_0$, $\hat X=\tilde c_1X+\tilde c_2$, 
where $c_0$, $c_1$, $c_2$, $\tilde c_0$, $\tilde c_1$ and~$\tilde c_2$ are arbitrary constants with $c_1\tilde c_1\ne0$.  

\end{enumerate}
\end{theorem}

Throughout the rest of the paper, we use the notation T$N$ 
for the admissible transformation given in item~$N$ of Theorem~\ref{thm:ClassWAdmTrans},
where $N$ is the corresponding (one- or two-level) label. 

\begin{remark}
The transformational part~$\Phi$ of admissible transformation~\ref{T4b} can be represented~as 
\begin{enumerate}
\item[\ref{T4b}:] 
$\tilde t=\cot\dfrac{x+t}2+\tan\dfrac{x-t}2$, \ $\tilde x=\cot\dfrac{x+t}2-\tan\dfrac{x-t}2$, \ $\tilde u=u$.
\end{enumerate}
The transformational parts~$\Phi$ of admissible transformations~\ref{T4c} and~\ref{T4d} can be replaced by alternative ones, 
which are analogous to that of~\ref{T4b},
\begin{enumerate}
\item[\ref{T4c}:] 
$\tilde t=\coth\dfrac{x+t}2-\tanh\dfrac{x-t}2$, \ $\tilde x=\coth\dfrac{x+t}2+\tanh\dfrac{x-t}2$, \ $\tilde u=u$;
\item[\ref{T4d}:] 
$\tilde t=\tanh\dfrac{x+t}2-\tanh\dfrac{x-t}2$, \ $\tilde x=\tanh\dfrac{x+t}2+\tanh\dfrac{x-t}2$, \ $\tilde u=u$.
\end{enumerate}
There also exist similar alternatives for transformational parts of other admissible transformations in the class~$\mathcal W$. 
The counterparts of modified admissible transformations \ref{T4b}--\ref{T4d} and of admissible transformation \ref{T3} 
for linear equations from the class~$\mathcal W_{\rm lin}$ were presented in Notes~1 and~2 of~\cite{zhda1993a}.
\end{remark}

As a result, we obtain the following independent additional equivalence transformations 
among classification cases given in Table~\ref{tab:GenWaveEqsExtensions}.
(Below we do not indicate the corresponding parameters if they are not changed.)

\begin{enumerate}
\item[\ref{T1}:] (a) \ 
\ref{case4}$_{\hat f,\hat g}$ $\to$ \ref{case4}$_{1/\hat f,-\hat g/\hat f}$, \ 
\ref{case5a}$_{\hat g}$ $\to$ \ref{case5a}$_{-\hat g}$ if $\ve=1$, \
\ref{case11}$_{\hat f}$ $\to$ \ref{case11}$_{1/\hat f}$, \
\ref{case12}$_{q,\ve,\ve'}\circ(u\to-u)$ $\to$ \ref{case12}$_{1-q,\ve,-\ve\ve'}$,\\[1ex] 
\ref{case13}$_{p,q,\ve,\ve'}$ $\to$ \ref{case13}$_{-p,q-p,\ve,-\ve\ve'}$, \ 
\ref{case14d}$_{\ve,\ve'}$ $\to$ \ref{case14a}$_{\ve,-\ve\ve'}$, \ 
\ref{case16}$_{p,\ve}$ $\to$ \ref{case16}$_{-p,\ve}$, \ 
\ref{case18a}$_{q,\ve,\ve'}$ $\to$ \ref{case18a}$_{q,\ve,-\ve\ve'}$, \\[1ex]
\ref{case19d} $\to$ \ref{case19a}, \ 
\ref{case20}$_{\ve,\ve'}$ $\to$ \ref{case20}$_{\ve,-\ve\ve'}$. 
\item[\ref{T2}:] 
(b) \ 
\ref{case8b}  $\to$ \ref{case8a}, \ 
\ref{case14b} $\to$ \ref{case14a}, \ 
\ref{case15b} $\to$ \ref{case15a}, \ 
\ref{case19b} $\to$ \ref{case19a};
\\
(c) \ 
\ref{case8c}  $\to$ \ref{case8a}, \ 
\ref{case14c} $\to$ \ref{case14a}, \ 
\ref{case15c} $\to$ \ref{case15a}, \ 
\ref{case19c} $\to$ \ref{case19a}.
\item[\ref{T3}:] 
\ref{case5b}  $\to$ \ref{case5a}$_{\ve=1},$ \ 
\ref{case18b}  $\to$ \ref{case18a}$_{\ve=1}.$ 
\quad \ref{T4}: \ 
(b) \ \ref{case6b} $\to$ \ref{case6a}$_{\ve=1}$, \ 
(c) \ \ref{case6c} $\to$ \ref{case6a}$_{\ve=1}$, \ 
(d) \ \ref{case6d} $\to$ \ref{case6a}$_{\ve=1}$. 
\item[\ref{T5}:] 
\ref{case5c} $\to$ \ref{case5a}$_{\ve=-1},$ \ 
\ref{case18c} $\to$ \ref{case18a}$_{\ve=-1}.$ 
\quad \ref{T6}: \ 
(b) \ \ref{case6e} $\to$ \ref{case6a}$_{\ve=-1}$, \ 
(c) \ \ref{case6f} $\to$ \ref{case6a}$_{\ve=-1}$.
\end{enumerate}

\begin{remark}\label{rem:GenWaveClassRegularCaseEquivToSingularCases}
In Table~\ref{tab:GenWaveEqsExtensions},
only Cases \ref{case1}--\ref{case4}, \ref{case9}--\ref{case13}, \ref{case14d}, \ref{case16}, \ref{case17} and \ref{case19d} 
present regular Lie-symmetry extensions in the class~$\mathcal W$.
Therefore, the regular Cases~\ref{case14d} and~\ref{case19d} are $G^\sim$-inequivalent but $\mathcal G^\sim$-equivalent 
to the singular Cases~\ref{case14a} and~\ref{case19a}, respectively. 
\end{remark}

Consider the subclass~$\mathcal W_{\rm c}$ of the class~$\mathcal W$ 
singled out by the additional constraints $f_x=g_x=0$ for the arbitrary-element tuple $\theta=(f,g)$,
i.e., the class of equations of the general form
\begin{equation}\label{eq:GenWaveClassConstCoeffSubclass}
u_{tt}=f(u)u_{xx}+g(u),
\end{equation}
where $(f_{u},g_{uu})\ne(0,0)$.
Cases~\ref{case4}, \ref{case5a}, \ref{case11}, \ref{case12}, \ref{case13}, 
\ref{case14a}--\ref{case14d}, \ref{case16}, \ref{case17}, \ref{case18a}, 
\ref{case19a}--\ref{case19d} and \ref{case20}
of Table~\ref{tab:GenWaveEqsExtensions} 
are related to the subclass~$\mathcal W_{\rm c}$.
The kernel Lie invariance algebra~$\mathfrak g_{\rm c}=\langle\p_t,\p_x\rangle$ 
of equations from the subclass~$\mathcal W_{\rm c}$ 
is given by Case~\ref{case4}, which is the general case within this subclass. 
It is obvious from Theorem~\ref{thm:GenWaveEqsGroupClassification} 
and the above list of additional equivalence transformations 
that a complete list of~$\mathcal G^\sim_{\rm c}$-inequivalent Lie-symmetry extensions 
within the subclass~$\mathcal W_{\rm c}$, where $\mathcal G^\sim_{\rm c}$ is its equivalence groupoid,
is exhausted by Cases~\ref{case5a}, \ref{case11}, \ref{case12}, \ref{case13}, 
\ref{case14a}, \ref{case16}, \ref{case17}, \ref{case18a}, \ref{case19a} and \ref{case20}
of Table~\ref{tab:GenWaveEqsExtensions}, where additionally 
$q>1/2$ in Case~\ref{case12}, 
$p>0$ in Cases~\ref{case13} and~\ref{case16}, 
and $\ve'=1$ in Cases~\ref{case18a} and~\ref{case20} with $\ve=1$.
The group classification of the subclass~$\mathcal W_{\rm c}$ up to equivalence 
generated by its equivalence group $G^\sim_{\rm c}$ is more delicate. 
The group $G^\sim_{\rm c}$ is generated by transformations of the form~\eqref{eq:EquivalenceGroupGenWaveEqs}
with $\varphi_{xx}=\psi_x=0$, which constitute the intersection $G^\sim_{\rm c}\cap G^\sim$, 
and one more (discrete) equivalence transformation 
$\tilde t=x$, $\tilde x=t$, $\tilde u=u$, $\tilde f=1/f$, $\tilde g=-f/g$ of~$\mathcal W_{\rm c}$, 
which generates the family \ref{T1} of admissible transformations within~$\mathcal W$. 
This is why some $G^\sim$-inequivalent Lie-symmetry extensions can be $G^\sim_{\rm c}$-equivalent, 
which occurs for Cases~\ref{case14a} and~\ref{case14d} as well as for Cases~\ref{case19a} and~\ref{case19d}. 
The converse situation is not possible since the subgroupoid of~$\mathcal G^\sim_{\rm c}$ 
generated by~$G^\sim$ is contained in the subgroupoid generated by~$G^\sim_{\rm c}$.
This results in the following assertion. 

\begin{corollary}
A complete list of~$G^\sim_{\rm c}$-inequivalent Lie-symmetry extensions 
within the subclass~$\mathcal W_{\rm c}$ is exhausted by 
Cases~\ref{case5a}, \ref{case11}, \ref{case12}, \ref{case13}, 
\ref{case14a}--\ref{case14c}, \ref{case16}, \ref{case17}, \ref{case18a}, \ref{case19a}--\ref{case19c} and \ref{case20}
of Table~\ref{tab:GenWaveEqsExtensions}, where additionally 
$q>1/2$ in Case~\ref{case12}, 
$p>0$ in Cases~\ref{case13} and~\ref{case16}, 
and $\ve'=1$ in Cases~\ref{case18a} and~\ref{case20} with $\ve=1$.
\end{corollary}

\section{Equivalence groupoid and singular Lie-symmetry extensions}\label{sec:GenWaveEqsEquivGroupoid}

The equivalence groupoid~$\mathcal G^\sim$ of the class~$\mathcal W$
contains admissible transformations that are not generated by elements of~$G^\sim$,
i.e., this class is not normalized.
Nevertheless, we can describe the groupoid~$\mathcal G^\sim$,
classifying admissible transformations within the class~$\mathcal W$
up to the $G^\sim$-equivalence;
see \cite{popo2010a} for posing the general problem on classifying admissible transformations. 
More precisely, proving Theorem~\ref{thm:ClassWAdmTrans}, 
we construct the generating (up to $G^\sim$-equivalence) subset~$\mathcal B$ of~$\mathcal G^\sim$
with the simultaneous classification of singular Lie-symmetry extensions within the class~$\mathcal W$.

Since the class~$\mathcal W$ is a subclass of~$\mathcal W_{\rm gen}$,
the transformational part~$\Phi$ of any admissible transformation
$\mathcal T=(\theta,\Phi,\tilde\theta)$ in the class~$\mathcal W$ takes the form
\[\tilde t=T(t,x),\quad \tilde x=X(t,x),\quad \tilde u=U=U^1(t,x)u+U^0(t,x)\]
with $T_tX_x-T_xX_t\ne0$ and $U^1\ne0$,
and additionally the system~\eqref{eq:DetEqsForAdmTrans1ut0ux0}--\eqref{eq:DetEqsForAdmTrans4ut0ux0}
is satisfied.
Here $\theta=(f,g)$ and $\tilde\theta=(\tilde f,\tilde g)$
are respectively the source and target arbitrary-element tuples for~$\mathcal T$, $\mathcal L_\theta,\mathcal L_{\tilde\theta}\in\mathcal W$.

By $\mathcal W_0$ and~$\mathcal W_1$ we respectively denote
the subclasses of~$\mathcal W$ singled out by the constraints $f_u=0$ and $f_u\ne0$.
The partition $\mathcal W=\mathcal W_0\sqcup\mathcal W_1$ of the class~$\mathcal W$
induces the partition of the equivalence groupoid~$\mathcal G^\sim$ of this class
since $\tilde f_{\tilde u}=0$ if and only if $f_u=0$,
cf.\ the end of Section~\ref{sec:PreliminaryStudyOfAdmTrans}.
In other words, equations in the subclass~$\mathcal W_0$ are not related
by point transformations to equations in the subclass~$\mathcal W_1$.
This claim can be nicely reformulated in terms of equivalence groupoids.

\begin{proposition}\label{pro:GenWaveEqsGroupoidPartition}
The equivalence groupoid~$\mathcal G^\sim$ of the class~$\mathcal W$
is the disjoint union of the equivalence groupoids~$\mathcal G^\sim_0$ and~$\mathcal G^\sim_1$
of the subclasses $\mathcal W_0$ and~$\mathcal W_1$,
$\mathcal G^\sim=\mathcal G^\sim_0\sqcup\mathcal G^\sim_1$.
\end{proposition}

We describe the equivalence groupoids~$\mathcal G^\sim_0$ and~$\mathcal G^\sim_1$ separately.

\begin{lemma}\label{lem:ClassW0AdmTrans}
The usual equivalence group of the subclass~$\mathcal W_0$ coincides with~$G^\sim$.
A generating (up to $G^\sim$-equivalence) set~$\mathcal B_0$ of admissible transformations for the class~$\mathcal W_0$, 
which is minimal and self-consistent with respect to $G^\sim$-equivalence, 
is the union of the restriction of the family~\ref{T1} to~$\mathcal W_0$ and the families~\ref{T3}--\ref{T9}.
\end{lemma}

\begin{lemma}\label{lem:ClassW0SpecialLieSymExts}
A complete list of $G^\sim$-inequivalent singular Lie-symmetry extensions for equations from the class~$\mathcal W_0$, 
which are not related to appropriated subalgebras of~$\mathfrak g^\sim$, is exhausted by 
Cases~\ref{case5a}--\ref{case7}, \ref{case18a}--\ref{case18c} and~\ref{case20} 
within the subclass of equations with arbitrary elements of the form $f=\ve$ and $g=g^1(x)g^2(u)$. 
\end{lemma}

\begin{proof}
We will simultaneously prove Lemmas~\ref{lem:ClassW0AdmTrans} and~\ref{lem:ClassW0SpecialLieSymExts}.
Let $\mathcal T\in \mathcal G^\sim_0$,
i.e., $f_u=0$, $g_{uu}\ne0$, $\tilde f_{\tilde u}=0$ and $\tilde g_{\tilde u\tilde u}\ne0$.
We express $X_t^{\;2}$ from the equation~\eqref{eq:DetEqsForAdmTrans2}, substitute this expression into the squared equation~\eqref{eq:DetEqsForAdmTrans1ut0ux0}.
After factorizing the resulting equation, we obtain the equation 
\[(T_t^{\;2}-fT_x^{\;2})(fX_x^{\;2}-\tilde fT_t^{\;2})=0,\]
which implies in view of $T_t^{\;2}-fT_x^{\;2}\ne0$ that $fX_x^{\;2}=\tilde fT_t^{\;2}$.
Then the equation~\eqref{eq:DetEqsForAdmTrans2} yields $X_t^{\;2} = f\tilde fT_x^{\;2}$.
This means that $f$ and~$\tilde f$ have the same sign and up to $G^\sim$-equivalence we can assume that $f=\tilde f=\ve$, where $\ve=\pm1$,
i.e., $X_x^{\;2}=T_t^{\;2}$ and $X_t^{\;2} = T_x^{\;2}$.
More precisely, the gauging of~$f$ and~$\tilde f$ can be realized via transformations
$\mathscr D(\varphi)$ and $\mathscr D(\tilde\varphi)$ of~$x$ and~$\tilde x$, respectively.
Taking into account the equation~\eqref{eq:DetEqsForAdmTrans1ut0ux0} and alternating the sign of~$t$
(this transformation belongs to the kernel group of the class~\eqref{eq:GenWaveEqs}),
we can set 
\[X_x=T_t,\quad X_t=\ve T_x.\]
Since these equations give $T_{tt}=\ve T_{xx}$ and $X_{tt}=\ve X_{xx}$,
the pair of the equations~\eqref{eq:DetEqsForAdmTrans4ut1} and~\eqref{eq:DetEqsForAdmTrans4ux1} reduce to
the system of linear homogeneous algebraic equations 
\[T_tU_{ut}-\ve T_xU_{ux}=0,\quad X_tU_{ut}-\ve X_xU_{ux}=0\]
with respect to~$U_{ut}$ and~$U_{ux}$.
The determinant of the associated matrix is nonzero, $\ve(T_tX_x-T_xX_t)\ne0$.
Hence $U_{ut}=U_{ux}=0$, i.e., $U_u$ is a nonzero constant and using a transformation $\mathscr D^u(c_2)$ we can set $U_u=1$.
In view of the above conditions, the equation~\eqref{eq:DetEqsForAdmTrans4ut0ux0} takes the form
\begin{equation}\label{eq:DetEqsForAdmTrans4ut0ux0Reduced1}
(X_x^{\;2}-\ve X_t^{\;2})\tilde g=g+U^0_{tt}-\ve U^0_{xx}.
\end{equation}
Sequentially acting on the equation~\eqref{eq:DetEqsForAdmTrans4ut0ux0Reduced1} by the operators $(X_x^{\;2}-\ve X_t^{\;2})^{-1}\p_t$ and~$\p_{\tilde t}$,
we obtain two differential consequences of~\eqref{eq:DetEqsForAdmTrans4ut0ux0Reduced1},
\begin{gather}\label{eq:DetEqsForAdmTrans4ut0ux0Reduced1a}
X_t\tilde g_{\tilde x}+U^0_t\tilde g_{\tilde u}+\frac{(X_x^{\;2}-\ve X_t^{\;2})_t}{X_x^{\;2}-\ve X_t^{\;2}}\tilde g
=\frac{(U^0_{tt}-\ve U^0_{xx})_t}{X_x^{\;2}-\ve X_t^{\;2}},
\\ \label{eq:DetEqsForAdmTrans4ut0ux0Reduced1b}
(X_t)_{\tilde t}\tilde g_{\tilde x}+(U^0_t)_{\tilde t}\tilde g_{\tilde u}
+\left(\frac{(X_x^{\;2}-\ve X_t^{\;2})_t}{X_x^{\;2}-\ve X_t^{\;2}}\right)_{\tilde t}\tilde g
=\left(\frac{(U^0_{tt}-\ve U^0_{xx})_t}{X_x^{\;2}-\ve X_t^{\;2}}\right)_{\tilde t}.
\end{gather}
Studying the consistency of the equations~\eqref{eq:DetEqsForAdmTrans4ut0ux0Reduced1a} and~\eqref{eq:DetEqsForAdmTrans4ut0ux0Reduced1b}
as first-order quasilinear partial differential equations with respect to~$\tilde g$,
we consider different cases depending on whether 
the matrix of coefficients of the derivatives~$\tilde g_{\tilde x}$ and~$\tilde g_{\tilde u}$
in the system of these equations is degenerate or nondegenerate.%
\footnote{\label{fnt:MethodOfFurcateSplittingForGenerating SetsOfAdmTrans}%
This procedure and the previous partition of~$\mathcal W$ into~$\mathcal W_0$ and~$\mathcal W_1$
fits well into the framework of the method of furcate splitting  
\cite{niki2001a,opan2020b,vane2015d,vane2012a}. 
The further consideration is the first construction of a generating set of admissible transformations 
using this method. 
The required computations were carried out in 2011 without involving algebraic techniques 
and gave the first application of furcate splitting to finding admissible transformations 
but they were too cumbersome, and the obtained results were not published. 
Later, the method of furcate splitting was used to describe 
the equivalence groupoid of the class of general Burgers--Korteweg--de Vries equations with space-dependent coefficients
via classifying maximal conditional equivalence groups of this class~\cite{opan2017a}, see also~\cite{opan2019a}.
Therefore, the method of furcate splitting can be extended to admissible transformations 
in various ways depending on which terms the corresponding equivalence groupoid can be described in.}

\medskip\par\noindent{\bf 1.}
Suppose first that this matrix is nondegenerate,
$X_t(U^0_t)_{\tilde t}-U^0_t(X_t)_{\tilde t}\ne0$.
We solve 
the system~\eqref{eq:DetEqsForAdmTrans4ut0ux0Reduced1a}--\eqref{eq:DetEqsForAdmTrans4ut0ux0Reduced1b}
as a system of linear algebraic equations with respect to~$\tilde g_{\tilde x}$ and~$\tilde g_{\tilde u}$, 
\begin{gather}\label{eq:DetEqsForAdmTrans4ut0ux0Reduced1Solved}
\tilde g_{\tilde x}=\alpha^1\tilde g+\alpha^0, \quad
\tilde g_{\tilde u}=\beta^1\tilde g+\beta^0. 
\end{gather}
Here the coefficients $\alpha^0$, $\alpha^1$, $\beta^0$ and~$\beta^1$ are functions of~$(\tilde t,\tilde x)$
whose explicit expressions in terms of~$X$ and~$U^0$ are not essential for the further consideration. 
Differentiating the equations~\eqref{eq:DetEqsForAdmTrans4ut0ux0Reduced1Solved} with respect to~$\tilde t$,
we derive the consequences $\alpha^1_{\tilde t}\tilde g+\alpha^0_{\tilde t}=0$, $\beta^1_{\tilde t}\tilde g+\beta^0_{\tilde t}=0$, 
which implies in view of $g_u\ne0$ that $\alpha^1_{\tilde t}=\alpha^0_{\tilde t}=\beta^1_{\tilde t}=\beta^0_{\tilde t}=0$. 
The cross-differentiation of the equations~\eqref{eq:DetEqsForAdmTrans4ut0ux0Reduced1Solved}  
with respect to~$\tilde x$ and~$\tilde u$ leads to the compatibility conditions for these equations, 
which are $\beta^1_{\tilde x}=0$ and $\beta^0_{\tilde x}=\alpha^1\beta^0-\alpha^0\beta^1$. 
Therefore, $\beta^1$ is a constant. 
Since $\tilde g_{\tilde u\tilde u}\ne0$, 
the second equation in~\eqref{eq:DetEqsForAdmTrans4ut0ux0Reduced1Solved} implies that $\beta^1\ne0$. 
Using the equivalence transformation $\mathscr D^u(1/\beta^1)$, we can gauge $\beta^1$ to~1. 
Then the second equation in~\eqref{eq:DetEqsForAdmTrans4ut0ux0Reduced1Solved} integrates to 
$\tilde g=\tilde g^0(\tilde x){\rm e}^{\tilde u}+\tilde g^1(\tilde x)$ with $\tilde g^1=-\beta^0$. 
The equation~\eqref{eq:DetEqsForAdmTrans4ut0ux0Reduced1} implies that 
the function~$g$ is of similar form in the initial variables, $g=g^0(x){\rm e}^u+g^1(x)$, and 
$(X_x^{\;2}-\ve X_t^{\;2}){\rm e}^{U^0}\tilde g^0=g^0$, 
$(X_x^{\;2}-\ve X_t^{\;2})\tilde g^1=g^1+U^0_{tt}-\ve U^0_{xx}$.
Using equivalence transformations of the form $\mathscr Z(\psi)$ in both the old and new variables, 
we set $g^0=\ve'$ and $\tilde g^0=\tilde\ve'$ with $\ve',\tilde\ve'=\pm1$.
Then $(X_x^{\;2}-\ve X_t^{\;2}){\rm e}^{U^0}\tilde\ve'=\ve'$, i.e., 
\[U^0=-\ln|X_x^{\;2}-\ve X_t^{\;2}|\] 
and thus $U^0_{tt}-\ve U^0_{xx}=0$ since $X_{tt}-\ve X_{xx}=0$.
The equation~\eqref{eq:DetEqsForAdmTrans4ut0ux0Reduced1} takes the form 
\begin{equation}\label{eq:DetEqsForAdmTrans4ut0ux0Reduced1c}
\frac{g^1(x)}{X_x^{\;2}-\ve X_t^{\;2}}=\tilde g^1(X).
\end{equation}
In other words, in this case it suffices to classify admissible transformations within 
the subclass~$\mathcal W_{00}$ of equations of the form 
\begin{gather}\label{eq:GenWaveEqsSingularSubclass00}
u_{tt}=\ve u_{xx}+\ve'{\rm e}^u+g^1(x)
\end{gather}
up to the subgroup~$G^\sim_{00}$ of~$G^\sim$ singled out by the constraints $\varphi_x=\pm c_1$, $c_2=|c_1|^{-1/2}$ and $\psi=0$.
This reduces to deriving possible $G^\sim_{00}$-inequivalent expressions
for $X=X(t,x)$, $g^1=g^1(x)$ and $\tilde g^1=\tilde g^1(\tilde x)$
satisfying the joint system of the equation~\eqref{eq:DetEqsForAdmTrans4ut0ux0Reduced1c} and the equation $X_{tt}=\ve X_{xx}$.
We change the independent variables in this system, $y=x+\iota t$ and $z=x-\iota t$,
where $\iota=1$ or $\iota=i$ if $\ve=1$ or $\ve=-1$, respectively, and $i$ is the imaginary unit, $i^2=-1$. 
Hence $\iota^2=\ve$.
In the variables~$y$ and~$z$ the equation $X_{tt}=\ve X_{xx}$ takes the form $X_{yz}=0$
and its general solution is represented as $X=Y(y)+Z(z)$, where $Z(z)$ coincides with the complex conjugate of~$Y(y)$ if $\ve=-1$.
Then the equation~\eqref{eq:DetEqsForAdmTrans4ut0ux0Reduced1c} can be rewritten as 
\begin{equation}\label{eq:DetEqsForAdmTrans4ut0ux0Reduced1c2}
\frac1{4Y_yZ_z}\,g^1(x)=\tilde g^1(X), 
\quad\mbox{assuming}\quad x=\frac{y+z}2, \quad X=Y+Z.
\end{equation}
Excluding the parameter function~$\tilde g^1$ via acting by the operator $Y_y\p_z-Z_z\p_y$
on the equation~\eqref{eq:DetEqsForAdmTrans4ut0ux0Reduced1c2}, 
we reduce this equation, after the expansion and algebraic transformations, to 
\[
2g^1(\p_y+\p_z)\big(Y_y^{-1}-Z_z^{-1}\big)=-g^1_x\big(Y_y^{-1}-Z_z^{-1}\big).
\]
The last equation integrates to $Y_y^{-1}-Z_z^{-1}=\iota h^0h^1$, 
where $h^0$ is a (real-valued) smooth function of~$t$,
$h^1:=|g^1|^{-1/2}\ne0$ and thus $h^1$ is a (real-valued) smooth function of~$x$. 
We act on the integration result by the operator $\p_t^{\,2}-\ve\p_x^{\,2}=-4\ve\p_y\p_z$ to get 
$h^0_{tt}h^1=\ve h^0h^1_{xx}$. 

If $h^0=0$, then $Y_y^{-1}=Z_z^{-1}=\const\in\mathbb R$ and thus $Y_y=Z_z=\const\in\mathbb R$, 
which implies $X_{xx}=0$.  
Therefore, $T_x=T_{tt}=0$ as well. 
This means that the admissible transformation~$\mathcal T$ is induced by an element of~$G^\sim$. 

The case $h^1=g^1=0$ corresponds to the Liouville equation. 
The sign of~$\ve'$ is alternated by~the corresponding admissible transformation from the family~\ref{T1} if $\ve=1$
and cannot be alternated in view of the equation 
$(X_x^{\;2}-\ve X_t^{\;2}){\rm e}^{U^0}\tilde\ve'=\ve'$ if $\ve=-1$. 
The equivalence group $G^\sim$ induces 
the subgroup~$H$ of the complete point symmetry group of the Liouville equation for each fixed value of~$(\ve,\ve')$, 
which is constituted by the transformations of the form $\tilde t=c_1t+c_0$, $\tilde x=c_1x+c_2$,  
where $c_0$, $c_1$ and~$c_2$ are arbitrary constants with $c_1\ne0$. 
For the minimality of the set of admissible transformations to be constructed, 
we should take a single representative in each coset of $G^\sim$-equivalent elements 
of the corresponding vertex group. 
This gives the family~\ref{T9} of admissible transformations.

Further we assume that $h^0h^1\ne0$ and thus $g^1\ne0$ as well. 
The separation of variables in the equation $h^0_{tt}h^1=\ve h^0h^1_{xx}$ implies that 
$h^0_{tt}/h^0=\ve h^1_{xx}/h^1$ is a constant, 
which can be assumed, modulo scalings from~$G^\sim$ preserving the constraint $f=\ve$, 
to take values from the set $\{-1,0,1\}$. 
Up to shifts of~$x$ and alternating the sign of~$x$, we have 
\[
g^1\in\mathcal C:=\big\{\nu,\,\nu x^{-2},\,\nu\cos^{-2}x,\,-\nu\cosh^{-2}x,\,\nu\sinh^{-2}x,\,\ve'' {\rm e}^{-2x}\mid\nu\in\mathbb R,\nu\ne0,\ve''=\pm1\big\}.
\]
Using the same arguments for the inverse of the admissible transformation~$\mathcal T$, 
we obtain that the function~$\tilde g^1$ also belongs to the set~$\mathcal C$ 
(up to replacing the argument~$x$ by $\tilde x$).

We first present a complete set of $G^\sim$-inequivalent 
(independent up to inversion and composing with each other and with equivalence transformations) 
non-identity admissible transformations for~$g^1$ running through the set~$\mathcal C$ 
and then explain the derivation of this set. 
It is exhausted by the family~\ref{T8}$|_{\mathcal W_{00}}$ and the following families:
  
\renewcommand{\theenumi}{{\rm T\arabic{enumi}$'$}}
\renewcommand{\labelenumii}{{\rm\alph{enumii}.}}
\begin{enumerate}\setcounter{enumi}{2}

\item\label{T3'} 
$f=1$, \ $g={\rm e}^u+\ve''{\rm e}^{-2x}$, \ $\tilde f=1$, \ $\tilde g={\rm e}^u+\ve''$,\\[.5ex]
$\Phi{:}$ \ $\tilde t={\rm e}^{-x}\sinh t$, \ $\tilde x={\rm e}^{-x}\cosh t$, \ $\tilde u=u+2x$;

\item\label{T4'} 
$f=1$, \ $g=\ve'{\rm e}^u+g^1(x)$, \ 
$\tilde f=1$, \ $\tilde g=\tilde\ve'{\rm e}^u+\nu\tilde x^{-2}$, \ $\nu\in\mathbb R_{\ne0}$, \vspace{-1ex}
  \begin{enumerate}
  \item\label{T4'a} 
  $g^1(x)=\nu x^{-2}$, \ $\tilde\ve'=\ve'$, \ $\Phi{:}$ \ $\tilde t=\dfrac t{x^2-t^2}$, \ $\tilde x=\dfrac x{x^2-t^2}$, \ $\tilde u=u+2\ln|x^2-t^2|$;
  \item\label{T4'b} 
  $g^1(x)=\nu\cos^{-2}x$, \ $\tilde\ve'=\ve'$,\\[.5ex] 
  $\Phi{:}$ \ $\tilde t=\dfrac{\cos t}{\sin t+\sin x}$, \ $\tilde x=\dfrac{\cos x}{\sin t+\sin x}$, \ $\tilde u=u+2\ln|\sin t+\sin x|$;\vspace{.5ex}
  \item\label{T4'c} 
  $g^1(x)=-\nu\cosh^{-2}x$, \ $\tilde\ve'=-\ve'$, \ $\Phi{:}$ \ $\tilde t={\rm e}^t\sinh x$, \ $\tilde x={\rm e}^t\cosh x$, \ $\tilde u=u-2t$;\vspace{1ex}
  \item\label{T4'd} 
  $g^1(x)=\nu\sinh^{-2}x$, \ $\tilde\ve'=\ve'$, \ $\Phi{:}$ \ $\tilde t={\rm e}^t\cosh x$, \ $\tilde x={\rm e}^t\sinh x$, \ $\tilde u=u-2t$;
  \end{enumerate}

\item\label{T5'} 
$f=-1$, \ $g=\ve'{\rm e}^u+\ve''{\rm e}^{-2x}$, \ $\tilde f=-1$, \ $\tilde g=\ve'{\rm e}^u+\ve''$,\\
$\Phi{:}$ \ $\tilde t={\rm e}^{-x}\sin t$, \ $\tilde x={\rm e}^{-x}\cos t$, \ $\tilde u=u+2x$;

\item\label{T6'} 
$f=-1$, \ $g=\ve'{\rm e}^u+g^1(x)$, \ 
$\tilde f=-1$, \ $\tilde g=\ve'{\rm e}^u+\nu\tilde x^{-2}$, \ $\nu\in\mathbb R_{\ne0}$, \vspace{-1ex}
  \begin{enumerate}
  \item\label{T6'a} 
  $g^1(x)=\nu x^{-2}$, \ $\Phi{:}$ \ $\tilde t=\dfrac t{x^2+t^2}$, \ $\tilde x=\dfrac x{x^2+t^2}$, \ $\tilde u=u+2\ln|x^2+t^2|$;\vspace{.5ex}
  \item\label{T6'b} 
  $g^1(x)=\nu\cos^{-2}x$, \ $\Phi{:}$ \ $\tilde t={\rm e}^t\sin x$, \ $\tilde x={\rm e}^t\cos x$, \ $\tilde u=u-2t$;\vspace{.5ex}
  \item\label{T6'c} 
  $g^1(x)=\nu\sinh^{-2}x$,\\[.5ex] 
  $\Phi{:}$ \ $\tilde t=\dfrac{\sin t}{\cos t+\cosh x}$, \ $\tilde x=\dfrac{\sinh x}{\cos t+\cosh x}$, \ $\tilde u=u+2\ln|\cos t+\cosh x|$;
  \end{enumerate}

\item\label{T7'} 
$f=-1$, \ $g=\ve'{\rm e}^u+\nu\cosh^{-2}x$, \ $\tilde f=-1$, \ $\tilde g=\ve'{\rm e}^{\tilde u}+\nu\cosh^{-2}\tilde x$, \ 
$\nu\in\mathbb R_{\ne0}$,\\[1.5ex] 
$\Phi{:}$ \ 
$\tilde t=\arctan\dfrac{\sin\gamma\,\sinh x+\cos\gamma\sin t}{\cos t}$, \ 
$\tilde x=\mathop{\rm arctanh}\dfrac{\cos\gamma\,\sinh x-\sin\gamma\sin t}{\cosh x}$,\\[1.5ex]  
$\tilde u=u+\ln\big|\cosh^2x-(\cos\gamma\,\sinh x-\sin\gamma\sin t)^2\big|$, \ $\gamma\in(0,2\pi)$.

\end{enumerate}

The direct way of checking which elements of the set~$\mathcal C$ 
are related via admissible transformations is to fix an element~$g^1$ in $\mathcal C$, 
thus defining $h^1:=|g^1|^{-1/2}\ne0$, 
to solve the equation $h^0_{tt}=\lambda h^0$ with $\lambda:=\ve h^1_{xx}/h^1=\const$, 
to find~$Y$ and~$Z$ by separating variables~$y$ and~$z$ in the equation $Y_y^{-1}-Z_z^{-1}=\iota h^0h^1$
and further integrating, 
and finally to determine $\tilde g^1$ from~\eqref{eq:DetEqsForAdmTrans4ut0ux0Reduced1c}. 

We follow an optimized strategy. 
In the above way, we find the mappings 
$\nu {\rm e}^{-2x}\mapsto\nu$ by~\ref{T3'} if $f=1$ and by~\ref{T5'} if $f=-1$,  
$\nu\cos^{-2}x\mapsto\nu x^{-2}$ by~\ref{T4'b} if $f=1$ and by~\ref{T6'b} if $f=-1$,  
$-\nu\cosh^{-2}x\mapsto\nu x^{-2}$ by~\ref{T4'c} if $f=1$,  
$\nu\sinh^{-2}x\mapsto\nu x^{-2}$ by~\ref{T4'd} if $f=1$ and by~\ref{T6'c} if $f=-1$. 
The sign of~$\ve'$ is alternated only in~\ref{T4'c}. 
For $f=-1$, the value $g^1=-\nu\cosh^{-2}x$ is mapped to the value $g^1=\nu x^{-2}$ 
by an admissible point transformation only over the complex field. 

The maximal Lie invariance algebras of the equations of the form~\eqref{eq:GenWaveEqsSingularSubclass00} 
with values of $(f,g^1)$ that have not been reduced to other ones are
\begin{gather*}
(f,g^1)=(1,\nu)            \colon\ \ \mathfrak g_\theta=\langle\p_t,\, \p_x,\, x\p_t+t\p_x\rangle,\\[.5ex]
(f,g^1)=(-1,\nu)           \colon\ \ \mathfrak g_\theta=\langle\p_t,\, \p_x,\, x\p_t-t\p_x\rangle,\\[.5ex]
(f,g^1)=(\ve,\nu x^{-2})   \colon\ \ \mathfrak g_\theta=\langle\p_t,\, t\p_t+x\p_x-2\p_u,\, (t^2+\ve x^2)\p_t+2tx\p_x-4t\p_u\rangle,\\[.5ex]
(f,g^1)=(-1,\nu\cosh^{-2}x)\colon\ \ \mathfrak g_\theta=\langle\p_t,\, \RR'(\cos t\,\cosh x),\, \RR'(\sin t\,\cosh x)\rangle,
\end{gather*}
where $\RR'(\Phi):=\Phi_x\p_t+\Phi_t\p_x-2\Phi_{tx}\p_u$.
These invariance algebras are given in Cases~\ref{case5a}$_{\ve=1}$ and~\ref{case5a}$_{\ve=-1}$ of Table~\ref{tab:GenWaveEqsExtensions} 
and are associated with Cases~\ref{case6a} and~\ref{case7} of the same table, respectively. 
They are realizations of 
the Poincar\'e algebra ${\rm p}(1,1)$, 
the Euclidian algebra ${\rm e}(2)$, 
the real special linear algebra ${\rm sl}(2,\mathbb R)$ 
and the orthogonal algebra ${\rm o}(3)$,  
which are not isomorphic to each other. 
At the same time, systems of differential equations are related by point transformations 
only if their maximal Lie invariance algebras are isomorphic.

Therefore, we need to classify admissible transformations within the four subclasses 
of equations of the form~\eqref{eq:GenWaveEqsSingularSubclass00}, 
where for each of these subclasses the tuple $(f,g^1)$ is of a fixed form 
in $\{(1,\nu),\,(-1,\nu),\,(\ve,\nu x^{-2})\,(-1,\nu\cosh^{-2})\}$, 
and $\nu$ runs through~$\mathbb R_{\ne0}$. 
For this purpose, we apply for the first time an extension of the algebraic method 
to finding admissible transformations. 
This method was suggested by Hydon in \cite{hydo1998a,hydo2000b,hydo2000A} 
for computing discrete symmetries 
and extended to equivalence transformations in~\cite{bihl2015a}.

We in detail consider only the first subclass. 
Let~$\mathcal L_\theta$ and~$\mathcal L_{\tilde\theta}$ be two fixed equations
of the form~\eqref{eq:GenWaveEqsSingularSubclass00} with
$f=\tilde f=1$, $\tilde g^1=\nu$, $g^1=\tilde\nu$ and some $\ve',\tilde\ve'=\pm1$. 
These equations have the same maximal Lie invariance algebra, 
$\mathfrak g_\theta=\mathfrak g_{\tilde\theta}=\langle\p_t,\p_x,t\p_x+x\p_t\rangle$, 
which is given in Case~\ref{case5a}$_{f=1}$ of Table~\ref{tab:GenWaveEqsExtensions}
and is a realization of the Poincar\'e algebra ${\rm p}(1,1)$. 
Therefore, the pushforward of vector fields by~$\Phi$ induces an automorphism of~$\mathfrak g_\theta$ 
associated with an automorphism of ${\rm p}(1,1)$. 
Recall that the transformation~$\Phi$ is completely defined by its $t$- and $x$-components. 
Inner automorphisms of ${\rm p}(1,1)$ correspond to continuous point transformations 
generated by vector fields from~$\mathfrak g_\theta$.
Such transformations are symmetries of~$\mathcal L_\theta$, 
i.e., they do not change the parameters~$\nu$ and~$\ve'$. 
Up to shifts of~$t$ and~$x$, which are induced by elements of~$G^\sim$, 
we obtain the family~\ref{T8a}$|_{\mathcal W_{00}}$ of admissible transformations. 
There are only two outer automorphisms of ${\rm p}(1,1)$ 
that are independent up to composing to each other and to inner automorphisms.%
\footnote{%
See~\cite{fish2013a,popo2003a} for necessary facts on automorphisms of low-dimensional Lie algebras. 
}
The corresponding transformations are 
the alternation of the sign of~$t$, 
which is a discrete symmetry of~$\mathcal L_\theta$ induced by $\mathscr D^t(-1)$,
and the permutation of~$t$ and~$x$, which belongs to the family~\ref{T1}. 
There is no point transformation that satisfies the restriction for~$\Phi$ 
and induces the identity automorphism of~$\mathfrak g_\theta$. 

The other three subclasses are considered in a similar way. 
Each of the algebras ${\rm e}(2)$ and ${\rm sl}(2,\mathbb R)$ possesses a single independent outer automorphism, 
which is here related, e.g., to alternating the sign of~$t$. 
The algebra ${\rm o}(3)$ admits no outer automorphism 
but alternating the sign of~$t$ generates the identity automorphism of the corresponding algebra~$\mathfrak g_\theta$. 
Factoring out shift and scaling symmetries of related equations, which are induced by elements of~$G^\sim$, 
we construct the families~\ref{T8b}$|_{\mathcal W_{00}}$, \ref{T4'a} and~\ref{T7'}, respectively. 
The last family consists of the non-identity transformations generated by 
the Lie-symmetry vector field~$\RR'(\cos t\,\cosh x)$ of the equation~\eqref{eq:GenWaveEqsSingularSubclass00}
with $(f,g^1)=(-1,\nu\cosh^{-2})$.

It is obvious that $G^\sim_{00}$-inequivalent singular cases of Lie-symmetry extensions within the subclass~$\mathcal W_{00}$
are exhausted by those with $g^1\in\mathcal C$.

\medskip\par\noindent{\bf 2.}
Now we suppose that the matrix of coefficients of the derivatives~$\tilde g_{\tilde x}$ and~$\tilde g_{\tilde u}$
in the system~\eqref{eq:DetEqsForAdmTrans4ut0ux0Reduced1a}, \eqref{eq:DetEqsForAdmTrans4ut0ux0Reduced1b}
is degenerate,
\begin{equation}\label{eq:DetEqsForAdmTrans4ut0ux0Reduced1DegeneracyCondition}
X_t(U^0_t)_{\tilde t}-U^0_t(X_t)_{\tilde t}=0.
\end{equation}
If $X_t=0$ then $T_x=X_{xx}=T_{tt}=0$ and
the equation~\eqref{eq:DetEqsForAdmTrans4ut0ux0Reduced1} implies in view of the condition $\tilde g_{\tilde u\tilde u}\ne0$
that $U^0_t=0$, i.e., the admissible transformation~$\mathcal T$ is generated by an element of~$G^\sim$.
This is why in what follows we assume that $X_t\ne0$.
Representing the equation~\eqref{eq:DetEqsForAdmTrans4ut0ux0Reduced1DegeneracyCondition} in the form
$(U^0_t/X_t)_{\tilde t}=0$, we integrate it by~$\tilde t$,
which yields $U^0_t=V^0(X)X_t$ for some smooth function~$V^0=V^0(\tilde x)$.
Then we integrate by~$t$, obtaining $U^0=V^1(X)+V^2(x)$,
where $V^1$ is an antiderivative of~$V^0$, $V^1_{\tilde x}(\tilde x)=V^0(\tilde x)$,
and $V^2=V^2(x)$ is a smooth function of~$x$.
Therefore, up to $G^\sim$-equivalence
(namely, up to composing the transformation~$\mathcal T$ with transformations from the subgroup $\{\mathscr Z(\psi)\}$),
we can set $U^0=0$ and thus $U=u$. 
We then rewrite the equation~\eqref{eq:DetEqsForAdmTrans4ut0ux0Reduced1a} as
\[
\frac{\tilde g_{\tilde x}}{\tilde g}=-\frac1{X_t}\frac{(X_x^{\;2}-\ve X_t^{\;2})_t}{X_x^{\;2}-\ve X_t^{\;2}}.
\]
The left- and right-hand sides of the last equations do not depend on~$\tilde t$ and~$\tilde u$, respectively,
and hence they are equal to a function of only~$\tilde x$.
Solving the equation with respect to~$\tilde g$ gives the representation of~$\tilde g$ as the product of functions of different arguments,
$\tilde g=\tilde g^1(\tilde x)\tilde g^2(\tilde u)$.
Since \mbox{$\tilde u=u$} and $g=(X_x^{\;2}-\ve X_t^{\;2})\tilde g$, 
the function~$g$ admits the similar representation $g=g^1(x)g^2(u)$, where $g^2(u)=\tilde g^2(u)$. 
As a result, we again obtain the equation~\eqref{eq:DetEqsForAdmTrans4ut0ux0Reduced1c}. 
Therefore, equations of the form 
\looseness=-1
\begin{gather}\label{eq:GenWaveEqsSingularSubclass01} 
u_{tt}=\ve u_{xx}+g^1(x)g^2(u),\quad\mbox{where}\quad g^2_ug^2_{uuu}\ne(g^2_{uu})^2,
\end{gather}
with coinciding values of the parameter function~$g^2$
are related by a point transformation if and only if
the equations of the form~\eqref{eq:GenWaveEqsSingularSubclass00} 
with the same values of the parameters~$\ve$ and~$g^1$ and some values of~$\ve'$ are related by a point transformation. 
The inequality $g^2_ug^2_{uuu}\ne(g^2_{uu})^2$, 
which is equivalent to the linear independence of $g^2_u$, $g^2$ and~1, 
is imposed for excluding the intersection 
of the subclass~\eqref{eq:GenWaveEqsSingularSubclass01} 
and the subclass of equations of the form $u_{tt}=\ve u_{xx}+g^0(x){\rm e}^u+g^1(x)$ with $g^0\ne0$, 
which are reduced by equivalence transformations to equations 
from the subclass~\eqref{eq:GenWaveEqsSingularSubclass00}.
To properly translate the classification of admissible transformations 
within the subclass~\eqref{eq:GenWaveEqsSingularSubclass00} 
to those within the subclass~\eqref{eq:GenWaveEqsSingularSubclass01}, 
we take into account the condition $g^1\ne0$ for the subclass~\eqref{eq:GenWaveEqsSingularSubclass01} 
and replace the $u$-components of all transformational parts by $\tilde u=u$. 
Since the coefficients $\nu$ and $\tilde\nu$ coincide, they can just be absorbed by~$g^2$.
In total, this gives the restrictions of the families~\ref{T1}, \ref{T3}--\ref{T7} and~\ref{T8} 
to the subclass~\eqref{eq:GenWaveEqsSingularSubclass01}.
Here the families \ref{T3}--\ref{T7} respectively correspond to the families \ref{T3'}--\ref{T7'}.

\medskip\par

To complete the classification of admissible transformations in the subclass~$\mathcal W_0$, 
we map each equation from the subclass~\eqref{eq:GenWaveEqsSingularSubclass00} with $g^1\in\mathcal C$
by the equivalence transformation $\mathscr Z(\ln\hat g^1)$
to the equation $u_{tt}=\ve u_{xx}+\hat g^1(x)g^2(u)$ with $g^2(u)=\ve'{\rm e}^u+\hat\nu$.
Here $\hat g^1$ is of the same form as $g^1$ but with fixed $\nu=1$, $\hat\nu=2\ve\ve''$, 
$\ve''=-1$ for $\hat g^1=\cosh^{-2}x$ and $\ve''=1$ otherwise. 
As a result, the families \ref{T3'}--\ref{T7'} are mapped into the families \ref{T3}--\ref{T7}. 
The completion of the latter families allows us to neglect 
the auxiliary inequality $g^2_ug^2_{uuu}\ne(g^2_{uu})^2$ of the subclass~\eqref{eq:GenWaveEqsSingularSubclass01}
for these families.

The obtained set~$\mathcal B_0$ of admissible transformations of the subclass~$\mathcal W_0$ 
is a generating set for~$\mathcal G^\sim_0$ up to $G^\sim$-equivalence by construction. 
No element of $\mathcal G^{G^\sim}_0\!\!$ relates different values of~$\theta$ from
${\rm s}(\mathcal B_0)\cup{\rm t}(\mathcal B_0)$. 
No element of~$\mathcal B_0$ can be represented 
as the composition of a finite number of other elements of~$\mathcal B_0$ or their inverses.  
Therefore, the generating set~$\mathcal B_0$ is minimal and self-consistent 
with respect to $G^\sim$-equivalence for~$\mathcal G^\sim_0$.

Up to $G^\sim$-equivalence, singular Lie-symmetry extensions in the subclass~\eqref{eq:GenWaveEqsSingularSubclass01}
are possible only for $g^1\in\mathcal C$, which gives Cases~\ref{case5a}--\ref{case5c}, \ref{case6a}--\ref{case6f} and~\ref{case7}.  
Since Case~\ref{case7} reduces to Case~\ref{case6a} over the complex field, 
for further extensions it suffices to check only equations with $g^1=1$ (Case~\ref{case5a}) and with $g^1=x^{-2}$ (Case~\ref{case6a}). 
For $g^1=1$, we obtain only Case~\ref{case18a} with two $\mathcal G^\sim$-equivalent Cases~\ref{case18b} and~\ref{case18c}. 
There are no further Lie-symmetry extensions for $g^1=x^{-2}$.
\end{proof}

\begin{remark}\label{rem:GenWaveEqsLinCase}
A generating set of admissible point transformations 
within the class~$\mathcal W_{\rm lin}$ of linear equations of the form~\eqref{eq:GenWaveEqs} 
and the group classification of this class can be easily derived 
from the computation of their counterparts for the class~$\mathcal W_0$. 
For this purpose, we need to consider the essential subgroupoid~$\mathcal G^{\sim\rm ess}_{\rm lin}$ 
of the equivalence groupoid~$\mathcal G^\sim_{\rm lin}$ of~$\mathcal W_{\rm lin}$ 
and essential Lie invariance algebras of equations from~$\mathcal W_{\rm lin}$, 
respectively factoring out transformations and Lie-symmetry vector fields 
related to the linear superposition of solutions, 
cf.\ \cite[Section~2]{popo2008a} and~\cite{kuru2018a}.
Elements of~$\mathcal W_{\rm lin}$ take the form $u_{tt}=f(x)u_{xx}+g^1(x)u+g^0(x)$. 
Their kernel point symmetry group is generated by the transformations 
$\pi_*\mathscr P^t(c_0)$ and $\pi_*\mathscr D^u(c_2)$, 
and their kernel invariance algebra is $\mathfrak g^\cap_{\rm lin}=\langle\p_t,u\p_u\rangle$.
The equivalence group of~$\mathcal W_{\rm lin}$ coincides with~$G^\sim$. 
Using equivalence transformations, we can gauge the parameter functions~$f$ and~$g^0$ 
to $\ve\in\{-1,1\}$ and~$0$, respectively.%
\footnote{%
Another possible maximal gauge for the arbitrary-element tuple~$\theta$ 
within the class~$\mathcal W_{\rm lin}$ is $g^0=g^1=0$, i.e., $g=0$, 
which leads to the subclass of equations of the form $u_{tt}=f(x)u_{xx}$. 
Such equations with $f>0$ are linear wave equations for an inhomogeneous medium, 
whose Lie symmetries were studied in~\cite{blum1987a}; 
see also \cite[Section~4.2(3)]{blum1989A} for a more arranged presentation of these results.   
Some Lie-symmetry extensions from the list given in \cite[Section~4.2(3)]{blum1989A}
cannot be represented in an explicit form, 
including both the corresponding values of the arbitrary element~$f$ 
and basis elements of the associated maximal Lie invariance algebras. 
It would be instructive to reduce this list using the equivalence with respect to 
the equivalence group of the subclass~$\mathcal W_{\rm lin''}$ of equations from~$\mathcal W_{\rm lin}$ with $g^0=g^1=0$.
Note that this group is of more complicated structure 
than the equivalence group of the subclass~$\mathcal W_{\rm lin'}$ of~$\mathcal W_{\rm lin}$ 
singled out by the gauge $f=\ve$, $g^0=0$. 
Additional equivalence transformations for the subclass~$\mathcal W_{\rm lin''}$ need a separate study. 
This is why the gauge $f=\ve$, $g^0=0$ is preferable.
One can try to translate the results of Remark~\ref{rem:GenWaveEqsLinCase} for the subclass~$\mathcal W_{\rm lin'}$ 
to similar results for the subclass~$\mathcal W_{\rm lin''}$ using a mapping between these subclasses 
that is induced by a family of equivalence transformations of the class~$\mathcal W_{\rm lin}$.
} 
As a result, we map the class~$\mathcal W_{\rm lin}$ onto its subclass~$\mathcal W_{\rm lin'}$ 
of linear wave and elliptic equations with $x$-dependent potentials, 
which are of the form 
\[u_{tt}=\ve u_{xx}+g^1(x)u,\] 
cf.~\cite{zhda1993a}. 
Each equation~$\mathcal L_\theta$ from~$\mathcal W_{\rm lin'}$ admits 
the (pseudo)group~$G_\theta^{\rm lin}$ of point symmetries 
associated with the linear superposition of solutions, 
$G^{\rm lin}_\theta=\{\Phi\colon\tilde t=t,\,\tilde x=x,\,\tilde u=u+h(t,x)\mid h\in\mathcal L_\theta\}$, 
where the notation ``$h\in\mathcal L_\theta$'' means 
that the function~$h$ runs through the solution set of~$\mathcal L_\theta$. 
The corresponding Lie algebra is 
$\mathfrak g^{\rm lin}_\theta=\langle h(t,x)\p_u\mid h\in\mathcal L_\theta\rangle$. 
Let $\mathcal G^{\sim\rm lin}_{\rm lin'}$ be the subgroupoid 
of the equivalence groupoid~$\mathcal G^\sim_{\rm lin'}$ of the class~$\mathcal W_{\rm lin'}$ 
that is constituted by the admissible transformations related the linear superposition of solutions, 
i.e., $\mathcal G^{\sim\rm lin}_{\rm lin'}$ is the union of~$G^{\rm lin}_\theta$ 
as subgroups of vertex groups in~$\mathcal G^\sim_{\rm lin'}$ 
for all~$\theta$ with $\mathcal L_\theta\in\mathcal W_{\rm lin'}$.
The essential equivalence groupoid~$\mathcal G^{\sim\rm ess}_{\rm lin'}$ of~$\mathcal W_{\rm lin'}$, 
which is the complement of $\mathcal G^{\sim\rm lin}_{\rm lin'}$ in~$\mathcal G^\sim_{\rm lin'}$, 
is naturally isomorphic to the equivalence groupoid of the class of equations 
of the form~\eqref{eq:GenWaveEqsSingularSubclass00} with a fixed value of~$\ve'$. 
Therefore, a generating set for~$\mathcal G^{\sim\rm ess}_{\rm lin'}$ 
and, thus, for the essential equivalence groupoid~$\mathcal G^{\sim\rm ess}_{\rm lin}$ of~$\mathcal W_{\rm lin}$, 
which is defined similarly to~$\mathcal G^{\sim\rm ess}_{\rm lin'}$,
consists of the counterparts of the families \ref{T3}--\ref{T9} for linear equations. 
That is, one should substitute $g^2=u$, $\tilde g^2=\tilde u$ into \ref{T3}--\ref{T7} and 
$g=\ve''u$, $\tilde g=\ve''\tilde u$ into~\ref{T8}
and replace $g$, $\tilde g$ and the $u$-component of~$\Phi$ in~\ref{T9}  
by $g=\ve''u$, $\tilde g=\ve''\tilde u$ and $\tilde u=u$. 
A complete list of $G^\sim$-inequivalent essential Lie-symmetry extensions in the class~$\mathcal W_{\rm lin}$
(i.e., extensions of $\mathfrak g^\cap_{\rm lin}\dotplus\mathfrak g^{\rm lin}_\theta$)
are exhausted by Cases \ref{case5a}--\ref{case7} of Table~\ref{tab:GenWaveEqsExtensions} with $\hat g=\ve''u$ 
and the counterpart of Cases~\ref{case20}, where $g=0$ and 
the extension is spanned by $\tau\p_t+\xi\p_x$ with the same condition on $(\tau,\xi)$. 
Additional equivalence transformations between the classification cases 
are exhausted by the counterparts of those for Cases \ref{case5a}--\ref{case6f}.%
\footnote{%
Therefore, a complete list of $\mathcal G^\sim_{\rm lin}$-inequivalent Lie-symmetry extensions 
within the class~$\mathcal W_{\rm lin}$ are exhausted by the equations from the class~$\mathcal W_{\rm lin'}$ with 
$g^1=0$ (the (1+1)-dimensional wave equation for $\ve=1$ and the two-dimensional Laplace equation for $\ve=-1$), 
$(g^1,\ve)=(-1,1)$ (the (1+1)-dimensional Klein--Gordon equation), 
$(g^1,\ve)=(\ve',-1)$ (the two-dimensional Helmholz equation) 
and $g^1=\nu x^{-2}$ 
(the (1+1)-dimensional wave equation with the potential $\nu x^{-2}$ for $\ve=1$ 
and the two-dimensional Laplace equation with the potential $\nu x^{-2}$ for $\ve=-1$).
}\looseness=-1
\end{remark}

It now remains to study the equivalence groupoid~$\mathcal G^\sim_1$ of the subclass~$\mathcal W_1$,
which is singled out from the class~\eqref{eq:GenWaveEqs} by the constraint $f_u\ne0$.
By~$\mathcal W_{11}$ and~$\mathcal W_{12}$ we respectively denote the subclasses of~$\mathcal W_1$
that is associated with the additional auxiliary condition $f_x=g_x=0 \bmod G^\sim$
and that consists of the equations $G^\sim$-equivalent to equations of the form
\begin{equation}\label{eq:GenWaveEqsSpecial}
u_{tt}=\ve u^{-4}u_{xx}+\mu(x)u^{-3}+\sigma u,
\end{equation}
where $\mu$ runs through the set of smooth functions of~$x$, $\ve$ and $\sigma$ are constants,
$\ve\ne0$ and hence $\ve=\pm1\bmod G^\sim$ and $\sigma\in\{-1,0,1\}\bmod G^\sim$.

\begin{lemma}\label{lem:ClassW1AdmTrans}
The usual equivalence group of the subclass~$\mathcal W_1$ coincides with~$G^\sim$.
Any admissible transformation in $\mathcal W_1\setminus(\mathcal W_{11}\cup\mathcal W_{12})$ is generated by a transformation from~$G^\sim$.
A generating (up to $G^\sim$-equivalence) set~$\mathcal B_1$ of admissible transformations for the class~$\mathcal W_1$, 
which is minimal and self-consistent with respect to $G^\sim$-equivalence, 
is the union of the restriction of the family~\ref{T1} to~$\mathcal W_{11}$ and the family~\ref{T2}, 
which acts within~$\mathcal W_{12}$.
\end{lemma}

\begin{proof}
For $\mathcal T\in \mathcal G^\sim_1$,
the equation~\eqref{eq:DetEqsForAdmTrans1ut0ux0} immediately implies that $T_tX_t=T_xX_x=0$
for admissible transformations within the subclass~$\mathcal W_1$.

Supposing $T_x\ne0$, we obtain that $X_x=0$, $X_t\ne0$ and hence $T_t=0$.
Up to $G^\sim$-equivalence of admissible transformations we can assume that $T=x$ and $X=t$
since under the above restrictions, the $(t,x)$-part of~$\mathcal T$ is represented as the composition 
of the $(t,x)$-parts of $\mathscr D(T)$, a transformation permuting~$t$ and~$x$ and $\mathscr D(X)$.
For $T=x$ and $X=t$, the equations~\eqref{eq:DetEqsForAdmTrans2}, \eqref{eq:DetEqsForAdmTrans4ut1} and~\eqref{eq:DetEqsForAdmTrans4ux1}
reduce to the simple equations $\tilde ff=1$, and $U_{ux}=U_{ut}=0$, i.e.\ $U_u=\const$. 
By a scaling of~$u$, which belongs to~$G^\sim$, the constant~$U_u$ can be set equal to~1.
Differentiating the equation $\tilde f=1/f$ with respect to~$t$ and then, assuming $(\tilde t,\tilde x,\tilde u)$ as basic variables, with respect to~$\tilde t$,
we derive the equation $U^0_{tx}=0$.
Therefore, $U^0=\psi(x)+\tilde\psi(t)$ and the reduced transformation~$\mathcal T$ with $T=x$, $X=t$ and~$U_u=1$
can be represented as the composition of the transformations $\mathscr Z(\psi)$, $t\leftrightarrow x$ and $\mathscr Z(\tilde\psi)$,
where $t\leftrightarrow x$ denotes the transformation which only permutes~$t$ and~$x$: $\tilde t=x$, $\tilde x=t$ and~$\tilde u=u$.
This means that up to $G^\sim$-equivalence of admissible transformations the transformation~$\mathcal T$ coincides with $t\leftrightarrow x$.
The corresponding transformation components for the arbitrary elements~$f$ and~$g$
follow from the equations~\eqref{eq:DetEqsForAdmTrans2} and~\eqref{eq:DetEqsForAdmTrans4ut0ux0}.
They read $\tilde f=1/f$ and~$\tilde g=-g/f$.
Since the left-hand (resp.\ right-hand) sides of these equalities do not depend on $\tilde t=x$ (resp.\ $\tilde x=t$),
the arbitrary elements of equations from the class~\eqref{eq:GenWaveEqs} that are connected by the transformation $t\leftrightarrow x$ satisfy
the additional auxiliary constraints $f_x=g_x=0$ and $\tilde f_{\tilde x}=\tilde g_{\tilde x}=0$.
In other words, admissible transformations
of the case under consideration are generated by transformations from the equivalence group~$G^\sim$ of the entire class~\eqref{eq:GenWaveEqs}
and the equivalence transformation $t\leftrightarrow x$
of the subclass~$\mathcal U$ which is singled out from the class~\eqref{eq:GenWaveEqs}
by the additional auxiliary constraints $f_x=g_x=0$ and $f_u\ne0$.
In particular, the equivalence group of the subclass~$\mathcal U$ consists of
the transformations of the form~\eqref{eq:EquivalenceGroupGenWaveEqs} with $\varphi_{xx}=\psi_x=0$
and the compositions of these transformations with $t\leftrightarrow x$.

Now we consider the case $T_x=0$ for which $T_t\ne0$, $X_t=0$ and $X_x\ne0$.
Then the equations~\eqref{eq:DetEqsForAdmTrans2}--
\eqref{eq:DetEqsForAdmTrans4ux1} reduce to
$\tilde fT_t^{\;2} = fX_x^{\;2}$, $(U_u^{\;2}/T_t)_t=0$, $(U_u^{\;2}/X_x)_x=0$.
From the first equation we can conclude that $\tilde f_{\tilde u}\ne0$ if and only if $f_u\ne0$.
Solving the other two equations with respect to $U_u$, equating the expressions obtained and
separating variables in this equality, we derive that $U^1:=U_u=\varkappa\sqrt{|T_tX_x|}$, where $\varkappa$ is a nonzero constant.
Differentiating the equation $\tilde fT_t^{\;2} = fX_x^{\;2}$ with respect to~$t$ results in the consequence
\begin{equation}\label{eq:DetEqsForAdmTrans5fune0}
\frac{T_{tt}}{T_t}\left((\tilde u-U^0)\tilde f_{\tilde u}+4\tilde f\right)+2U^0_t\tilde f_{\tilde u}=0.
\end{equation}

If $T_{tt}=0$, then $U^1/\sqrt{|X_x|}=\const$, the equation~\eqref{eq:DetEqsForAdmTrans5fune0} implies $U^0_t=0$,
and thus the transformation~$\mathcal T$ belongs to the action groupoid of the equivalence group~$G^\sim$.

Further we assume that  $T_{tt}\ne0$.
By fixing a value of~$t$, we derive from the equation~\eqref{eq:DetEqsForAdmTrans5fune0}
that the arbitrary element~$\tilde f$
is a solution of an ordinary differential equation of the general form
$
(\tilde u+\tilde \beta(\tilde x))\tilde f_{\tilde u}+4\tilde f = 0,
$
where the variable~$\tilde x$ plays the role of a parameter
and $\tilde \beta$ is a smooth function of~$\tilde x$.
This implies that
$\tilde f = \tilde\alpha(\tilde x)(\tilde u+\tilde\beta(\tilde x))^{-4}$
for some smooth function $\tilde\alpha=\tilde\alpha(\tilde x)$.
Combining the equation $\tilde fT_t^{\;2} = fX_x^{\;2}$ with the expression for $\tilde f$ yields
\[
 f=\frac{T_t^2}{X_x^2}\frac{\tilde\alpha (X)}{(U^1u+U^0+\tilde\beta (X))^4}
 =\frac{\alpha(x)}{(u+\beta(x))^4},
\]
where $\alpha(x):=(\varkappa X_x)^{-4}\tilde\alpha(X)$ and $\beta(x) := (\tilde\beta(X)+U^0)/U^1$.
Furthermore, upon using transformations from the equivalence group~$G^\sim$,
we can set $\tilde\beta=\beta=0$, which consequently implies that $U^0=0$.
By means of equivalence transformations, we can also set $\alpha,\tilde\alpha\in\{-1,1\}$
and as the multiplier relating $\alpha$ and $\tilde\alpha$ is strictly positive, we have that $\tilde\alpha=\alpha=:\ve\in\{-1,1\}$.
Then $X_x$ is a constant and we can set $X=x$ and $\varkappa=1$ using a scaling and a translation of~$x$ and a scaling of~$u$, which belong to~$G^\sim$.
Therefore, $U=\omega u$, where $\omega:=\sqrt{|T_t|}$ and hence $\omega_t\ne0$.
After taking into account all the conditions derived, we reduce the equation~\eqref{eq:DetEqsForAdmTrans4ut0ux0} to the form
$\omega^3\tilde g+\omega(\omega^{-1})_{tt}u=g$.
Differentiating the last equation with respect to~$t$ and dividing the result by $\omega^2\omega_t$,
we obtain $\tilde u\tilde g_{\tilde u}+3\tilde g=4\tilde\sigma\tilde u$,
where $\tilde\sigma:=-\big(\omega(\omega^{-1})_{tt}\big)_t/(4\omega^3\omega_t)$ is a constant.
The general solution of the equation for~$\tilde g$ is $\tilde g=\tilde\mu(x)\tilde u^{-3}+\tilde\sigma\tilde u$.
The expression for~$g$ is similar: $g=\mu(x)u^{-3}+\sigma u$, where $\mu=\tilde\mu$, and
$\sigma:=\tilde\sigma\omega^4+\omega(\omega^{-1})_{tt}$ is, like~$\tilde\sigma$, a constant.
We rewrite the relation defining~$\sigma$ as an ordinary differential equation for~$\omega$, $(\omega^{-1})_{tt}=\sigma\omega^{-1}-\tilde\sigma\omega^3$.
Up to scalings from $G^\sim$ there are only three essentially different values of~$\sigma$ (resp.\ $\tilde\sigma$), $\sigma,\tilde\sigma\in\{-1,0,1\}$.
Finally, from the class~\eqref{eq:GenWaveEqs} we single out the subclass of equations of the general form~\eqref{eq:GenWaveEqsSpecial}.

For each pair of values of~$\sigma$, the corresponding equations from the subclass~\eqref{eq:GenWaveEqsSpecial}
with the same value of the parameter function~$\mu$ are related by a point transformation.
This is why within this subclass it suffices to classify admissible transformations with $\tilde\sigma=0$.
We solve the equation $(\omega^{-1})_{tt}=\sigma\omega^{-1}$ with respect to~$\omega$ and then construct~$T$ using the relation
$T_t=\omega^2\bmod G^\sim$. We~find
\[
T=\left\{
\begin{array}{ll}
(a_1t+a_0)/(a_3t+a_2)\quad&\mbox{if}\quad \sigma=0,\\[.5ex]
(a_1{\rm e}^{2t}+a_0)/(a_3{\rm e}^{2t}+a_2)\quad&\mbox{if}\quad \sigma=1,\\[.5ex]
b_1\tan(t+b_0)+b_2\quad&\mbox{if}\quad \sigma=-1,
\end{array}
\right.
\]
where
$a_0$, \dots, $a_3$, are constants with $a_1a_2-a_0a_3\ne0$ 
that are determined up to a common nonvanishing multiplier, 
and $b_0$, $b_1$ and $b_2$ are constants with $b_1\ne0$.

In the case $\sigma=0$, we obtain a subgroup of the complete point symmetry group of the corresponding equation.
This subgroup is obviously isomorphic to ${\rm PGL}(2,\mathbb R)$.
The condition $T_{tt}\ne0$ is equivalent to $a_3\ne0$, 
and we can assume $a_3=1$ due to the indeterminacy up to a constant multiplier.
Then $a_0-a_1a_2\ne0$, and we gauge $a_2$, $a_0$ and $a_1$ to 0, 1 and 0 
using the ${\rm s}$-action of $\mathscr P^t(a_2)$
and the ${\rm t}$-action of $\mathscr P^t(-a_1)\circ\mathscr D^t(c_2^2)\circ\mathscr D^u(c_2)$ 
with $c_2:=(a_0-a_1a_2)^{-1}$.
All the above transformations from the equivalence group~$G^\sim$ induce point symmetries of the equation under consideration.
Therefore, we can assume that $T=t^{-1}\bmod G^\sim$, obtaining the family~\ref{T2a} of admissible transformations. 

In the same way we derive that $T=\frac12{\rm e}^{2t}\bmod G^\sim$ and $T=\tan t\bmod G^\sim$ if $\sigma=1$  and $\sigma=-1$, 
which gives the family~\ref{T2b} and~\ref{T2c} of admissible transformations, respectively.  

We set $\mu_x\ne0$ and $\tilde\mu_{\tilde x}\ne0$ for admissible transformations from the family~\ref{T2}
since similar admissible transformations with $\mu_x=0$
are $G^\sim$-equivalent to admissible transformations from the restriction of the family~\ref{T1} to~$\mathcal W_{11}\cap\mathcal W_{12}$.
The equations of the form~\eqref{eq:GenWaveEqsSpecial} with $\mu_x\ne0$
are not related to those with $\mu_x=0$ by point transformations.

Following the argumentation for the generating set~$\mathcal B_0$ of~$\mathcal G^\sim_0$ 
from the end of the simultaneous proof of Lemmas~\ref{lem:ClassW0AdmTrans} and~\ref{lem:ClassW0SpecialLieSymExts}, 
we can show that the singled out set~$\mathcal B_1$ of admissible transformations of the subclass~$\mathcal W_1$ 
is a minimal self-consistent generating set for~$\mathcal G^\sim_1$ with respect to $G^\sim$-equivalence.
\end{proof}

The equivalence groups of the subclasses~$\mathcal W_0$ and~$\mathcal W_1$ coincide
with the equivalence group~$G^\sim$ of the entire class~$\mathcal W$, 
and $\mathcal G^\sim=\mathcal G^\sim_0\sqcup\mathcal G^\sim_1$.
Therefore, after uniting the generating (up to $G^\sim$-equivalence) sets~$\mathcal B_0$ and~$\mathcal B_1$ 
of $\mathcal G^\sim_0$ and~$\mathcal G^\sim_1$, 
which are minimal and self-consistent with respect to $G^\sim$-equivalence within the corresponding groupoids,
we get the generating (up to $G^\sim$-equivalence) set~$\mathcal B$ of~$\mathcal G^\sim$, 
which is minimal and self-consistent with respect to $G^\sim$-equivalence within~$\mathcal G^\sim$.

\begin{lemma}\label{lem:ClassW1SpecialLieSymExts}
A complete list of $G^\sim$-inequivalent Lie-symmetry extensions for equations of the general form~\eqref{eq:GenWaveEqsSpecial}
is exhausted by the following cases:
\[
\hspace*{-\arraycolsep}%
\begin{array}{lll}
\text{\rm 1a--1c}.& \text{general}\ \mu\colon & \mathfrak g_\theta=\mathfrak g^\cap_\sigma,\\[.5ex]
\text{\rm 2a--2c}.& \mu=\pm1\colon & \mathfrak g_\theta=\mathfrak g^\cap_\sigma +\langle\p_x\rangle,\\[.5ex]
\text{\rm 3a--3c}.& \mu=\nu x^{-2},\ \nu\ne0\colon & \mathfrak g_\theta=\mathfrak g^\cap_\sigma +\langle2x\p_x-u\p_u\rangle,\\[.5ex]
\text{\rm 4a--4c}.& \mu=0\colon & \mathfrak g_\theta=\mathfrak g^\cap_\sigma +\langle\p_x,\,2x\p_x-u\p_u\rangle\\[.5ex]
\end{array}
\]\vspace{-3ex}\par\noindent
with\vspace{-1.5ex}
\[
\hspace*{-\arraycolsep}%
\begin{array}{lll}
\text{\rm a.}& \sigma= 0\colon & \mathfrak g^\cap_0 = \langle \p_t,\,2t\p_t+u\p_u,\,t^2\p_t+tu\p_u\rangle,\\[.5ex]
\text{\rm b.}& \sigma= 1\colon & \mathfrak g^\cap_1 = \langle \p_t,\,{\rm e}^{2t}(\p_t+u\p_u),\,{\rm e}^{-2t}(\p_t-u\p_u)\rangle,\\[.5ex]
\text{\rm c.}& \sigma=-1\colon & \mathfrak g^\cap_{-1} = \langle \p_t,\,\cos(2t)\p_t-\sin(2t)u\p_u,\,\sin(2t)\p_t+\cos(2t)u\p_u\rangle.\\[.5ex]
\end{array}
\]
The cases $\sigma=1$ and $\sigma=-1$ reduce to the case $\sigma=0$
with the same value of the parameter function $\mu=\mu(x)$
by the additional equivalence transformations
$\tilde t=\frac12{\rm e}^{2t}$, $\tilde x=x$, $\tilde u={\rm e}^t u$ and
$\tilde t=\tan t$,        $\tilde x=x$, $\tilde u=u\cos t$, respectively.
\end{lemma}

\begin{proof}
It follows from Lemma~\ref{lem:ClassW1AdmTrans} that it suffices 
to classify only equations of the form~\eqref{eq:GenWaveEqsSpecial} with $\sigma=0$. 
Spitting the system of determining equations 
\eqref{eq:DetEqForLieSymsOfGenWaveEqs1}--\eqref{eq:DetEqForLieSymsOfGenWaveEqs5} 
for a Lie-symmetry vector field~$Q$ of the equation~$\mathcal L_\theta$ 
with $\theta=(f,g)=(\ve u^{-4},\mu(x)u^{-3})$ with respect to~$u$, 
we derive that the components of~$Q$ are of the form 
$\tau=\tau(t)$, $\xi=2c_1x+c_0$, $\eta=(\frac12\tau_t-c_1)u$, 
where $\tau_{ttt}=0$, and $c_0$ and~$c_1$ are constants with $(2c_1x+c_0)\mu_x=-4c_1\mu$. 
The four cases for~$\mu$ from the lemma's statement arise in the course of analysis of the last equation. 
\end{proof}

Lemmas~\ref{lem:ClassW1AdmTrans} and~\ref{lem:ClassW1SpecialLieSymExts} jointly imply 
that there are no more singular Lie-symmetry extensions within the class~$\mathcal W_1$. 

\begin{remark}\label{rem:StructuteOfGroupoidOfW12}
The groupoid of the class of equations of the form~\eqref{eq:GenWaveEqsSpecial} 
with $\ve=\pm1$ and $\sigma\in\{-1,0,1\}$ 
can be represented in the form~\eqref{eq:GroupoidsWithMixedTrans}, 
where the parameter~$\sigma$ plays the role of~$\gamma$, 
$\Phi_0$ is the identity transformation of~$(t,x,u)$, 
and $\Phi_1$ and~$\Phi_{-1}$ are transformational parts 
of~\ref{T2a} and~\ref{T2b}, respectively. 
Since this class is a subclass of~$\mathcal W_{12}$ 
that is obtained  by gauging the arbitrary elements~$\mathcal W_{12}$ 
with equivalence transformations of~$\mathcal W_{12}$, 
then the equivalence groupoid of~$\mathcal W_{12}$ is of similar structure.    
The analogue of the last claim also holds 
for the intermediate class of equations of the form~\eqref{eq:GenWaveEqsSpecial}  
with $\ve\in\mathbb R_{\ne0}$ and $\sigma\in\mathbb R$. 
\end{remark}

\begin{remark}\label{rem:NormalizationOfW120}
The class~$\mathcal W_{120}$ of equations of the form~\eqref{eq:GenWaveEqsSpecial} 
with $\ve=\pm1$ and $\sigma=0$ is normalized. 
Its equivalence group~$G^\sim_{120}$ consists of the transformations 
\[
\tilde t=T(t):=\frac{a_1t+a_0}{a_3t+a_2},\quad 
\tilde x=b_1x+b_0, \quad 
\tilde u=\pm\sqrt{\big|b_1^{-1}T_t\big|}\,u, \quad 
\tilde\ve=\ve,\quad
\tilde\mu=b_1^{-2}\mu,
\]
where $a_0$, \dots, $a_3$ are arbitrary constants with $a_1a_2-a_0a_3\ne0$ 
that are defined up to a common nonzero multiplier, 
and $b_0$ and $b_1$ are arbitrary constants with $b_1\ne0$. 
The normal subgroup of~$G^\sim_{120}$ associated with the kernel point symmetry group 
of equations from~$\mathcal W_{120}$ is singled out from~$G^\sim_{120}$ 
by the constraints $b_0=0$ and $b_1=1$, 
and thus the kernel invariance algebra of equations from~$\mathcal W_{120}$ 
coincides with~$\mathfrak g^\cap_0$. 
This is why the complete group classification of the class~$\mathcal W_{120}$ 
within the framework of the algebraic method 
reduces to the classification of subalgebras of the algebra $\langle\p_x,\,2x\p_x-u\p_u\rangle$, 
which is trivial. 
A complete list of inequivalent subalgebras of this algebra is exhausted by 
$\{0\}$, $\langle\p_x\rangle$, $\langle2x\p_x-u\p_u\rangle$, $\langle\p_x,\,2x\p_x-u\p_u\rangle$,
all of which are appropriate, cf.\ Lemma~\ref{lem:ClassW1SpecialLieSymExts}.
\end{remark}

\section{Classification of appropriate subalgebras}\label{sec:ClassificationSubalgebrasGenWaveEqs}

The equivalence group~$G^\cap$ and the equivalence algebra~$\mathfrak g^\sim$
admit related representations in the form of a semi-direct product and a semi-direct sum,
$G^\sim=\hat G^\cap\rtimes G^\sim_{\rm ess}$
and $\mathfrak g^\sim=\hat{\mathfrak g}^\cap\rsemioplus\mathfrak g^\sim_{\rm ess}$, respectively.
Here
$\hat G^\cap=\{\mathscr P^t(c_0)\mid c_0\in\mathbb R\}$ is the normal subgroup of~$G^\sim$
associated with the kernel group~$G^\cap$ of the class~\eqref{eq:GenWaveEqs},
$G^\sim_{\rm ess}$~is the subgroup of~$G^\sim$ that consists of the transformations
of the form~\eqref{eq:EquivalenceGroupGenWaveEqs} with $c_0=0$
and thus effectively acts on the class~\eqref{eq:GenWaveEqs},
$\hat{\mathfrak g}^\cap=\langle\PP^t\rangle$ is the ideal of~$\mathfrak g^\sim$
corresponding to the kernel algebra~$\mathfrak g^\cap$ and
$\mathfrak g^\sim_{\rm ess}=\langle\DDD^u,\DDD^t,\DDD(\zeta),\ZZ(\chi)\rangle$
is a subalgebra of~$\mathfrak g^\sim$, which is the ``essential'' part of~$\mathfrak g^\sim$
from the point of view of Lie-symmetry extensions within the class~\eqref{eq:GenWaveEqs}.
Denote by $\pi$ the projection from the space with the coordinates $(t,x,u,f,g)$ onto the space with the coordinates $(t,x,u)$,
and by $\pi_*Q$ the pushforward of a projectable vector field~$Q$ in the space with the coordinates $(t,x,u,f,g)$ by~$\pi$.
A subalgebra~$\mathfrak a$ of~$\mathfrak g^\sim$ is called appropriate
if its projection $\pi_*\mathfrak a$ is the maximal Lie invariance algebra~$\mathfrak g_\theta$
of an equation~$\mathcal L_\theta$ from the class~\eqref{eq:GenWaveEqs}.
Any appropriate subalgebra~$\mathfrak a$ of~$\mathfrak g^\sim$ should contain $\hat{\mathfrak g}^\cap$ as an ideal.
Hence it can also be represented in the form of the semi-direct sum $\mathfrak a=\hat{\mathfrak g}^\cap\rsemioplus\mathfrak s$,
where $\mathfrak s$ is a subalgebra of~$\mathfrak g^\sim_{\rm ess}$.
We call a subalgebra $\mathfrak s$ of~$\mathfrak g^\sim_{\rm ess}$ \emph{appropriate}
if $\mathfrak s=\mathfrak g^\sim_{\rm ess}\cap\mathfrak a$
for an appropriate subalgebra~$\mathfrak a$ of~$\mathfrak g^\sim$.
Appropriate subalgebras~$\mathfrak a_1$ and~$\mathfrak a_2$ of~$\mathfrak g^\sim$
are $G^\sim$-equivalent if and only if
the corresponding subalgebras~$\mathfrak s_1$ and~$\mathfrak s_2$
of~$\mathfrak g^\sim_{\rm ess}$ are $G^\sim_{\rm ess}$-equivalent.
As a result, the classification of Lie-symmetry extensions
induced by subalgebras of~$\mathfrak g^\sim$ up to $G^\sim_{\rm ess}$-equivalence reduces to
the classification of appropriate subalgebras of~$\mathfrak g^\sim_{\rm ess}$
up to $G^\sim_{\rm ess}$-equivalence.

For the latter classification,
we need to compute the adjoint action of the group~$G^\sim_{\rm ess}$ on the algebra~$\mathfrak g^\sim_{\rm ess}$.
Since this algebra is infinite-dimensional,
it is convenient to realize this computation
via pushing forward
the vector fields $\DDD^u$, $\DDD^t$, $\DDD(\zeta)$ and $\ZZ(\chi)$, which span~$\mathfrak g^\sim_{\rm ess}$,
by elementary equivalence transformations from~$G^\sim_{\rm ess}$,
i.e., by $\mathscr D^u(c_2)$, $\mathscr D^t(c_1)$, $\mathscr D(\varphi)$ and~$\mathscr Z(\psi)$, cf.\ Section~\ref{sec:EquivGroup}.
In~other words, the usual transformation rule of vector fields under point transformations will be used \cite{bihl2012b,card2011a}.
This yields the following non-identity actions:
\begin{align*}
&\mathscr Z_*(\psi) \DDD^u = \DDD^u- \ZZ(\psi),                                             &&\mathscr D^u_*(c_2) \ZZ(\chi) = c_2\ZZ(\chi),\\
&\mathscr Z_*(\psi) \DDD(\zeta)=\DDD(\zeta) + \ZZ(\zeta\psi_x-\tfrac12\zeta_x\psi), &&\mathscr D_*(\varphi) \ZZ(\chi) = \ZZ\big(|\hat\varphi_x|^{-1/2}\chi(\hat\varphi)\big),\\
&\mathscr D_*(\varphi) \DDD(\zeta) = \DDD\big(\zeta(\hat\varphi)/\hat\varphi_x\big),       &&
\end{align*}
where $\hat\varphi=\hat\varphi(x)$ is the inverse of the function $\varphi$.

All vector fields from~$\pi_*\mathfrak g^\sim_{\rm ess}$ identically satisfy the determining equations
for Lie symmetries of equations from the class~\eqref{eq:GenWaveEqs},
except the equations~\eqref{eq:DetEqForLieSymsOfGenWaveEqs4} and~\eqref{eq:DetEqForLieSymsOfGenWaveEqs5}.
The latter two equations imply restrictions on appropriate subalgebras of~$\mathfrak g^\sim_{\rm ess}$.

\begin{lemma}\label{lem:OnAppropriateSubalgebras1}
$\mathfrak s\cap\langle\DDD^u,\ZZ(\chi)\rangle=\mathfrak s\cap\langle\DDD^t\rangle=\{0\}$
for any appropriate subalgebra~$\mathfrak s$.
\end{lemma}

\begin{proof}
Suppose that an appropriate subalgebra $\mathfrak s$ of~$\mathfrak g^\sim_{\rm ess}$ contains a vector field~$Q=b\DDD^u+\ZZ(\chi)$,
where the constant~$b$ or the function~$\chi=\chi(x)$ does not vanish.
Then $\pi_*Q$ is a Lie-symmetry vector field for an equation~$\mathcal L_\theta$ from the class~\eqref{eq:GenWaveEqs}.
Substituting the components of the vector field~$\pi_*Q$
into the determining equations~\eqref{eq:DetEqForLieSymsOfGenWaveEqs4} and~\eqref{eq:DetEqForLieSymsOfGenWaveEqs5}
implies the following conditions for the arbitrary-element tuple $\theta=(f,g)$:
\[
(bu+\chi)f_u=0,\quad (bu+\chi)g_u = bg-\chi_{xx}f.
\]
Then $f_u=0$ and $g_{uu}=0$ if $b\ne0$ or $\chi\ne0$. This contradicts the definition of the class~\eqref{eq:GenWaveEqs}.

Analogously, the condition~$\pi_*\DDD^t\in\mathfrak g_\theta$ gives the equation $f=0$,
which is also inconsistent with the definition of the class~\eqref{eq:GenWaveEqs}.

Therefore, any appropriate subalgebra contains no vector fields of the forms considered.
\end{proof}

\begin{lemma}\label{lem:OnAppropriateSubalgebras2}
$\dim\big(\mathfrak s\cap\langle\DDD(\zeta),\ZZ(\chi)\rangle\big)\in\{0,1,3\}$ for any appropriate subalgebra~$\mathfrak s$.
\end{lemma}

\begin{proof}
Suppose that~$\mathfrak s$ is an appropriate subalgebra of $\mathfrak g^\sim_{\rm ess}$
and $\dim\big(\mathfrak s\cap\langle\DDD(\zeta),\ZZ(\chi)\rangle\big)\geqslant2$.
This means that the subalgebra~$\mathfrak s$ contains (at least) two vector fields~$Q^i=\DDD(\zeta^i)+\ZZ(\chi^i)$,
where the functions~$\zeta^i$, $i=1,2$, should be linearly independent in view of Lemma~\ref{lem:OnAppropriateSubalgebras1}.
In other words, the projections $\pi_*Q^i$ of $Q^i$ simultaneously
are Lie-symmetry vector fields of an equation from the class~\eqref{eq:GenWaveEqs}.
By~$W$ we denote the Wronskian of the functions $\zeta^1$ and $\zeta^2$, $W=\zeta^1\zeta^2_x-\zeta^2\zeta^1_x$.
$W\ne0$ as the functions~$\zeta^1$ and $\zeta^2$ are linearly independent.

Plugging the coefficients of~$\pi_*Q^i$ into the equation~\eqref{eq:DetEqForLieSymsOfGenWaveEqs4} gives two equations with respect to~$f$ only,
\begin{equation}\label{eq:DetermingEquationsSimplifiedTwice1}
2\zeta^if_x+(\zeta^i_x u+2\chi^i)f_u = 4\zeta^i_xf.
\end{equation}
We multiply the equation~\eqref{eq:DetermingEquationsSimplifiedTwice1} with $i=1$ by $\zeta^2$ 
and subtract it from the equation~\eqref{eq:DetermingEquationsSimplifiedTwice1} with $i=2$ multiplied by $\zeta^1$. 
Dividing the resulting equation by~$W$, we obtain the ordinary differential equation
$(u+\beta)f_u=4f$,
where $\beta=\beta(x):=2(\zeta^1\chi^2_x-\zeta^2\chi^1_x)/W$ and the variable~$x$ plays the role of a parameter.
It is possible to set $\beta=0$ by means of an equivalence transformation, $\mathscr Z(-\beta)$.
Indeed, this transformation preserves the form of the vector fields~$Q^i$, only changing the values of the functional parameters $\chi^i$.
In particular, it does not affect the linear independency of the functions $\zeta^i$.
The integration of the above equation for $\beta=0$ yields that $f=\alpha u^4$, where $\alpha=\alpha(x)$ is a nonvanishing function of~$x$.
In view of the derived form of~$f$, the splitting of equations~\eqref{eq:DetermingEquationsSimplifiedTwice1} with respect to~$u$
leads to $\zeta^i\alpha_x=0$ and $\chi^i\alpha=0$, i.e., $\alpha_x=0$ and $\chi^i=0$.
The constant $\alpha$ can be scaled to $\alpha=\pm1$ by an equivalence transformation.

In a similar manner, consider the equation~\eqref{eq:DetEqForLieSymsOfGenWaveEqs5}, taking into account the restrictions set on parameter functions and the form of~$f$.
For each~$Q^i$, the equation~\eqref{eq:DetEqForLieSymsOfGenWaveEqs5} gives an equation with respect to~$g$,
\begin{equation}\label{eq:DetermingEquationsSimplifiedTwice2}
2\zeta^ig_x+\zeta^i_x ug_u = \zeta^i_x g-\zeta^i_{xxx}\alpha u^5.
\end{equation}
Again, we multiply the equation~\eqref{eq:DetermingEquationsSimplifiedTwice2} with $i=1$ by $\zeta^2$ 
and subtract it from the equation~\eqref{eq:DetermingEquationsSimplifiedTwice2} with $i=2$ multiplied by $\zeta^1$, 
divide the resulting equation by~$W$ and thereby obtain that
$ug_u = g + \mu^0 u^5$,
where $\mu^0=\mu^0(x):=-\alpha(\zeta^1\zeta^2_{xxx}-\zeta^2\zeta^1_{xxx})/W$  and the variable~$x$ again plays the role of a parameter. 
Integrating the last equation for~$g$ directly gives $g=\mu^0u^5/4+\mu^1u$, where $\mu^1=\mu^1(x)$ is a smooth function of~$x$. 
The parameter function~$\mu^1$ can be set equal to zero by the equivalence transformation~$\mathscr D(\varphi)$,
where the function $\varphi=\varphi(x)$ is a solution 
of the equation $\alpha(2\varphi_{xxx}\varphi_x-\varphi_{xx}^{\,\,\,\,2}) +\mu^1\varphi_x^{\,\,2}=0$.
Substituting the derived form of~$g$ into the equations~\eqref{eq:DetermingEquationsSimplifiedTwice2} and splitting with respect to~$u$, 
we find that~$\mu^1_x=0$, $\zeta^i_{xxx}=0$.

Summing up, we have proved that any equation of the class~\eqref{eq:GenWaveEqs}
admitting (at least) two linearly independent vector fields~$\pi_*Q^i$
in fact possesses exactly three linearly independent vector fields of this form and
is $G^\sim$-equivalent to an equation of the form
$u_{tt} = \pm u^4u_{xx} + \mu^1u$,
where $\mu^1$ is a constant which can be scaled to $\pm1$ if it is not zero.
\end{proof}

The equation $u_{tt} = \pm u^4u_{xx}$, for which $\mu^1=0$, admits an additional Lie-symmetry extension.

\begin{corollary}\label{cor:OnAppropriateSubalgebras1}
There are only two $G^\sim$-inequivalent cases of Lie-symmetry extensions in the class~\eqref{eq:GenWaveEqs}
where the corresponding Lie invariance algebras contain at least two linearly independent vector fields of the form~$\pi_*Q^i$
with $Q^i=\DDD(\zeta^i)+\ZZ(\chi^i)$,
\begin{gather*}
\ref{case14d}.\ u_{tt} = \ve u^4u_{xx} + \ve' u\colon\quad\mathfrak g^{\max} = \mathfrak g^\cap + \pi_*\big\langle\DDD(1),\DDD(x),\DDD(x^2)\big\rangle,\\
\ref{case19d}.\ u_{tt} = \ve u^4u_{xx}\colon\quad\mathfrak g^{\max} = \mathfrak g^\cap + \pi_*\big\langle\DDD(1),\DDD(x),\DDD(x^2),\DDD^u-2\DDD^t\big\rangle 
\end{gather*}
with $\ve,\ve'=\pm1$.
\end{corollary}

Corollary~\ref{cor:OnAppropriateSubalgebras1} gives the classification of appropriate subalgebras of~$\mathfrak g^\sim_{\rm ess}$
the dimensions of whose intersections with $\langle\DDD(\zeta),\ZZ(\chi)\rangle$ are not less than two.
Hence we should continue with the computation of inequivalent appropriate subalgebras of~$\mathfrak g^\sim_{\rm ess}$
that contain at most one linearly independent vector field of the form $\DDD(\zeta)+\ZZ(\chi)$,
where $\zeta=\zeta(x)$ is a nonvanishing function.
In view of Lemma~\ref{lem:OnAppropriateSubalgebras1} it is obvious that the dimension of such subalgebras cannot be greater than three.
Here we select candidates for such subalgebras using only restrictions on appropriate subalgebras presented in Lemma~\ref{lem:OnAppropriateSubalgebras1}.
Since there exist specific restrictions for two- and three-dimensional appropriate subalgebras,
we will make an additional selection of appropriate subalgebras from the set of candidates directly in the course of the construction of invariant equations.

The result of the classification is formulated in the subsequent lemmas.

\begin{lemma}\label{lem:1DimInequivExtsForGenWaveEqs}
A complete list of $G^\sim_{\rm ess}$-inequivalent appropriate one-dimensional subalgebras of~$\mathfrak g^\sim_{\rm ess}$ is given by
\begin{gather}\label{eq:OneDimensionalSubalgebrasGenWaveEqs}
\langle2\DDD^u-q\DDD^t+2\DDD(\delta)\rangle, \quad
\langle\DDD^t-\DDD(2)\rangle, \quad 
\langle\DDD^t-\ZZ(2)\rangle, \quad
\langle\DDD(1)\rangle,
\end{gather}
where $\delta\in\{0,1\}$ and $q$ is an arbitrary constant.
\end{lemma}

\begin{proof}
The classification of the appropriate one-dimensional subalgebras of~$\mathfrak g^\sim_{\rm ess}$
can be carried out effectively by simplifying a general element of~$\mathfrak g^\sim_{\rm ess}$,
\[
Q=a_1\DDD^u+a_2\DDD^t+\DDD(\zeta)+\ZZ(\chi),
\]
using scalings of~$Q$ and pushforwards by elementary transformations from~$G^\sim_{\rm ess}$.
For this aim, it is necessary to distinguish multiple cases,
subject to which of the constants $a_i$ or the functions $\zeta$ and $\chi$ are nonzero.
Note that in the proofs of this and the next two lemmas,
we indicate only the kinds of elementary transformations to be used for simplifying
but not the required values of associated parameters.

For $a_1\ne0$ we can scale the vector field $Q$ to achieve $a_1=2$.
Using $\mathscr Z_*(\psi)$ 
we can set $\chi=0$.
If $\zeta\ne0$, then we set $\zeta=2$ by pushing forward~$Q$ by~$\mathscr D(\varphi)$.
By denoting $a_2=-q$ we obtain the first case from the list~\eqref{eq:OneDimensionalSubalgebrasGenWaveEqs}.

If $a_1=0$ and $a_2\ne0$, we set $a_2=1$ by a scaling of~$Q$.
For $\zeta\ne0$, we can scale $\zeta=-2$ by means of~$\mathscr D_*(\varphi)$
and additionally set~$\chi=0$ upon using the pushforward by~$\mathscr Z(\psi)$. 
If $\zeta=0$, then we have $\chi\ne0$ in view of Lemma~\ref{lem:OnAppropriateSubalgebras1}
and hence we can use $\mathscr D^u_*(c_2)$ and $\mathscr D_*(\varphi)$ in order to set $\chi=-2$.
This gives the second and the third elements of the list~\eqref{eq:OneDimensionalSubalgebrasGenWaveEqs}, respectively.

In case of $a_1=a_2=0$ but $\zeta\ne0$, we can set $\zeta=1$ by~$\mathscr D_*(\varphi)$ 
and use the pushforward~$\mathscr Z_*(\psi)$ 
to arrive at $\chi=0$,
which yields the fourth element of the above list of one-dimensional inequivalent subalgebras.

In view of Lemma~\ref{lem:OnAppropriateSubalgebras1}, the case $a_1=a_2=0$ and $\zeta=0$ is not appropriate.
\end{proof}

\begin{remark}
In Lemma~\ref{lem:1DimInequivExtsForGenWaveEqs} and in the next two lemmas,
we choose such values of parameters in basis elements of appropriate subalgebras
among possible ones up to $G^\sim$-equivalence
that the corresponding equations from the class~$\mathcal W$ have a simple form.
\end{remark}

\begin{lemma}\label{lem:2DimInequivExtsForGenWaveEqs}
Up to $G^\sim_{\rm ess}$-equivalence, any appropriate two-dimensional subalgebra of~$\mathfrak g^\sim_{\rm ess}$
that contains at most one linearly independent vector field of the form $\DDD(\zeta)+\ZZ(\chi)$ belongs to the following list:
\begin{gather}\label{eq:TwoDimensionalSubalgebrasGenWaveEqs}
\begin{split}
&\langle\DDD^u-\DDD(p),\,\DDD^t-\DDD(2)\rangle,\quad
 \langle\DDD^u-2\DDD(x),\,\DDD^t-\ZZ(2)\rangle, \\
&\langle a_1\DDD^u+a_2\DDD^t+a_3\DDD(x)+\ZZ(\delta),\,\DDD(1)\rangle,
\end{split}
\end{gather}
where $p$, $a_1$, $a_2$, $a_3$ and $\delta$ are constants with
$p\ne0$, $(a_1,a_2)\ne(0,0)$, $(a_2,a_3)\ne(0,0)$ and $(a_1,a_3,\delta)\ne(0,0,0)$.
Due to scalings of the first basis element and $G^\sim_{\rm ess}$-equivalence, we can also assume that
one of $a$'s equals 1, $(2a_1+a_3)\delta=0$, and $\delta\in\{0,1\}$.
\end{lemma}

\begin{proof}
Let~$Q^1$ and~$Q^2$ be two arbitrary linearly independent vector fields from~$\mathfrak g^\sim_{\rm ess}$ 
that span a subalgebra~$\mathfrak s$ of~$\mathfrak g^\sim_{\rm ess}$ 
satisfying the inequality $\dim\big(\mathfrak s\cap\langle\DDD(\zeta),\ZZ(\chi)\rangle\big)\leqslant1$
and the conditions from Lemma~\ref{lem:OnAppropriateSubalgebras1}. 
We simplify the basis elements~$Q^1$ and~$Q^2$ as much as possible 
by linear combining and simultaneous pushforwards by transformations from~$G^\sim_{\rm ess}$.
The proof is split into two parts.

First, we consider possible two-dimensional subalgebras of~$\mathfrak g^\sim_{\rm ess}$ 
not containing vector fields of the form $\DDD(\zeta)+\ZZ(\chi)$.
In view of this additional restriction and Lemma~\ref{lem:OnAppropriateSubalgebras1},
basis vector fields of~$\mathfrak s$ can be chosen in the form 
\[
Q^1=\DDD^u+\DDD(\zeta^1)+\ZZ(\chi^1), \quad Q^2=\DDD^t+\DDD(\zeta^2)+\ZZ(\chi^2),
\]
where $\zeta^1\ne0$ and $(\zeta^2,\chi^2)\ne(0,0)$.

If $\zeta^2\ne0$, then we set $\zeta^2=-2$ and $\chi^2=0$ 
successively using $\mathscr D_*(\varphi)$ and $\mathscr Z_*(\psi)$. 
Since the subalgebra~$\mathfrak s$ is closed with respect to the Lie bracket of vector fields, 
i.e., $[Q^1,Q^2]\in\langle Q^1,Q^2\rangle$, 
we derive $[Q^1,Q^2]=2\DDD(\zeta^1_x)+2\ZZ(\chi^1_x)=0$, 
and hence $\zeta^1_x=0$ and $\chi^1_x=0$, 
i.e., $\zeta^1$ and $\chi^1$ are constants. 
We re-denote the nonzero constant~$\zeta^1$ by $-p$. 
The pushforward $\mathscr Z_*(\psi)$ does not change~$Q^2$ and sets $\chi^1=0$, 
which leads to the first family of subalgebras in the list~\eqref{eq:TwoDimensionalSubalgebrasGenWaveEqs}. 

For $\zeta^2=0$, we use $\mathscr D_*(\varphi)$ and $\mathscr Z_*(\psi)$ 
to set $\zeta^1=-2x$ and $\chi^1=0$.
The condition $[Q^1,Q^2]\in\langle Q^1,Q^2\rangle$ implies 
$[Q^1,Q^2]=-2\ZZ(\chi^2_x)=0$. 
Therefore, $\chi^2$ is a nonzero constant, 
which can be gauged by $\mathscr D^u_*(c_2)$ to $-2$, 
giving the second family of subalgebras in the list~\eqref{eq:TwoDimensionalSubalgebrasGenWaveEqs}.

Now we study the case $\dim\big(\mathfrak s\cap\langle\DDD(\zeta),\ZZ(\chi)\big)=1$.
Up to linearly combining the basis elements~$Q^1$ and~$Q^2$, we can initially take
\[
Q^1=a_1\DDD^u+a_2\DDD^t+\DDD(\zeta^1)+\ZZ(\chi^1),\quad
Q^2=\DDD(\zeta^2)+\ZZ(\chi^2),
\]
where $(a_1,a_2)\ne(0,0)$ and $\zeta^2\ne0$.
We set $\zeta^2=1$ and $\chi^2=0$ using $\mathscr D_*(\varphi)$ and $\mathscr Z_*(\psi)$.
Since $\mathfrak s$ is a Lie algebra, we have that $[Q^2,Q^1]=\DDD(\zeta^1_x)+\ZZ(\chi^1_x)=a_3Q_2$ for some constant~$a_3$.
Therefore, $\zeta^1_x=a_3$ and $\chi^1_x=0$.
Combing $Q^1$ with $Q^2$, we obtain that $\zeta^1=a_3x$ and $\chi^1=c=\const$.
Up to $G^\sim_{\rm ess}$-equivalence we can assume that $(2a_1+a_3)c=0$.
Indeed, acting by $\mathscr Z_*(2c/(2a_1+a_3))$ in the case $2a_1+a_3\ne0$, we set $c=0$ in~$Q^1$ and do not change the vector field~$Q^2$.
Using pushforwards by scalings of the variable~$u$ and by alternating its signs,
we can scale the constant parameter $c$ and change its sign.
Additionally we can multiply the whole vector field~$Q^1$ by a nonvanishing constant in order to scale one of nonvanishing $a$'s to one.
The conditions $(a_2,a_3)\ne(0,0)$ and $(a_1,a_3,c)\ne(0,0,0)$ follow from Lemma~\ref{lem:OnAppropriateSubalgebras1}.
After denoting $c$ by $\delta$, this yields the third case of the list~\eqref{eq:TwoDimensionalSubalgebrasGenWaveEqs}
and thereby completes the proof of the lemma.
\end{proof}

\begin{lemma}\label{lem:3DimInequivExtsForGenWaveEqs}
Up to $G^\sim_{\rm ess}$-equivalence, any appropriate three-dimensional subalgebra of~$\mathfrak g^\sim_{\rm ess}$
that contains at most one linearly independent vector field of the form $\DDD(\zeta)+\ZZ(\chi)$ has one of the forms
\begin{equation}\label{eq:3Dextensions}
\langle\DDD^u+p_1\DDD(x),\,\DDD^t+p_2\DDD(x),\,\DDD(1)\rangle, \quad
\langle\DDD^u-2\DDD(x)+\ZZ(d),\,\DDD^t-\ZZ(2),\,\DDD(1)\rangle,
\end{equation}
where $p_1$, $p_2$ and~$d$ are constants such that $p_1p_2\ne0$.
\end{lemma}

\begin{proof}
In view of Lemma~\ref{lem:OnAppropriateSubalgebras1},
any appropriate three-dimensional subalgebra of~$\mathfrak g^\sim_{\rm ess}$,
which contains at most one linearly independent vector field of the form $\DDD(\zeta)+\ZZ(\chi)$, is spanned by vector fields
$Q^1=\DDD^u+\DDD(\zeta^1)+\ZZ(\chi^1)$,
$Q^2=\DDD^t+\DDD(\zeta^2)+\ZZ(\chi^2)$ and
$Q^3=       \DDD(\zeta^3)+\ZZ(\chi^3)$,
where~$\zeta^i$ and~$\chi^i$ are smooth functions of~$x$, 
$\zeta^1$ and~$\zeta^3$ are linearly independent, 
and $(\zeta^1,\chi^1)$ and $(\zeta^2,\chi^2)$ are linearly independent as well.
We also have $[Q^i,Q^j]\in\langle Q^1,Q^2,Q^3\rangle$, $i,j=1,2,3$.

Using $\mathscr D_*(\varphi)$ and $\mathscr Z_*(\psi)$ with suitably chosen functions $\varphi$ and $\psi$ of~$x$,
we set $\zeta^3=1$ and $\chi^3=0$, i.e., we make $Q^3=\DDD(1)$.
The commutation relations of~$Q^3$ with~$Q^1$ and~$Q^2$ are
\begin{gather*}
[Q^3,Q^1]=\DDD(\zeta^1_x)+\ZZ(\chi^1_x)=p_1Q^3,\\
[Q^3,Q^2]=\DDD(\zeta^2_x)+\ZZ(\chi^2_x)=p_2Q^3
\end{gather*}
for some constants $p_i$, $i=1,2$.
These commutation relations imply the conditions $\zeta^i_x=p_i$ and  $\chi^i_x=0$.
Therefore, up to combining $Q^i$ with~$Q^3$ we obtain $\zeta^i=p_ix$ and $\chi^i=d_i$ for some constants $d_i$, 
and $p_1\ne0$ and $(p_2,d_2)\ne(0,0)$ in view of the above linear independence.
Then the commutation relation
\[
[Q^2,Q^1]=\tfrac12\ZZ((p_1+2)d_2-p_2d_1)=0
\]
yields $p_2d_1=(p_1+2)d_2$.
If $p_1\ne-2$, we can set $d_1=0$ using $\mathscr Z_*(2d_1/(p_1+2))$ and then $d_2=0$.
Analogously, in the case $p_2\ne0$ we can set $d_2=0$ using $\mathscr Z_*(d_2/p_2)$ and then $d_1=0$.
Therefore, up to $G^\sim$-equivalence we have two different cases,
$d_1=d_2=0$ and $(p_1,p_2)=(-2,0)$.
In view of Lemma~\ref{lem:OnAppropriateSubalgebras1} we obtain 
$p_1p_2\ne0$ and $d_2\ne0$ in the first and second cases, respectively.
Any nonzero value of~$d_2$ can be gauged by $\mathscr D^u_*(c_2)$ to a fixed nonzero value, e.g., $-2$.
Re-denoting $d_2$ by~$d$ completes the proof of the lemma.
\end{proof}


\section{Regular Lie-symmetry extensions}
\label{sec:GenWaveEqsRegularLieSymmetryExtensions}

For each vector field~$\mathcal Q$ from~$\mathfrak g^\sim$,
the substitution of the components of~$\pi_*\mathcal Q$ 
into the system~\eqref{eq:DetEqForLieSymsOfGenWaveEqs4}--\eqref{eq:DetEqForLieSymsOfGenWaveEqs5}
results in the condition on the arbitrary-element tuple $\theta=(f,g)$
for the equation~$\mathcal L_\theta$ to be invariant with respect to~$\pi_*\mathcal Q$.
This is why equations from the class~\eqref{eq:GenWaveEqs} 
that are invariant with respect to the projection~$\pi_*\mathfrak s$
of an appropriate subalgebra~$\mathfrak s$ of~$\mathfrak g^\sim$ 
can be described by the following way:
For each basis element~$\mathcal Q$ of~$\mathfrak s$, 
we substitute the components of~$\pi_*\mathcal Q$ into the equations~\eqref{eq:DetEqForLieSymsOfGenWaveEqs4} and~\eqref{eq:DetEqForLieSymsOfGenWaveEqs5}.
Collecting all the equations derived from the entire basis~$\mathfrak s$ 
leads to a system of first-order (quasi)linear partial differential equations
in the arbitrary elements~$f$ and~$g$ to be solved.
Simultaneously we check whether the projection~$\pi_*\mathfrak s$ is really the maximal Lie invariance algebra 
of the equation~$\mathcal L_\theta$ for obtained values of the arbitrary-element tuple~$\theta=(f,g)$.

Each of the algebras listed in Lemma~\ref{lem:1DimInequivExtsForGenWaveEqs} 
is really an appropriate one-dimensional subalgebra of~$\mathfrak g^\sim_{\rm ess}$
and results in a simple uncoupled system of two first-order linear differential equations in~$f$ and~$g$.
The corresponding list of equations from the class~\eqref{eq:GenWaveEqs},
which possess one-dimensional Lie-symmetry extensions of~$\mathfrak g^\cap$ related to $\mathfrak g^\sim$, reads
\begin{gather*}\hspace*{-\arraycolsep}
\begin{array}{lll}
\ref{case1}.& 2\DDD^u-q\DDD^t+2\DDD(\delta) \colon &u_{tt} = |u|^q(\hat f(\omega)u_{xx}+\hat g(\omega)u),\\[.5ex]
\ref{case2}.& \DDD^t-\DDD(2)                \colon &u_{tt} = {\rm e}^x(\hat f(u)u_{xx}+\hat g(u))             ,\\[.5ex]
\ref{case3}.& \DDD^t-\ZZ(2)                 \colon &u_{tt} = {\rm e}^u(\hat f(x)u_{xx}+\hat g(x))             ,\\[.5ex]
\ref{case4}.& \DDD(1)                       \colon &u_{tt} = \hat f(u)u_{xx}+\hat g(u)                  ,
\end{array}\hspace*{-5ex}
\end{gather*}
where $\omega:=x-\delta\ln|u|$, $\delta\in\{0,1\}$ and $q$ is an arbitrary constant.
Here and in what follows, in each case we present only vector fields 
that extend the basis $(\PP^t)$ of the ideal~$\hat{\mathfrak g}^\cap$ of~$\mathfrak g^\sim$
into a basis of the corresponding subalgebra of~$\mathfrak g^\sim$.

The computation related to two-dimensional extensions is more complicated.
We first present its result and then give some explanations.
\begin{gather*}\hspace*{-\arraycolsep}
\begin{array}{rl}
\ref{case9}.  & \DDD^u-\DDD(p),\ \DDD^t-\DDD(2),\ p\ne0             \colon \quad u_{tt}=\pm {\rm e}^x|u|^p(u_{xx}+\nu u),      \\[.5ex]
\ref{case10}. & \DDD^u-2\DDD(x),\,\DDD^t-\ZZ(2)                     \colon \quad u_{tt}=\pm x^2{\rm e}^uu_{xx}+\nu {\rm e}^u  ,\\[.5ex]
\ref{case11}. & -\DDD^u+2\DDD^t+2\DDD(x),\ \DDD(1)                  \colon \quad u_{tt}=\hat f(u)u_{xx}           ,            \\[.5ex]
\ref{case12}. & (1-q)\DDD^u+2q\DDD^t-2(1-q)\DDD(x)-\ZZ(4),\ \DDD(1) \colon \quad u_{tt}=\pm {\rm e}^uu_{xx}+\ve' {\rm e}^{qu} ,\\[.5ex]
\ref{case13}. & (3-p+q)D^u+2(1-q)\DDD^t+2(1+p-q)\DDD(x),\ \DDD(1)   \colon \quad u_{tt}=\pm |u|^pu_{xx}+\ve'|u|^q .
\end{array}\hspace*{-5ex}
\end{gather*}
Constraints for constant and functional parameters 
that are imposed by the maximality condition for the corresponding extensions and their inequivalence
are discussed after Theorem~\ref{thm:GenWaveEqsGroupClassification}.

Cases~\ref{case9} and~\ref{case10} are associated with the first and second families of subalgebras
listed in Lemma~\ref{lem:2DimInequivExtsForGenWaveEqs}, respectively.
In both the cases, $\nu$ is an arbitrary constant.
Note that an arbitrary nonzero constant multiplier in the expression for the arbitrary element~$f$, 
which arises in the course of integrating the equation for~$f$, can always be set to~$\pm1$, e.g., by a scaling of~$t$.

The third span from Lemma~\ref{lem:2DimInequivExtsForGenWaveEqs} 
in fact represents a multiparametric series of candidates for appropriate extensions,
which is partitioned in the course of the construction of invariant equations into Cases~\ref{case11}--\ref{case13}.
Not all values of series parameters give appropriate extensions.
Additional constraints for parameters follow 
from the compatibility conditions of the associated system in the arbitrary elements,
\begin{gather*}
\hspace*{-\arraycolsep}\begin{array}{ll}
f_x=0,\quad & ((a_1+\frac12a_3)u+\delta)f_u=pf,\\[.5ex]
g_x=0,\quad & ((a_1+\frac12a_3)u+\delta)g_u=qg,
\end{array}
\end{gather*}
with the inequalities $f\ne0$ and $(f_u,g_{uu})\ne(0,0)$ and the requirement that the dimension of extensions should not exceed two.
Here we introduce the notation $p=2(a_3-a_2)$ and $q=a_1+\frac12a_3-2a_2$.

The above partition is carried out in the following way.
If $a_3=-2a_1$ and $\delta=0$, the inequality $f\ne0$ implies that $p=0$, i.e., $a_2=a_3$.
Since $a_1$, $a_2$ and $a_3$ cannot simultaneously be zero, we obtain that $q\ne0$ and hence $g=0$.
Multiplying the first basis element by $-a_1^{-1}$, we set $a_1=-1$.
This gives Case~\ref{case11}.
For $a_3=-2a_1$ and $\delta=1$ we have $a_2=-q/2$, $a_3=(p-q)/2$ and $a_1=-(p-q)/4$. 
The parameter~$p$ should be nonzero since otherwise we obtain the Liouville equation 
whose maximal Lie invariance algebra is infinite-dimensional. 
We additionally multiply the first basis element by $-4$ 
and scale $p$ with $\mathscr D^u(c_2)$ for some~$c_2$ to~$1$
and obtain Case~\ref{case12}.
Case \ref{case13} corresponds to the condition $a_3\ne-2a_1$.
Scaling the first basis element allows us to set $a_1+\frac12a_3=4$.
Then $a_2=2(1-q)$, $a_3=2(1+p-q)$ and $a_3=(3-p+q)$.
In both Cases~\ref{case12} and~\ref{case13} the parameter $\ve'$ is nonzero 
(otherwise the extension dimension is greater than two)
and can be gauged to $\pm1$ by a simultaneous scaling of~$t$ and~$x$.

Consider the candidates for three-dimensional appropriate extensions listed in Lemma~\ref{lem:3DimInequivExtsForGenWaveEqs}.
The compatibility of the associated systems in the arbitrary elements, supplemented with the inequality $f\ne0$, implies
$p_1=2(p_2-1)$ and $d=-4$ for the first and the second span of Lemma~\ref{lem:3DimInequivExtsForGenWaveEqs}, respectively.
The general solutions of these systems up to $G^\sim$-equivalence are $(f,g)=(\pm |u|^p,0)$ and $(f,g)=(\pm {\rm e}^u,0)$.
This gives the following cases of Lie-symmetry extensions:
\begin{gather*}\hspace*{-\arraycolsep}
\begin{array}{lll}
\ref{case16}. & (p-4)\DDD^u-2p\DDD(x),\ (p-4)\DDD^t-4\DDD(x),\ \DDD(1),\ p\ne0,4\colon & u_{tt}=\pm |u|^pu_{xx},\\[.5ex]
\ref{case17}. & \DDD^u-2\DDD(x)-\ZZ(4),\ \DDD^t-\ZZ(2),\ \DDD(1)\colon & u_{tt}=\pm {\rm e}^uu_{xx}.
\end{array}
\end{gather*}
Here $p:=4(p_2-1)/p_2\ne4$ since for $p=4$ the corresponding equation 
admits the Lie-symmetry vector fields $\pi_*\DDD(x)$ and~$\pi_*\DDD(x^2)$.

Equations from the class~\eqref{eq:GenWaveEqs} which are invariant with respect to two linearly independent vector fields of the form~$\pi_*Q^i$,
where $Q^i=\DDD(\zeta^i)+\ZZ(\chi^i)$, are classified in Corollary~\ref{cor:OnAppropriateSubalgebras1}. 
Therefore, $G^\sim$-inequivalent regular Lie-symmetry extensions in the class~\eqref{eq:GenWaveEqs} 
are exhausted by Cases \ref{case1}--\ref{case4}, \ref{case9}--\ref{case13}, \ref{case14d}, \ref{case16}, \ref{case17} and \ref{case19d}.

\section{Conclusion and discussion}\label{sec:ConclusionGenWaveEqs}

In the present paper, we have carried out 
the complete group classification of the class~$\mathcal W$ 
of (1+1)-dimensional nonlinear wave and elliptic equations of the form~\eqref{eq:GenWaveEqs}
up to both $G^\sim$- and $\mathcal G^\sim$-equivalences
using the new version of the algebraic method of group classification 
for non-normalized classes of differential equations. 
The results of the classification are collected in Theorem~\ref{thm:GenWaveEqsGroupClassification}.
The key ingredient of the classification procedure is the construction 
of a generating set for the equivalence groupoid~$\mathcal G^\sim$ of the class~$\mathcal W$ modulo $G^\sim$-equivalence. 
This 
set is given in Theorem~\ref{thm:ClassWAdmTrans}. 
In view of the partition $\mathcal G^\sim_{\rm gen}=\mathcal G^\sim\sqcup\mathcal G^\sim_{\rm lin}$ 
of the equivalence groupoid~$\mathcal G^\sim_{\rm gen}$ of the superclass~$\mathcal W_{\rm gen}$ 
constituted by all the equations of the form~\eqref{eq:GenWaveEqs} with $f\ne0$,
cf.\ Remark~\ref{rem:OnInequivOfLinAndNonlinCasesOfGenWaveEqs}, 
we can merge the results on~$\mathcal W$ 
with the analogous results from Remark~\ref{rem:GenWaveEqsLinCase} 
on the class~$\mathcal W_{\rm lin}$ of linear equations of the form~\eqref{eq:GenWaveEqs}
to those for~$\mathcal W_{\rm gen}$. 
In other words, we have also obtained 
the complete group classifications of the classes~$\mathcal W_{\rm lin}$ and~$\mathcal W_{\rm gen}$ 
and the classifications of admissible transformations of these classes. 

Below we compare these paper's results with some similar results existing in the literature 
for related classes of differential equations. 
The problem of group classification for the class of semilinear wave equations of the general form 
\begin{gather}\label{eq:LahnoClass}
u_{tt}=u_{xx}+g(t,x,u,u_x)
\end{gather}
was solved in \cite{lahn2005a,lahn2006a}. 
The class~\eqref{eq:LahnoClass} was partitioned into four (normalized) subclasses,
and each of these subclasses was classified separately. 
One of these subclasses, which we denote by~$\mathcal K$, 
is singled out from the class~\eqref{eq:LahnoClass} by the constraint~$g_{u_x}=0$.
The group classification of the subclass~$\mathcal K$ 
was carried out in Section~6 of~\cite{lahn2006a} and the major part of classification results 
was collected in Table~1 therein, see also Section~V and Table~I in~\cite{lahn2005a}. 
Cases~\ref{case1}$_{\delta=1,p=2,\smash{\hat f}={\rm e}^{p\omega}}$, \ref{case2}$_{\smash{\hat f}=1}$, 
\ref{case5a}$_{\ve=1}$, \ref{case6a}$_{\ve=1}$ and~\ref{case18a}$_{\ve=1}$ 
of Table~\ref{tab:GenWaveEqsExtensions} in the present paper 
correspond to Cases~3, 2, 8, 5 and~9 of Table~1 in~\cite{lahn2006a}, 
whereas the Liouville equation is given as Case~\ref{case20} in Table~\ref{tab:GenWaveEqsExtensions} of the present paper 
and as the equation~(5.4) in~\cite{lahn2006a},
and this exhausts all possible analogous cases. 
The counterpart of Case~\ref{case1}$_{\delta=1,p\ne2,\smash{\hat f}={\rm e}^{p\omega}}$ of our Table~\ref{tab:GenWaveEqsExtensions} 
is missed in~\cite{lahn2005a,lahn2006a}.
In fact, each of Cases~3 and~4 of Table~1 in~\cite{lahn2006a} should contain one more constant parameter, 
which cannot be removed by equivalence transformations of the subclass~$\mathcal K$,
and Case~7 therein should be excluded from the classification since it is equivalent to Case~9. 
In~\cite{lahn2007a,lahn2011a}, Lahno and Spichak classified 
the semilinear elliptic equations of the rather general form \[u_{tt}+u_{xx}=F(t,x,u,u_t,u_x)\] 
whose maximal Lie invariance algebras are finite-dimensional. 
Cases~\ref{case6a}, \ref{case7}, \ref{case5a} and~\ref{case18a} of Table~\ref{tab:GenWaveEqsExtensions}
are the restrictions of the first cases of Theorems~3.1 and~3.2 from~\cite{lahn2007a} 
and of the cases ``$A_{3.8}$-Invariant Equations (1)'' and ``$A_{4.10}$-Invariant Equations (3)'' from~\cite{lahn2011a}
to the class~$\mathcal W$, respectively.
There are no other related cases in~\cite{lahn2007a,lahn2011a} and the present paper.  

More important than the solutions of the above specific classification problems 
are the development and modification of general concepts and techniques 
as well as their combinations 
that have been carried out in the course of solving these problems in the present paper.

Partitions of classes of differential equations into subclasses 
that induce partitions of the corresponding equivalence groupoids
had regularly been applied in the course of the study of equivalence groupoids 
\cite{bihl2012b,popo2006b,popo2010a}.   
We have made the two partitions of classes, 
\[
\mathcal W_{\rm gen}=\mathcal W\sqcup\mathcal W_{\rm lin}
\quad\mbox{and}\quad
\mathcal W=\mathcal W_0\sqcup\mathcal W_1, 
\] 
obtaining the partitions of the groupoids 
\[
\mathcal G^\sim_{\rm gen}=\mathcal G^\sim\sqcup\mathcal G^\sim_{\rm lin}
\quad\mbox{and}\quad
\mathcal G^\sim=\mathcal G^\sim_0\sqcup\mathcal G^\sim_1,
\] 
see Remark~\ref{rem:OnInequivOfLinAndNonlinCasesOfGenWaveEqs} and Proposition~\ref{pro:GenWaveEqsGroupoidPartition}. 
All the above classes and subclasses have the same equivalence group~$G^\sim$.
Nevertheless, in contrast to the examples existing in the literature, 
the subclasses in these two partitions do not have better normalization properties than their superclasses.
This is why no kind of normalization can be used for justifying the partitions, 
which are rather derived via the direct analysis of the determining equations for  admissible transformations. 
Although the structure of the partition components is simpler than the entire groupoid for both the groupoid partitions, 
this becomes clear only after a comprehensive study of admissible transformations. 
\looseness=1

We have separately constructed the generating sets~$\mathcal B_0$ and~$\mathcal B_1$ 
of the equivalence groupoids~$\mathcal G^\sim_0$ and~$\mathcal G^\sim_1$, 
which are constituted by the families~\ref{T1}$|_{\mathcal W_0}$ and~\ref{T3}--\ref{T9} 
and by the families~\ref{T1}$|_{\mathcal W_1}$ and~\ref{T2} given in Theorem~\ref{thm:ClassWAdmTrans}, 
respectively. 
Due to constructing the sets~$\mathcal B_0$ and~$\mathcal B_1$ modulo $G^\sim$-equivalence, 
we can and should factor out elements of the action groupoid~$\mathcal G^{G^\sim}\!\!$ 
from admissible transformations before including them in these sets. 
This is realized via successively gauging arbitrary elements of singled out subclasses of equations
that are sources or, equivalently, targets of elements from $\mathcal G^\sim\setminus\mathcal G^{G^\sim}\!$. 
In other words, mapping these subclasses onto smaller ones with simpler equivalence groupoids
by families of equivalence transformations, 
we have factored out subgroups of $G^\sim$ 
and have simplified the consideration for the corresponding classification cases. 

To classify admissible transformations of the class~$\mathcal W_0$ in the optimal way, 
we split the construction of the generating set~$\mathcal B_0$ 
in the simultaneous proof of Lemmas~\ref{lem:ClassW0AdmTrans} and~\ref{lem:ClassW0SpecialLieSymExts}
into two cases depending on the number of independent constraints for arbitrary elements 
that arise in the course of the classification.
In this way, we have extended for the first time the method of furcate splitting 
to the construction of generating sets of admissible transformations.

Moreover, we have found a bijective functor between two categories, 
which are the equivalence groupoids~$\mathcal G^\sim_{00\ve'}$ and~$\mathcal G^\sim_{01g_2}$ 
of the subclasses~$\mathcal W_{00\ve'}$ and~$\mathcal W_{01g_2}$ of equations 
of the forms~\eqref{eq:GenWaveEqsSingularSubclass00} and~\eqref{eq:GenWaveEqsSingularSubclass01} 
with $g^1\ne0$ and a fixed $\ve'\in\{-1,1\}$ 
and with a fixed~$g^2$ satisfying $g^2_ug^2_{uuu}\ne(g^2_{uu})^2$, respectively. 
The isomorphism from~$\mathcal G^\sim_{00\ve'}$ to~$\mathcal G^\sim_{01g_2}$ is given by 
\begin{gather*}
\ve\mapsto\check\ve=\ve,\quad 
g^1\mapsto\check g^1=g^1,\\
\Phi\colon\, \tilde t=T,\ \tilde x=X,\ \tilde u=u-\ln|X_x^{\;2}-\ve X_t^{\;2}|
\quad\mapsto\quad
\check\Phi\colon\, \tilde t=T,\ \tilde x=X,\ \tilde u=u.
\end{gather*}
Fixing~$\ve'$ and~$g^2$ is natural
since values of these parameters cannot be changed by admissible transformations in 
the entire classes~$\mathcal W_{00}$ and~$\mathcal W_{01}$ 
up to gauge equivalence transformations of moving a nonzero constant multiplier 
between $g^1$ and~$g^2$ within~$\mathcal W_{01}$, which can be neglected. 
That is, the partitions of the classes~$\mathcal W_{00}$ and~$\mathcal W_{01}$ 
into the subclasses associated with fixed values of~$\ve'$ and of~$g^2$,
\[
\mathcal W_{00}=\sqcup_{\ve'}\mathcal W_{00\ve'} 
\quad\mbox{and}\quad
\mathcal W_{01}=\sqcup_{g^2}\mathcal W_{01g^2},
\] 
induce the partition of the corresponding equivalence groupoids, 
\[
\mathcal G^\sim_{00}=\sqcup_{\ve'}\,\mathcal G^\sim_{00\ve'}
\quad\mbox{and}\quad
\mathcal G^\sim_{01}=\sqcup_{g^2}\,\mathcal G^\sim_{01g^2}.
\] 
No equations from~$\mathcal W_{00\ve'}$ are related to equations from~$\mathcal W_{01g_2}$ 
by point transformations. 
In other words, the functor from~$\mathcal G^\sim_{00\ve'}$ to~$\mathcal G^\sim_{01g_2}$ 
is not underlaid by a family of point transformations 
generating a mapping from~$\mathcal W_{00\ve'}$ onto~$\mathcal W_{01g_2}$ or conversely. 
Nevertheless, it allows us to easily obtain the equivalence groupoid~$\mathcal G^\sim_{01g_2}$ 
from the equivalence groupoid~$\mathcal G^\sim_{00\ve'}$.

A necessary preliminary step for finding the above functor is 
the proper selection of classes to be related via a functor. 
For the first (degenerate) case in the simultaneous proof 
of Lemmas~\ref{lem:ClassW0AdmTrans} and~\ref{lem:ClassW0SpecialLieSymExts},
under the gauge $f=\ve$ we derive the specific form $g=g^0(x){\rm e}^u+g^1(x)$ 
for values of~$g$ of the source and target equations 
of admissible transformations that are not generated by elements of~$G^\sim$. 
There are two possibilities for a further gauging of parameters in the above form of~$g$, 
either to $g^1=0$ or to ${g^0=\ve'}$. 
The first possibility seems preferable since after gauging 
we obtain equations of the same general form as those in the class~\eqref{eq:GenWaveEqsSingularSubclass01}. 
In this way, the study can be reduced to describing the equivalence groupoid of 
the single class of equations of the form~\eqref{eq:GenWaveEqsSingularSubclass01}, 
where the auxiliary inequality \[{g^2_ug^2_{uuu}\ne(g^2_{uu})^2}\] is neglected.
At the same time, the structure of the subgroupoid of the above groupoid 
that is the equivalence groupoid of the subclass singled out by the constraint $g^2_ug^2_{uuu}=(g^2_{uu})^2$ 
is different from and more complicated than the structure of its complement, 
and thus this subgroupoid needs a separate consideration. 
As a result, the preferable gauge is in fact $g^0=\ve'$. 
Although we then have to study two classes of equations of different forms,
via excluding the evidently marked out value $g^1=0$, 
which corresponds to the Liouville equations giving rise to the family~\ref{T9} of admissible transformations,
and via fixing~$\ve'$ and~$g^2$ 
we have partitioned the corresponding equivalence groupoids 
into naturally isomorphic subgroupoids. 
Therefore, it suffices to describe only one of them. 

It is convenient to construct a generating set for the equivalence groupoid~$\mathcal G^\sim_{00\ve'}$ 
up to the equivalence group of the subclass~$\mathcal W_{00\ve'}$
since then we can apply various algebraic techniques, 
including an original extension of Hydon's algebraic method to admissible transformations.%
\footnote{%
This consideration shows that the algebraic method can further be developed 
to the construction of the complete equivalence groupoids 
for classes of differential equations 
via applying the algebraic method to the corresponding equivalence algebroids,  
which are infinitesimal counterparts of the equivalence groupoids. 
}
These techniques are based on knowing the maximal Lie invariance algebras 
of equations from the subclass~$\mathcal W_{00\ve'}$ 
whose efficient classification involves a preliminary knowledge on admissible transformations
within the subclass~$\mathcal W_{00\ve'}$.  
This is why we have merged the proofs 
of Lemmas~\ref{lem:ClassW0AdmTrans} and~\ref{lem:ClassW0SpecialLieSymExts}.
Mapping  the families \ref{T3'}--\ref{T7'} into the families \ref{T3}--\ref{T7}
and uniting the restrictions of the families \ref{T3}--\ref{T8} to~$\mathcal W_{01}$ 
and to the class of equations of the same form with ${g^2_ug^2_{uuu}=(g^2_{uu})^2}$
provide us with the presentation of the final results in Theorem~\ref{thm:ClassWAdmTrans} 
in a concise~form.  

An unexpected by-product of the proper additional gauging of the arbitrary elements for equations 
from the class~$\mathcal W_0$ with $f=\ve$ and $g=g^0(x){\rm e}^u+g^1(x)$ 
by transformations from the group~$G^\sim$ is 
that this gauging is in accordance with the maximal natural gauging of the arbitrary elements 
within the class~$\mathcal W_{\rm lin}$ by transformations from the same group, 
which leads to the subclass~$\mathcal W_{\rm lin'}$ of~$\mathcal W_{\rm lin}$.
There exists a canonical isomorphism between 
the essential equivalence groupoid~$\mathcal G^{\sim\rm ess}_{\rm lin'}$ of~$\mathcal W_{\rm lin'}$ 
and the equivalence groupoid of the class of equations 
of the form~\eqref{eq:GenWaveEqsSingularSubclass00} with a fixed value of~$\ve'$, 
and it is the above concordance that makes this existence evident.
As a result, the complete group classifications of the class~$\mathcal W_{\rm lin}$ 
up to $G^\sim$- and $\mathcal G^\sim$-equivalences and 
the classification of admissible transformations within this class are carried out 
in the single Remark~\ref{rem:GenWaveEqsLinCase}.
This is one more demonstration of the efficiency of the functor method in 
classification problems of group analysis of differential equations. 
Note that analogously to the previous isomorphism 
between~$\mathcal G^\sim_{00\ve'}$ and~$\mathcal G^\sim_{01g_2}$,
this groupoid isomorphism is not induced by families of admissible point transformations 
within the superclass~$\mathcal W_{\rm gen}$.
 
Necessary preliminaries for the classification of singular Lie-symmetry extensions within the subclass~$\mathcal W_1$ 
have been given by the classification of admissible transformations within this subclass.
As a result, the former classification can easily be completed by either the direct or the algebraic method. 

The classification of regular Lie-symmetry extensions within the class~$\mathcal W$ has been carried out 
within the framework of the algebraic method 
and has reduced to the preliminary group classification of the class~$\mathcal W$. 
We have used our optimized version of this method, 
which involves the classification of candidates for appropriate subalgebras of~$\mathfrak g^\sim$  
by taking into account the principal restrictions on the dimensions and structure of such subalgebras 
and the completion of selecting appropriate subalgebras 
in the course of constructing the corresponding equations possessing Lie-symmetry extensions. 

The results obtained in this paper can be used, in particular, for finding exact solutions 
of equations from the class~\eqref{eq:GenWaveEqs}.

\section*{Acknowledgements}

The authors thank the anonymous reviewers for their valuable comments and suggestions.
The authors are also grateful to Vyacheslav Boyko, Michael Kunzinger, Dmytro Popovych and Galyna Popovych 
for productive and helpful discussions.
OV acknowledges the financial support of her research within the L'Or\'eal-UNESCO \emph{For Women in Science} International Rising Talents Programme.
The research of AB was undertaken, in part, thanks to funding from the Canada Research Chairs program,
the InnovateNL LeverageR{\&}D program and the NSERC Discovery Grant program.
The research of ROP was supported by the Austrian Science Fund (FWF), projects P25064 and P30233.

\footnotesize

\end{document}